\documentclass[11pt,english]{article} 

\usepackage{fullpage}
\usepackage[T1]{fontenc}    
\usepackage[utf8]{inputenc} 
\usepackage{lmodern}        
\usepackage{microtype}      

\usepackage{amsmath, amsthm, amsfonts, amssymb}
\usepackage{bbm}
\usepackage{nicefrac}

\usepackage{graphicx}
\usepackage{subcaption}
\usepackage{wrapfig}
\usepackage{xcolor}
\usepackage{tikz}
\usetikzlibrary{shapes,arrows,positioning,arrows.meta} 
\usepackage[most]{tcolorbox}
\usepackage{varwidth}
\usepackage{booktabs}

\usepackage[nottoc]{tocbibind}

\usepackage{url}
\usepackage[colorlinks=true, allcolors=blue]{hyperref} 

\newcommand{\calT}{\mathcal{T}}

\newcommand{\calP}{\mathcal{P}}
\newcommand{\bY}{\boldsymbol{Y}}
\newcommand{\llhd}{\ell}

\newcommand{\R}{\mathbb{R}}

\DeclareMathOperator{\EX}{\mathbb{E}}

\newtheorem{theorem}{Theorem}
\newtheorem{lemma}{Lemma}
\newtheorem{claim}{Claim}
\newtheorem{corollary}{Corollary}
\newtheorem{definition}{Definition}
\newtheorem{conjecture}{Conjecture}
\newtheorem{remark}{Remark}
\newtheorem{fact}{Fact}
\newtheorem{prop}{Proposition}

\makeatletter
\newtheorem*{rep@theorem}{\rep@title}
\newcommand{\newreptheorem}[2]{%
\newenvironment{rep#1}[1]{%
 \def\rep@title{#2 \ref{##1}}%
 \begin{rep@theorem}}%
 {\end{rep@theorem}}}
\makeatother

\newreptheorem{theorem}{Theorem}
\newreptheorem{lemma}{Lemma}

\newcommand{\bX}{\boldsymbol{X}}
\newcommand{\bZ}{\boldsymbol{Z}}
\newcommand{\ignore}[1]{}

\newcommand{\eps}{\epsilon}
\newcommand{\veps}{\varepsilon}

\newcommand{\proj}{\mathrm{proj}}
\newcommand*{\vertbar}{\rule[-1ex]{0.5pt}{2.5ex}}
\newcommand\ip[2]{\langle #1,#2 \rangle}

\def\final{1}  
\ifnum\final=1
    \newcommand{\nnote}[1]{{\color{red}[{Ning: \bf #1}]\marginpar{\color{red}*}}}
    \newcommand{\enote}[1]{{\color{green}[{\small Elena: \bf #1}]\marginpar{\color{red}*}}}
    \newcommand{\knote}[1]{{\color{purple}[{Karl: \bf #1}]\marginpar{\color{red}*}}}
\else
    \newcommand{\nnote}[1]{}
    \newcommand{\enote}[1]{}
    \newcommand{\knote}[1]{}
\fi

\title{Hardness of Maximum Likelihood Learning of DPPs\thanks{An extended abstract of this work, with
all technical proofs omitted due to space constraint, 
will appear in the Proceedings of the 35th Annual Conference on Learning Theory (COLT 2022).}}
\author{
Elena Grigorescu\thanks{Purdue University, {\tt elena-g@purdue.edu}}
\and Brendan Juba\thanks{Washington University in St.\ Louis, {\tt bjuba@wustl.edu}}
\and Karl Wimmer\thanks{Duquesne University, {\tt wimmerk@duq.edu}}
\and Ning Xie\thanks{Florida International University, {\tt nxie@cis.fiu.edu}}
}

\date{}

\begin{document}

\pagenumbering{roman}

\maketitle
\begin{abstract}
Determinantal Point Processes (DPPs) are a widely used probabilistic model for negatively correlated sets. 
DPPs have been successfully employed in Machine Learning applications to select a diverse, yet representative subset of data. In these applications, the parameters of the DPP need to be fitted to match the data; typically, 
we seek a set of parameters that maximize the likelihood of the data. The algorithms used for this task to date either optimize over a limited family of DPPs, or use local improvement heuristics that do not provide theoretical guarantees of optimality.
   
It is natural to ask if there exist efficient algorithms for finding a maximum likelihood DPP model for a given data set.
In seminal work on DPPs in Machine Learning, Kulesza conjectured in his PhD Thesis (2011) that the problem is NP-complete. 
The lack of a formal proof  prompted Brunel, Moitra, Rigollet and Urschel (COLT 2017) to suggest that, 
in opposition to Kulesza's conjecture, there might exist a polynomial-time algorithm for computing a maximum-likelihood DPP. 
They also presented some preliminary evidence supporting a conjecture that they suggested might lead to such an algorithm.
   
In this work we prove Kulesza's conjecture. In fact, we prove the following stronger hardness of approximation result:
even computing a $\left(1-O(\frac{1}{\log^9{N}})\right)$-approximation to the maximum log-likelihood of a DPP 
on a ground set of $N$ elements is NP-complete.
At the same time, we also obtain the first polynomial-time algorithm that achieves a nontrivial worst-case approximation 
to the optimal log-likelihood: the approximation factor is $\frac{1}{(1+o(1))\log{m}}$ unconditionally 
(for data sets that consist of $m$ subsets), 
and can be improved to $1-\frac{1+o(1)}{\log N}$ if all $N$ elements appear in a $O(1/N)$-fraction of the subsets.
   
From a technical perspective, we reduce the problem of approximating the maximum log-likelihood of a DPP to solving a gap instance of a \textsc{$3$-Coloring} problem on a hypergraph. This hypergraph is based on the bounded-degree construction of Bogdanov, Obata, and Trevisan (2002), which we enhance using the strong expanders of Alon and Capalbo (2007). We demonstrate that if a rank-$3$ DPP achieves near-optimal log-likelihood, its marginal kernel must encode an almost perfect ``vector-coloring" of the hypergraph. Finally, we show that these continuous vectors can be decoded into a proper $3$-coloring after removing a small fraction of ``noisy" edges.
\end{abstract}

\newpage

\tableofcontents

\newpage
\pagenumbering{arabic}
\setcounter{page}{1}
\section{Introduction}

Determinantal Point Processes (DPPs) are a family of probability distributions on sets 
that feature repulsion among elements in the ground set. 
Roughly speaking, a DPP is a distribution over all $2^N$ subsets of $\{1,\ldots, N\}$ defined by 
a positive semidefinite (PSD) $N\times N$ matrix $K$ (called a \emph{marginal kernel}
or \emph{correlation kernel}) whose eigenvalues all lie in $[0, 1]$, 
such that, 
for any $S\subseteq \{1,\ldots, N\}$, random subsets $\bX$ drawn according to the distribution satisfy 
$\Pr[S \subseteq \bX] = \det(K_S)$, where 
$K_S$ is the principal submatrix of $K$ indexed by $S$.

DPPs were first proposed in quantum statistical physics to model fermion systems in
thermal equilibrium~\cite{Mac75},
but they also arise naturally in diverse fields such as random matrix theory, probability theory, number theory, 
random spanning trees and non-intersecting paths~\cite{dyson1962statistical,BP93,RS96,Sos00}.
After the seminal work of~\cite{KT12}, DPPs have attracted a flurry of attention from 
the machine learning community due to their computational tractability and 
excellent capability at producing diverse but representative subsets, 
and subsequently fast algorithms have been developed for 
sampling from DPPs~\cite{hough2006determinantal,kulesza2010structured,rebeschini2015fast,li2016nystrom,li2016fast,
anari2016monte,derezinski2019exact,launay2020exact}.
Furthermore, DPPs have since found a vast variety of applications throughout machine learning and data analysis, 
including text and image search, segmentation and summarization~\cite{lin2012learning,KT12,zou2012priors,
gillenwater2012near,gillenwater2012discovering,yao2016tweet, kulesza2011k,AFAT14, lee2016individualness, 
affandi2013nystrom,chao2015large, affandi2013approximate}, 
signal processing~\cite{xu2016scalable, krause2008near,guestrin2005near}, 
clustering~\cite{zou2012priors,kang2013fast,shah2013determinantal}, 
recommendation systems~\cite{zhou2010solving}, revenue maximization~\cite{dughmi2009revenue}, 
multi-agent reinforcement learning~\cite{osogami2019determinantal,yang2020multi}, 
modeling neural spikes~\cite{snoek2013determinantal}, 
sketching for linear regression~\cite{derezinski2018reverse,derezinski2020exact}, 
low-rank approximation~\cite{guruswami2012optimal}, and likely many more.

\paragraph{Maximum likelihood estimation.} 
Many of these applications require inferring a set of parameters for a DPP capturing a given data set. 
As a DPP specifies a probability distribution, in contrast to supervised learning problems, 
the quality of a DPP cannot be estimated by the ``error rate'' of the model's predictions. 
The standard approach to estimate a DPP from data is to find parameters that maximize the \emph{likelihood} of the 
given data set being produced by a sample from the DPP~\cite{KT12}, i.e., the probability density of the observed data 
in the joint distribution. When the samples are identically and independently chosen from the DPP, the likelihood is 
the product of the probability densities the DPP assigns to the sampled subsets.  The goal of the maximum likelihood 
estimator (MLE) method is to find a kernel matrix that maximizes the likelihood of the data set. 
\cite{BMRU17b} showed that the maximum likelihood estimate indeed converges to the true kernel at a polynomial rate. 
In general, maximizing the likelihood of a DPP gives rise to a non-convex optimization problem, 
and has been approached with heuristics such as expectation maximization~\cite{GKFT14}, 
fixed point algorithms~\cite{MS15}, and MCMC~\cite{AFAT14}. 
In the continuous case, the problem has been studied under strong parametric assumptions~\cite{LMR15}, 
or smoothness assumptions~\cite{Bar13}.

\subsection{Our results}
In spite of the wide applications of DPPs and the central role of the learning step, 
no efficient algorithms are known to find a maximum likelihood DPP. Instead,  
as mentioned above, two families of algorithms are known: 
one seeks to learn an optimal DPP within certain parameterized families of DPP structures~\cite{KT12,AFAT14,GPK16,MS16,
GPK17,UBMR17,DB18}, 
while the other invokes heuristics to maximize the likelihood in an unconstrained search over 
the parameter space~\cite{KT11,GKFT14,AFAT14,MS15}. 
Neither of these approaches provides any guarantees for how close the likelihood of the obtained parameters 
are to the maximum over all DPPs. 

Indeed, Kulesza~\cite{KT11, Kul12} conjectured over a decade ago that 
the problem of finding a  set of parameters is NP-hard, but indicated that he was unable to 
formally establish a reduction: his thesis includes a sketch of a reduction from {\sc Exact-3-cover} to 
a related problem\footnote{Technically, the reduction proposed by Kulesza 
targets a variant of the maximum-likelihood DPP learning problem
in which the instance specifies a set of positive-semidefinite matrices along with the data, 
and the objective is to find a DPP kernel in the cone of the given matrices that maximizes the likelihood.} 
with numerical evidence suggesting its correctness but without a formal proof. 
The subsequent literature adopted this belief, despite the lack of a solid theoretical foundation. 

In this work, we resolve this question by proving Kulesza's conjecture: 
computing maximum likelihood DPP kernels is indeed NP-hard. 
In fact, we establish a stronger, gapped hardness result: even approximating the maximum likelihood is NP-hard.

\begin{theorem}[Informal Statement of the Main Theorem]\label{thm:informal}
There is a ground set of size $N$ such that
it is NP-hard to $\left(1-O(\frac{1}{\log^9{N}})\right)$-approximate
the maximum DPP log-likelihood value of a sample set.
\end{theorem}

\begin{remark}\label{rem:log-likelihood}
Some comments on our (somewhat confusing) convention of approximation factors are in place. 
Since log likelihood functions are always negative real numbers and it is a bit
awkward to work with optimizing negative quantities, 
we therefore add a minus sign at the front of our definition of log likelihood functions. 
Consequently, we \emph{minimize} (negative) log likelihood functions instead of maximizing them.
On the other hand, as our hard learning instances are reduced from \textsc{Max-$3$SAT} and \textsc{$3$-Coloring},
to be consistent with hardness results in the literature on these problems, 
we use $\alpha$-approximation (where $0<\alpha <1$, for maximization problems) 
in the statements of our hardness and algorithmic results.
Note that here ``$\alpha$-approximation'' actually means that the (negative) log likelihood function (in our definition
and ought to be minimized) output by an algorithm
is at most $\frac{1}{\alpha}$ time the optimal log likelihood function.  
\end{remark}

Therefore, the difficulty of learning a DPP is not tied to any particular representation of the marginal kernel, 
as in fact estimating the maximum likelihood \emph{value} itself is NP-hard. 
Note, however, that many problems in learning are hard merely due to the difficulty of 
finding a specific representation~\cite{PV88}, which is not the case for our problem. 

The NP-hardness of maximum likelihood learning naturally raises the question of whether any nontrivial approximation 
is possible. We show that such an approximation is possible: we present a simple, polynomial-time algorithm 
obtaining a $\frac{1}{(1+o(1))\log{m}}$-approximation for a data set with $m$ subsets.

\begin{theorem}[Informal Statement of the Approximation Algorithm]
There is a polynomial-time approximation algorithm achieving the following:
on an input data set consisting of $m$ subsets over a ground set of size $N$,
it returns a kernel that $\frac{1}{(1+o(1))\log{m}}$-approximates the maximum log likelihood. 
Moreover, if every element in the ground set appears in at most $O(1/N)$-fraction of the subsets, 
the kernel achieves a $(1-\frac{1+o(1)}{\log{N}})$-approximation to the maximum log likelihood. 
\end{theorem}

We stress that in contrast to the prior work on learning DPP kernels with guarantees \cite{UBMR17}, 
our algorithm does not rely on the data being produced by a DPP to have a ``cycle basis'' of short cycles 
or large nonzero entries. We obtain an approximation to the optimal likelihood for {\em any} data set. 
Although a $\frac{1}{(1+o(1))\log{m}}$-approximation is weak, when every element appears in relatively few subsets 
(which is common in practice), our algorithm is much better: the actual approximation factor is 
$1-\frac{1}{\log(m/a_{\max})}$, where $a_{\max}/m$ is the fraction of the data subsets containing 
the most frequent element in the ground set. 
Hence, if all elements appear in at most a $\sim 1/N$-fraction of the subsets, 
we obtain a $(1-\frac{1+o(1)}{\log N})$-approximation to the log likelihood. 
Although we don't expect our algorithm to obtain substantially better likelihood than various heuristics 
employed in practice, it may nevertheless serve as a benchmark to estimate how close to optimal these solutions are. 
Moreover, the family of instances constructed in our reduction indeed satisfies this condition; 
therefore, to improve the hardness of approximation bound beyond $1-\frac{1+o(1)}{\log N}$, 
the hard instance of data set must have some elements appearing in $\omega(1/N)$-fraction of the subsets.

\subsection{Our approach and techniques}
We show that it is NP-hard to approximate the optimal DPP likelihood function by reducing from a coloring problem, rather than from  \textsc{Exact-3-Cover}, which was Kulesza's~\cite{Kul12} initial approach.

We begin with the intuition leading to the notion of \emph{vector coloring}, which plays a central role in our proof. Since any marginal kernel $K \in \R^{N \times N}$ is positive semidefinite, it can be factored as $K = Q^{\top}Q$, where $Q \in \R^{k \times N}$ and $k$ is the dimension of the kernel. If we denote the column vectors of $Q$ as $q_1, \ldots, q_N$ (where each $q_i \in \R^k$), these vectors provide an embedding of the elements $\{1, \ldots, N\}$ into $\mathbb{R}^k$ that captures similarities between elements. Intuitively, DPPs prioritize diversity due to the geometric properties of the kernel: as elements in a subset $S$ become more similar, the columns of $K_S$ approach collinearity, causing the determinant --- and thus the selection probability --- to vanish.\footnote{Recall that because $K_S$ is a Gram matrix, $\det(K_S)$ equals the squared volume of the parallelepiped spanned by the embedding vectors of elements in $S$.}

Consider a training set consisting of a collection of subsets of $\{1, \ldots, N\}$, each of a constant size $r$. What characterizes a maximum-likelihood DPP kernel for such data? Ideally, the embedding vectors should encode an $r$-vector-coloring of the elements, in the following sense. Each ``vector-color'' is represented by a unit vector in Euclidean space. To maximize the likelihood function, for every subset $S = \{i_1, \ldots, i_r\}$ in the training set, the corresponding embedding vectors $\{q_{i_1}, \ldots, q_{i_r}\}$ should form a ``rainbow coloring,'' meaning the $r$ vectors are pairwise orthogonal. Conversely, for $r$-subsets not present in the training set, we prefer that as many as possible contain ``repeated colors'' (non-orthogonal vectors).\footnote{Implicitly, we also seek $k=r$ so that no non-degenerate parallelepipeds of dimension higher than $r$ exist; see Conjecture~\ref{conj:3-dim}.} 

The kind of $r$-vector-coloring that corresponds to a maximum-likelihood DPP kernel in this way differs from various related concepts in key ways. One such notion, the \emph{orthonormal representation} used to define Lovász's $\vartheta$ function~\cite{Lov79, GLS81}, assigns unit vectors $v_i \in \R^N$ to vertices such that $v_i^{\top}v_j = 0$ strictly for all edges $(i,j) \in E$. Another, the \emph{vector $k$-coloring} introduced in~\cite{KMS98} for any \emph{real number} $k>1$, assigns unit vectors such that $v_i^{\top}v_j \leq -1/(k-1)$ for all $(i,j) \in E$, with the goal of minimizing $k$. Our problem is distinct: 
\begin{enumerate}
\item The integer $r$ is a fixed parameter analogous to the number of discrete colors. 
\item Rather than minimizing $k$, we maximize ``orthogonality'' (defined precisely in Section~\ref{sec:subsection-soundness-reduction}), measuring how closely the embedding approaches an orthonormal representation. 
\item Our objective function uses an averaged orthogonality over all edges. Our reductions must produce a ``gap'' in this average value between \textsc{YES} and \textsc{NO} instances.
\end{enumerate}
Ultimately, $r$-vector-coloring serves as a conceptual bridge: it allows us to interpret learning a DPP kernel as a continuous, ``fuzzy'' version of a discrete coloring problem. It is a reformulation of the problem, not a target for one of our reductions.

It is natural to attempt a reduction from the NP-complete $r$-Coloring problem to Maximum Likelihood Learning for DPPs by mapping edges to sets that should be rainbow-colored. However, if we treat each graph edge as a set of its two endpoints, the problem is not inherently hard because graph $2$-coloring can be solved by a polynomial-time algorithm. We overcome this by ``lifting'' each edge to a size-$3$ subset (transforming the graph into a $3$-uniform hypergraph) so that we may use \textsc{$3$-Coloring} as our starting point. Our goal is to show a completeness/soundness gap: if a graph $G$ is $3$-colorable, there exists a DPP kernel with high likelihood; if $G$ is not $3$-colorable, the likelihood of \emph{every} DPP kernel is small. 

We address several technical challenges to realize this reduction. First, we characterize the structure of maximum-likelihood kernels. Extending an argument from~\cite{BMRU17}, we prove that the squared norm of each embedding vector must equal the element's empirical frequency. Furthermore, we show that, for a near-optimal DPP kernel, by projecting all the ``good'' embedding vectors to an appropriate $3$-dimensional subspace and carefully handling the ``bad'' embedding vectors, we construct a rank-$3$ DPP kernel whose distance in log likelihood to the optimum increases only by a polynomial factor. Consequently, our soundness theorem can focus on rank-$3$ kernels, significantly reducing the complexity of the gadget analysis. Second, we require a gapped reduction rather than a simple decision result. In the \textsc{YES} case, the average volume of the $3$-dimensional parallelepipeds spanned by the hyperedge embedding vectors must be large; in the \textsc{NO} case, this average must be small for \emph{any} embedding. This requires the hypergraph to remain ``far from'' $3$-vector-colorable even if a small fraction of edges are removed. While strong hardness results for coloring $3$-colorable graphs exist~\cite{KLS00,GK04,DMR09,BG16}, they typically rely on dense graphs; the requirement on \textsc{NO} case mentioned above, when applied to dense graphs, would make the problem fall into the regime of property testing, which is unfortunately known to be computationally easy~\cite{GGR98}. 

We circumvent this by adapting the sparse graph construction of Bogdanov, Obata, and Trevisan~\cite{BOT02} (referred to as a \emph{BOT graph} henceforth). Originally used for query lower bounds in $3$-colorability testing, we refine this construction by correcting minor errors and enhancing robustness using the strong expanders of Alon and Capalbo \cite{AC07}. These modifications allow us to prove that for some absolute constant $\delta$, we can decode a $3$-coloring satisfying a $(1-\delta)$-fraction of the original edges, provided the DPP log-likelihood is sufficiently close to the theoretical maximum. The reduction sequence is summarized in Fig.~\ref{fig:reductions}.

\begin{figure}
\centering
\tikzstyle{pinstyle} = [pin edge={thin,black}]
\tikzset{
    block/.style={draw, fill=blue!20, rectangle, minimum height=3em, minimum width=6em, align=center},
    input/.style={coordinate},
    output/.style={coordinate}
}
\begin{tikzpicture}[auto, node distance=2cm,>=latex']
    \node [block, pin={[pinstyle]above:\textsc{Max-$3$SAT}}] (3SAT) {$3$-CNF formula};
    \node [block, right of= 3SAT, pin={[pinstyle]above:\textsc{$3$-Coloring}},
            node distance=3.25cm] (BOT) {BOT graph};
    \draw [->] (3SAT) --  (BOT);
	\node [block, right of= BOT, pin={[pinstyle]above:\textsc{$3$-Coloring}}, 
	      node distance=3.25cm] (hypergraph) {BOT hypergraph};
	\draw [->] (BOT) -- (hypergraph);
	\node [block, right of= hypergraph, pin={[pinstyle]above:\textsc{MLE-DPP-kernel}}, 
	      node distance=6.5cm] (DPP) {DPP learner};
	\draw [->] (hypergraph) -- (DPP) 
    node[midway, above=3pt, align=center] {
        \small\textbf{Completeness} \\ 
        \small\textbf{Theorem (Thm~\ref{thm:completeness})}
    } 
    node[midway, below=3pt, align=center] {
        \small\textbf{Soundness} \\ 
        \small\textbf{Theorem (Thm~\ref{thm:soundness})}
    };
\end{tikzpicture}
\caption{High level overview of our reductions. }
\label{fig:reductions}
\end{figure}
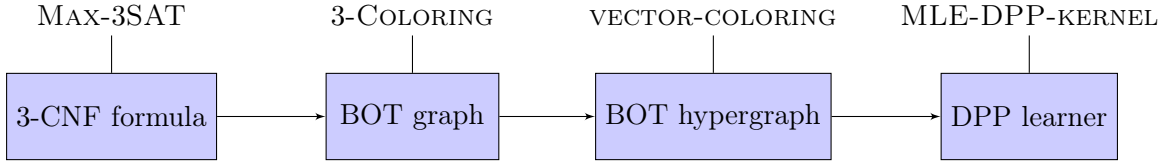

\paragraph{Algorithmic results.} 
For the upper bound, we obtain an approximation algorithm by using some of the properties required
for the analysis of our reduction. 
The algorithm itself is very simple: given a data set $X_1,X_2,\ldots,X_m$, 
output the DPP marginal kernel given by the $N\times N$ diagonal matrix $K$ such that 
$K_{ii} = |\{j : X_j \ni i\}|/m$ for all $i$ in the ground set. 
In other words, the diagonal entries of the marginal kernel are just the empirical probabilities
of elements in the data set.
Hadamard's inequality gives a lower bound on the optimal likelihood that is similar to the likelihood of the 
diagonal kernel; if the elements all appear in at most a $a_{\max}/m$-fraction of the subsets, 
the ratio $\frac{\text{log likelihood output by the algorithm}}{\text{optimal log likelihood}}$ is at most 
$1+\frac{\log((1-\frac{a_{\max}}{m})^{1-a_{\max}/m})}{\log((\frac{a_{\max}}{m})^{a_{\max}/m})}\approx 1+\frac{1}{\log N}$ 
when $a_{\max}/m$ is $O(1/N)$. 
For an unconditional upper bound, observe that elements in all $m$ sets should occur with probability $1$ 
in the maximum likelihood DPP (and thus may be disregarded without loss of generality), 
hence we may plug $a_{\max}=m-1$ into the aforementioned bound and obtain a $(1+o(1))\log{m}$ upper bound on the ratio.

\subsection{Related work}\label{sec:relatedwork}
\paragraph{Learning DPPs.} As mentioned earlier, \cite{UBMR17} in particular proposed an algorithm for recovering a DPP's 
kernel up to similarity, which is efficient when 
(i) the graph defined by interpreting the kernel as a weighted adjacency matrix has a 
``cycle basis'' of cycles of bounded length and 
(ii) the nonzero entries are not too small. 
Furthermore, they gave a lower bound on the sample complexity of estimating the DPP kernel, 
showing that it indeed depends similarly on these quantities. 
Thus, when there is enough data to permit exact recovery of the kernel, 
this algorithm will perform well, but otherwise there is no guarantee 
that the algorithm produces a kernel for a DPP with likelihood close to maximum.

In a companion paper, \cite{BMRU17} 
also studied the rate of estimation obtained by the maximum likelihood kernel. 
Again, they determined classes of DPPs for which it is efficient (or not). 
Moreover, they identified an exponential number of saddle points, 
and conjectured that these are the only critical points; 
they further suggested that a proof of this conjecture might lead to 
an efficient algorithm for computing a maximum likelihood kernel. 
But, they did not actually provide algorithms for computing the likelihood or the kernel itself.

Starting with the pioneering work of~\cite{KT11}, 
various empirical learning algorithms have been proposed for learning discrete DPPs,
such as Bayesian methods~\cite{AFAT14}, expectation-maximization (EM) algorithms~\cite{GKFT14},
fixed-point iteration~\cite{MS15}, learning non-symmetric DPPs~\cite{GBD19}, 
learning with negative sampling~\cite{MGS19},
and minimizing Wasserstein distance~\cite{AGR+20}.
However, none of these algorithms has theoretical guarantees.
Efficient learning algorithms have also been designed for restricted classes of DPPs~\cite{MS16,GPK17,DB18,ORG+18}.
A related problem, namely testing DPPs, recently has been investigated by~\cite{GAJ20}.

We note that in contrast to the problem of learning the DPP kernel from a data set as considered here, 
the problem of computing the mode (``MAP estimate'') of a DPP given by its kernel has long been known to 
be NP-complete~\cite{ko1995exact,civril2009selecting}. 
The inapproximability for this problem was recently strengthened substantially by \cite{Ohs21}.

\paragraph{Vector coloring problems.}
The notion of \emph{orthogonal representation} (in which there is an edge between two vertices if and only if the two corresponding representation vectors are orthogonal\footnote{Some authors, for example~\cite{Lov79}, define orthogonal representation by mandating the vectors of two non-adjacent vertices to be orthogonal.}) was introduced by \cite{Lov79}, and was used in the definition of the famous Lov{\'{a}}sz's $\vartheta$ function, 
which has been employed to bound quantities such as Shannon capacity, the clique numbers or the chromatic numbers of graphs. More generally, a \emph{geometric representation} of a graph is a mapping from vertices of the graph to points in a metric space $\mathcal{M}$, such that two vertices are connected by an edge if and only if the distance between the two corresponding points satisfies a certain condition. For example, orthogonal representation is a special case of the \emph{unit-distance} graph where (in the framework of geometric representation) the underlying metric space is 
the unit sphere (with distance $1$ replaced by angular distance $\pi/2$). Geometric representation of graphs is a well-studied subject, revealing many surprising connections between parameters (e.g. dimension) of geometric representations and properties of the original graph, such as chromatic number, connectivity, excluded subgraphs, tree width, planarity, etc; see e.g.~\cite{Lov79,LSS89,PP89,KMS98,LV99,LSS00,HPS+10} and the recent textbook~\cite{Lov19}.

\paragraph{Matrix completion problem.}
Geometric representations of graphs are intimately connected to another class of problems,
\emph{matrix completion problems}.
For instance, the celebrated result of \cite{Pee96}, showing NP-hardness of deciding whether 
a $3$-dimensional orthogonal representation over a finite field exists for a graph,
was obtained through reducing \textsc{$3$-Colorability} to a low-rank matrix completion problem.
Matrix completion studies under what conditions a partially specified matrix can be completed
into one which belongs to a certain class of matrices, 
such as low-rank matrices, semidefinite matrices, Euclidean distance matrices, etc.
See e.g.~\cite{Lau09} for an overview of the important results in this area.
Interestingly, a recent work of \cite{HMRW14}, which proved the hardness of
low-rank matrix completion problem under the incoherence assumption 
(a commonly used assumption for many matrix completion results), was also based on gapped versions of 
computationally hard problems on graphs such as the \textsc{$r$-Coloring} problem and 
the \textsc{($r,\epsilon$)-Independent-Set} problem.\footnote{In this problem, one is given an undirected graph 
that is promised to be $r$-colorable and is asked to find an independent set of size $\epsilon n$, where $\epsilon<1/r$ and 
$n$ is the number of the vertices in the graph.}


\subsection{Organization of the paper}
The rest of the paper is organized as follows.
In Section~\ref{sec:main_result} we summarize notation, definitions and theorems to be used later,
define formally the problem of maximum likelihood learning of DPPs and state our 
main theorem on hardness of maximum likelihood learning of DPPs.
Then we outline the proof of the main theorem, while deferring some technical proofs to later sections.
Specifically, in Section~\ref{sec:BOT}, 
we provide a detailed description of our slightly modified BOT graph construction
and prove some useful properties of these graphs.
In Section~\ref{sec:diagonal}, we show the following: 
given any training set, at least one of its optimal DPP kernels satisfy that their diagonals are
equal to the empirical frequencies of elements in the ground set.
In Section~\ref{sec:completeness-proof}, we explicitly construct a rank-$3$ optimal kernel
for $3$-colorable BOT graphs.
For general BOT graphs, we further prove in Section~\ref{sec:rank-3} that one may restrict 
to optimizing over rank-$3$ DPP kernels only, without sacrificing too much in likelihood.
Then in Section~\ref{sec:soundness} we put these pieces together and prove the soundness theorem, 
which completes the proof of our main NP-hardness of DPP learning theorem. 
In Section~\ref{sec:algorithm} we present and analyze our simple approximation algorithm. 
Finally, we conclude in Section~\ref{sec:discussion} with discussions and open problems.

\section{Maximum likelihood learning of DPP and our main hardness result}\label{sec:main_result}

\subsection{Preliminaries}
Unless stated otherwise, all logarithms in this paper are to the base $e$ (i.e. natural logarithms). 
For positive integer $n$, we write $[n]$ to denote the set $\{1,2,\ldots, n\}$.
For an $n$-dimensional real vector $x$, we use $\|x\|_{2}=\sqrt{x_1^2+\cdots+x_n^2}$ to denote the
$\ell_2$ or Euclidean norm of $x$. The inner product of two vectors $x, y\in \R^{n}$ is $\ip{x}{y}=x^{\top}y=\sum_{i=1}^{n}x_{i}y_{i}$.
We use $\theta(x, y)$ to denote the angle between two $n$-dimensional real vectors $x$ and $y$. Note that $\theta(x,y)=\cos^{-1}\left(\frac{\ip{x}{y}}{\|x\|_{2} \cdot \|y\|_{2}}\right)$.

\paragraph{Matrix analysis.}
Let $A$ be an $m\times n$ matrix. The $(i, j)^{\text{th}}$ entry of $A$ will be denoted by $A_{i,j}$. For a matrix $A \in \R^{m \times n}$, $S\subset [m]$ and $T \subset [n]$, then $A_{S,T}$ denotes the sub-matrix of $A$ whose rows are indexed by $S$ and columns indexed by $T$. When $S=T$, we simply write $A_{S}$ to denote the principal submatrix of $A$ indexed by $S$.

All matrices in this paper are over real numbers $\R$; therefore the
Hermitian adjoint of $A$, $A^{H}$ is the same as $A^{\top}$, the transpose of $A$.
By the spectral theorem, the eigenvalues
of a real, symmetric matrix $M\in \R^{n\times n}$ are all real numbers, and will
be denoted $\lambda_1(M)\geq \lambda_2(M) \geq \cdots \geq \lambda_n(M)$.
A real, symmetric matrix $M$ is called \emph{positive semidefinite} (PSD) 
if all its eigenvalues are non-negative (i.e., $\lambda_n(M)\geq 0$).
Well-known equivalent characterizations of PSD matrices include $x^{\top} M x\geq 0$ for all $x\in \R^n$,
and the existence of a matrix $Q\in \R^{k\times n}$ for some $k>0$ such that $M=Q^{\top}Q$.

A useful variational characterization of the eigenvalues of real, symmetric matrices is the
Courant-Fischer theorem, which states that for every $1\leq k\leq n$ (when a set of vectors whose 
indices are outside the range $[n]$, then the set is understood to be empty) we have
\[
\lambda_{k}(A)=\min_{x_1, \ldots, x_{k-1}\in \R^n} 
\max_{\substack{y\neq 0, \, y\in \R^n \\ y\perp x_1, \ldots, x_{k-1}}} \frac{y^{\top} A y}{y^{\top} y},
\]
and 
\[
\lambda_{k}(A)=\max_{x_{k+1},\ldots, x_n \in \R^{n}}
\min_{\substack{y\neq 0, \, y\in \R^n \\ y\perp x_{k+1}, \ldots, x_{n}}} \frac{y^{\top} A y}{y^{\top} y}.
\]

The \emph{singular values} of a matrix $A\in \R^{m\times n}$ are defined as the (positive) square roots 
of the eigenvalues of $A^{H}A=A^{\top}A$ (a real, symmetric $n\times n$ matrix). 
Namely, $\sigma_i(A)=\sqrt{\lambda_i(A^{\top}A)}$, $i=1,\ldots, n$.
We also arrange the singular values of a matrix $A$ in decreasing order, that is
$\sigma_1(A)\geq \sigma_2(A)\geq \cdots \geq \sigma_n(A)$.
The \emph{Frobenius norm} of $A$, denoted $\|A\|_{F}$, is defined to be 
$\|A\|_{F}=\sqrt{\sum_{i=1}^{m}\sum_{j=1}^{n}|A_{i,j}|^{2}}$. It is well-known that
$\|A\|_{F}^{2}=\sigma_{1}^{2}(A)+\cdots+\sigma_{n}^{2}(A)$. Finally, the 
\emph{spectral norm} of a square $n\times n$ matrix $A$ is defined as the square root of 
the maximum eigenvalue of $A^{H}A$, i.e.,
\[
\|A\|_{2}=\sqrt{\lambda_1(A^{\top}A)} = \max_{x\neq 0}\frac{\|Ax\|_{2}}{\|x\|_{2}} = \sigma_{1}(A).
\] 

\paragraph{Discrete determinantal point processes.}
A \emph{discrete determinantal point process} (DPP) $\calP$ over a finite set $\mathcal{X}$ 
is a probability measure over the set of all subsets of the ground set $\mathcal{X}$.
The distribution of $\calP$ is specified by a {\em marginal kernel} $K\in \R^{\mathcal{X}\times \mathcal{X}}$, 
which is a positive semidefinite matrix with eigenvalues in $[0,1]$, in the following manner:
if $\bY \subseteq \mathcal{X}$ is a random subset drawn according to $\calP$, 
then its probability mass function $\calP_K$ is defined such that, for every $S \subseteq \mathcal{X}$,
\[
\Pr_{\bY \sim \calP_K}[S \subseteq \bY] = \det(K_S).
\]
Here $K_S$ is the principal submatrix of $K$ indexed by $S\subseteq \mathcal{X}$.


If it is the case that all eigenvalues of $K$ are in $[0,1)$, 
then $\calP$ is called an \emph{$L$-ensemble}, whose kernel can be defined to be the 
positive definite\footnote{A real, symmetric matrix is called \emph{positive definite} (PD) 
if all its eigenvalues are positive.} matrix $L=K(I-K)^{-1}$. 
In this case, the corresponding probability mass function, denoted $\calP_L$,
can be shown to be 
\[
\Pr_{\bY \sim \calP_L}[ \bY=S] = \frac{\det(L_S)}{\det(I+L)},
\]
for every $S \subseteq \mathcal{X}$.
Hence, $\Pr_{\bY \sim \calP_L}[ \bY=\emptyset] = {\det(I-K)}$, 
and consequently a DPP is an $L$-ensemble if and only if the random variable 
$\bY=\emptyset$ with non-zero probability.

\subsection{Maximum Likelihood Learning of DPPs}
We define the \emph{Maximum Likelihood Learning of DPPs} problem as follows.
A learning algorithm receives a training data sample $\{X_t\}_{t=1}^{T}$ (viewed as a multiset) 
drawn independently and identically from a distribution $D$ over the subsets of a ground set $\mathcal{X}$.
The goal of the learning algorithm is to find a DPP kernel $K$ based on the training sample so as to 
minimize\footnote{We add a negative sign at the front to make the estimator to be always positive, 
and thus changing the maximization problem into a minimization one. To be consistent, we still refer to the quantity 
as a ``maximum'' likelihood; see Remark~\ref{rem:log-likelihood}.} the following 
\emph{log likelihood} objective function:
\begin{align}\label{eqn:MLE}
\ell(K)=-\frac{1}{T} \log \prod_{t=1}^{T} \Pr_{\bY \sim \calP_K}[\bY=X_t]
=-\frac{1}{T} \sum_{t=1}^{T} \log  \Pr_{\bY \sim \calP_K}[\bY=X_t],
\end{align}
When the training data sample $\{X_t\}_{t=1}^{T}$ is clear from context, we simply denote the value of the maximum log likelihood of an optimal DPP kernel by $\ell^*$.

One common way to establish the hardness of maximum likelihood learning problems is to show that
even computing the maximum value of the log likelihood $\ell^*$ is hard:
if one could efficiently find an optimal DPP kernel $K$, then since evaluations of determinants can be performed
in polynomial time, clearly the corresponding log likelihood $\ell^*$ would be efficiently computable as well. 
That is also the approach we take in this paper. 
In fact, the lower bound actually proved is much stronger: we show that it is NP-hard even to compute a
$1-O(\frac{1}{\log^9{N}})$-approximation of $\ell^*$ for a ground set of size $N$. 

\begin{theorem}[Main]\label{thm:main}
There are infinitely many positive even integers $N$ such that the following holds.\footnote{As we will see later, 
$N/2$ is the number of edges in a specially constructed graph. We then construct a $3$-uniform hypergraph
based on this graph, and add a new node to the hypergraph for each edge in the graph. 
The ground set consists of these newly added nodes together with the set of vertices in the original graph, 
hence the cardinality of the ground set is at most $N$.} 
Let $\mathcal{X}=\{1, 2, \ldots, N\}$. There is a training data sample $\{X_t\}_{t=1}^{N/2}$ of size $N/2$,
where $X_i \subseteq \mathcal{X}$ for each $i$, such that it is NP-hard to 
$\left(1-O(\frac{1}{\log^9{N}})\right)$-approximate the maximum log likelihood value
that a DPP kernel can achieve on the training set.
\end{theorem}

\subsection{Proof of the Main Theorem: an outline}\label{sec:proof_overview}
As pointed out in Introduction, there are three main reductions in the proof of our main theorem, as illustrated in Fig.~\ref{fig:reductions}. In the following we break this reduction sequence into two parts.
\subsubsection{The first two reductions}
\paragraph{\textsc{Max-$3$SAT} with bounded occurrence.}
Our starting point is the hardness of \textsc{Max-$3$SAT},
in which given a Boolean formula $\phi$ in $3$-CNF form, the goal is to output the maximum number of clauses of $\phi$ 
that can be satisfied by any truth assignment of the variables. 
A classical hardness result is the H{\aa}stad's 3-bit PCP theorem~\cite{Has01}, which states that
it is NP-hard to $(7/8+\eps)$-approximate \textsc{Max-$3$SAT} for any constant $\eps>0$.
However, for our purpose, we need the formula $\phi$ to have bounded occurrence of any variable.
Let \textsc{Max-$3$SAT($k$)} to denote a subclass of \textsc{Max-$3$SAT}, in which the instances
satisfy that every variable occurs in at most $k$ clauses. 
\cite{Has00} showed that it is NP-hard to $7/8+1/(\log{k})^c$-approximate \textsc{Max-$3$SAT($k$)}
where $c$ is some absolute constant.\footnote{We may also use the NP-hardness results of~\cite{BKT03} for
$3$-SAT instance in which every variable appears exactly $4$ times, or
assuming $\mathbf{RP}\neq \mathbf{NP}$ and use the hardness result of~\cite{Tre01} with better parameters.} Therefore,

\begin{lemma}[\cite{Has00}]\label{lem:max-3sat}
There exists an integer $k$ and constant $\eps>0$ (depending only on $k$) such that 
it is NP-hard to $(1-\eps)$-approximate \textsc{Max-$3$SAT($k$)}.
\end{lemma}

That is, for infinitely many integers $n$,  there are two families of instances $\phi_{Y}$ and $\phi_{N}$ in 
\textsc{Max-$3$SAT($k$)} of size $n$ each with the following property:
$\phi_{Y}$ is satisfiable; 
every truth assignment can satisfy at most an  $1-\eps$ fraction of the clauses in $\phi_{N}$; 
and it is NP-hard to distinguish whether one is given a graph from $\phi_{N}$ or from $\phi_{Y}$.

\paragraph{\textsc{$3$-Coloring} for bounded degree graphs.}
Next, we adapt a gap-preserving reduction of Bogdanov, Obata and Trevisan~\cite{BOT02}, which was originally used to
prove an $\Omega(n)$ query lower bound for testing $3$-Colorability in bounded-degree graphs 
under the property testing model. On input an instance $\phi$ of \textsc{Max-$3$SAT($k$)}, 
the reduction outputs a degree-bounded graph $G_{\phi}$ (BOT graph) which satisfies the following:
if $\phi$ is satisfiable then $G_{\phi}$ is $3$-colorable; 
and if every truth assignment can satisfy at most $1-\eps$ fraction of the clauses in $\phi$ then
every $3$-coloring of the vertices of $G_{\phi}$ can make at most $1-\eps'$ fraction of the edges in $G_{\phi}$
non-monochromatic. Here $\eps'$ is a constant depending only on $\eps$ and $k$. 

\begin{lemma}[\cite{BOT02}]\label{lem:BOT-graph}
There are absolute constants $d$ and $\eps'>0$ such that the following holds.
For infinitely many integers $n$, there are two families of degree-$d$ bounded graphs $G_{\phi, Y}$ and $G_{\phi, N}$ of size $n$ 
such that: every $G_{\phi} \in G_{\phi, Y}$ is $3$-colorable; no $1-\eps'$ fraction of the edges of any $G_{\phi} \in G_{\phi, N}$ is $3$-colorable; and yet when given either $G_{\phi} \in G_{\phi, Y}$ or $G_{\phi} \in G_{\phi, N}$, it is NP-hard to distinguish which family $G_\phi$ is from.
\end{lemma}

\paragraph{Enhancing BOT graphs with strong expanders.}
The main idea of the reduction of~\cite{BOT02} is to make $k$ copies of the nodes labeled
$\textsc{True}$, $\textsc{False}$ and $\textsc{Dummy}$ for each variable and its negation,
and use equality gadgets arranged with the topology of a low-degree expander to connect these copies, to ensure the copies are colored identically.
Any constant-degree expander with reasonable vertex-expansion suffices:
on one hand, the resulting graph $G_{\phi}$ is of bounded degree; on the other hand, by the expansion property,
deleting any small fraction of the edges in $G_{\phi}$ will still leave the graph with a large connected component.
Using the coloring of $G_{\phi}$ on this connected component, one can decode an assignment that satisfies most of
the clauses.

However, for our purpose of proving hardness of DPP Maximum Likelihood Learning, 
we need the expander to have some additional properties, which are encapsulated in the following definition.  
 
\begin{definition}[Very strong expanders~\cite{AC07}]
A graph $G=(V,E)$ is called a \emph{$d$-regular very strong expander} on $n$ vertices if
the average degree in any subgraph of $G$ on at most $n/10$ vertices is at most $d/6$, and 
the average degree in any subgraph of $G$ on at most $n/2$ vertices is at most $2d/3$.
\end{definition}

The nice properties of very strong expanders that we require are summarized in the following theorem from~\cite{AC07}.
\begin{theorem}[\cite{AC07}] \label{thm:strong.expander}
Let $G=(V,E)$ be a very strong $d$-regular expander on $n$ vertices.
If we delete an arbitrary subset of $m'\leq nd/150$ edges from $G$ and denote the resulting graph by $G'$,
then $G'$ contains a subgraph $H$ on at least $n-15m'/d$ vertices and the diameter of $H$ is $O(\log{n})$.  
\end{theorem}
The known \emph{explicit} constructions of Ramanujan graphs~\cite{LPS88,Mar88} yield  families of
$d$-regular strong expanders on $n$ vertices for infinitely many $n$'s.

\paragraph{Transforming BOT graphs into $3$-uniform hypergraph.}
To obtain the input data, we transform a BOT graph $G_{\phi}=(V, E)$ into a $3$-uniform hypergraph
$H_{\phi}=(V', E')$ as follows. The vertex set $V'$ is $V(G) \cup E(G)$, and by abuse of notation,
we will simply label the ``graph-vertex'' vertices by $v$ for every $v\in V(G)$, and 
label the ``graph-edge'' vertices in $V'$ by $(u,v)$ for every edge $(u,v)\in E(G)$.
Therefore the set of hyper-edges is 
$
E'=\{(a, v, (u,v)): (u, v)\in E(G)\}.
$
It follows that $H_{\phi}$ is a $3$-uniform hypergraph with\footnote{In order to not overload the statement in Theorem~\ref{thm:main} with multiple parameters, we may add isolated vertices
to graph $G_{\phi}$ so that $n=m$ and hence the ground set size is $N$ and the sample size is $m=N/2$.}
$N=|V'|=n+m$ and $|E'|=m$, where $n$ and $m$ denote the number of vertices and edges of the BOT graph $G_{\phi}$, respectively.

\subsubsection{The final reduction and proof of the soundness theorem}\label{sec:subsection-soundness-reduction}
What if we treat the set of hyperedges of $H_{\phi}$ as the training data $\{X_t\}_{t=1}^{m}$ and provide it as input to a DPP Maximum Likelihood Estimation (MLE) algorithm? Our final reduction aims to establish a gap property:
\begin{description}
\item[Completeness:] If the hypergraph $H_{\phi}$ (or equivalently, BOT graph $G_{\phi}$) is $3$-colorable, then the log-likelihood value of an optimal DPP for the sample set $\{X_t\}_{t=1}^{m}$ is low.
\item[Soundness:] If $H_{\phi}$ (or equivalently, BOT graph $G_{\phi}$) is $\epsilon'$-far from $3$-colorable, then the  log-likelihood value of any DPP for $\{X_t\}_{t=1}^{m}$ is high.
\end{description}
To achieve this, we first characterize the structure of optimal DPP kernels and establish a formal connection between the kernels learned from BOT hypergraphs and the \emph{vector coloring} problem.

\paragraph{DPP kernels and vector coloring.}  
Let $K$ be the marginal kernel of a DPP. Since $K$ is a positive semidefinite matrix, we can write $K$ as $K = Q^{\top}Q$ for some matrix $Q$. Let $q_1,\ldots,q_N \in \R^k$ be the columns of $Q$. We can further decompose these vectors as $q_i = \|q_i\|_{2} \chi_i$, where $\chi_i \in \R^k$ is a unit vector. The quantity $\|q_i\|_2$ is a measure of the ``importance'' of item $i$, and $\chi_i$ is a normalized vector which encodes diversity features of item $i$.  The entries of the marginal kernel satisfy $K_{ij} = \|q_i\|_2 \chi_i^{\top} \chi_j \|q_j\|_2$, 
where $\chi_i^{\top} \chi_j \in [-1,1]$ is a signed measure of the similarity between item $i$ and item $j$. In particular, the diagonal entries satisfy that $K_{ii} = \|q_i\|^{2}_{2}$ for every $i \in [N]$.  

As a first step, we prove the following theorem, which somewhat decouples good values for $\|q_i\|_2$ and $\chi_i$ for each item $i$, and identifies (from the training set) the value of the ``importance'' (i.e. value $\|q_i\|_2$) for each item.\footnote{It is no coincidence that our simple algorithm (see Section~\ref{sec:algorithm} for details), using this ``first-moment'' information from the training set in a similar manner, constructs its DPP that achieves nontrivial worst-case approximation to the optimal log likelihood.} Our result essentially determines at least one of the optimal settings of the importance of each item.  
\begin{theorem}\label{thm:diagonal}
Let $K$ be a marginal kernel with likelihood $\llhd(K)$.  
Then there exists a marginal kernel $K'$ with 
$\llhd(K')\leq \llhd(K)$ such that the diagonal of $K'$ (indexed by vertices and edges of $G_\phi$) satisfies
\[K'_{ii}=
\begin{cases}
\frac{\deg_{G_\phi}(u)}{m} & \mathrm{for \;} i=u \in V(G_\phi); \\
\frac{1}{m} & \mathrm{for \;} i=(u,v) \in E(G_\phi).
\end{cases}
\]
\end{theorem}
Thus, it remains to determine the diversity features that maximize the likelihood. For a training set of BOT hypergraph edges, an optimal DPP marginal kernel ``encodes'' a \emph{vector coloring} of the BOT graph nodes. Intuitively --- and provably so when the marginal kernel has rank $3$ --- this encoding assigns a real vector to each vertex such that adjacent vertices are nearly orthogonal, thereby maximizing the log-likelihood, i.e., minimizing \eqref{eqn:MLE}. 

To formalize this connection, recall that a standard $r$-coloring of a graph $G = (V, E)$ is a function $c : V(G) \to \{1, \dots, r\}$ such that $c(u) \neq c(v)$ for all $(u, v) \in E$. The vector coloring problem extends this to a continuous setting. Since there exist at most $r$ pairwise orthogonal unit vectors in the sphere $S^{r-1} = \{x \in \mathbb{R}^r : \|x\|_2 = 1\}$, the discrete palette of colors is replaced by the unit sphere $S^{r-1}$, and the non-monochromatic constraint is replaced by orthogonality. Specifically, an \emph{$r$-vector-coloring} of a graph $G$ is a function $\chi : V(G) \to S^{r-1}$. We say $\chi$ is \emph{perfect} if $\chi_u^\top \chi_v = 0$ for every $(u, v) \in E$, and a graph is \emph{$r$-vector-colorable} if it admits such a coloring.

When a graph $G$ is not perfectly $r$-vector-colorable, we must quantify its proximity to this state. We define the \emph{orthogonality} of an $r$-vector-coloring $\chi$ as the geometric mean of the squared sines of the angles between adjacent vertex vectors:
\[
\mathrm{orth}(G, \chi)= \left(\prod_{(u, v)\in E}\sin^{2}\theta(\chi(u), \chi(v)) \right)^{1/|E|}.
\]
Here, $0 \leq \mathrm{orth}(G, \chi) \leq 1$, where $\mathrm{orth}(G, \chi)=1$ if and only if the coloring is perfect. This leads us to the following optimization problem:
\begin{tcolorbox}[
    title=$r$-Vector-Coloring Problem,   
    halign title=center,        
    colback=white,              
    colframe=blue!75!black,              
    width=0.7\textwidth,        
    center,                     
    sharp corners,               
	toptitle=2mm,    
    bottomtitle=2mm, 
    top=-2mm,         
    bottom=5mm       
]

\begin{align*}
	\textbf{input:} \quad & \text{graph } G=(V, E) \text{ and integer } r \\
    \textrm{maximize} \quad & \mathrm{orth}(G, \chi) \\
    \textrm{subject to} \quad & \chi: V(G) \to S^{r-1} \\
	                          & (\text{i.e. } \chi(v) \in \mathbb{R}^r \text{ and } \|\chi(v)\|_2 = 1 \text{ for all } v\in V(G))
\end{align*}
\end{tcolorbox}

A key observation is that when the optimal marginal kernel $K$ has rank $3$, the proximity of its log-likelihood $\ell(K)$ to the theoretical optimum, as established by the completeness theorem for $3$-colorable BOT graphs $G_{\phi}$, implies that $K$ encodes an almost perfect $3$-vector-coloring of $G_{\phi}$. However, this relationship requires a formal statement of the completeness theorem to be made precise.
\begin{remark}\label{rem:vector-coloring}
Note that the vector coloring problem defined here is not part of the hardness of approximation reduction sequence. Rather, it serves as a bridge between a marginal kernel with a near-optimal log-likelihood and a proper $3$-coloring of the training BOT graph.    
\end{remark}

To facilitate the statement of the completeness theorem, we introduce the following definition.
\begin{definition}\label{def:ell-yes}
Let $G_{\phi}$ be a BOT graph constructed from a $3$-CNF formula $\phi$. Let $n=|V(G_{\phi})|$ be the number of vertices and $m=|E(G_{\phi})|$ be the number of edges of $G_{\phi}$, respectively. Define the following quantity
\[
\ell_{\textsc{yes}}(G_{\phi})=3\log{m}-\frac{1}{m}\sum_{(u,v)\in E(G_{\phi})}\left(\log(\deg_{G_{\phi}}(u))+\log(\deg_{G_{\phi}}(v))\right).
\]
as a function of the graph $G_{\phi}$. 
\end{definition}
Note that $\ell_{\textsc{yes}}(G_{\phi})$ is defined for any BOT graph $G_{\phi}$, regardless of whether it is $3$-colorable or not. Moreover, $\ell_{\textsc{yes}}(G_{\phi})$ can be computed in linear time by scanning the adjacency list of $G_{\phi}$.

Now, by combining Theorem~\ref{thm:diagonal} with the fact that any $3$-coloring of $G$ naturally induces
a $3$-vector-coloring of $G$, it is not hard to prove the following ``completeness'' theorem.
\begin{theorem}[Completeness theorem]\label{thm:completeness}
Let $G_{\phi}$ be a BOT graph. Then for any DPP marginal kernel $K$, $\ell(K) \geq \ell_{\textsc{yes}}(G_{\phi})$.
If $G_{\phi}$ is $3$-colorable, then there exists a rank-$3$ DPP marginal kernel $K$ such that
\[
\ell(K)=\ell_{\textsc{yes}}(G_{\phi}).
\]
\end{theorem}
That is, $\ell_{\textsc{yes}}(G_{\phi})$ provides an optimal value benchmark of the maximum log likelihood of any DPP kernel for a BOT graph $G_{\phi}$.

\paragraph{Projecting DPP kernels to $\R^3$.}  
Intuition suggests that the maximum likelihood marginal kernel has dimension $3$ so that
zero probability measure will be assigned for subsets of size at least $4$.
We were unable to prove this, 
but we nevertheless manage to show that the loss in making such an assumption is not too great:

\begin{theorem}\label{thm:dim3}
Let $G_{\phi}$ be a BOT graph with maximum degree at most $k$ and let $\ell_{\textsc{yes}}(G_{\phi})$ be as defined in Definition~\ref{def:ell-yes}. Then there is a constant $C_k$ depending only on $k$ such that the following holds. Let $K$ be an optimal marginal kernel for $G_{\phi}$ with likelihood $\llhd(K) \leq \ell_{\textsc{yes}}(G_{\phi}) + \delta$ for some $0 < \delta \leq 1/(128k)^2$, then there exists a marginal kernel $K'$ for $G_{\phi}$ of dimension $3$ such that 
$\ell(K') \leq \ell_{\textsc{yes}}(G_{\phi})+C_{k}\delta^{1/4}$.
\end{theorem}

We conjecture that an even stronger statement is in fact true.
\begin{conjecture}[Cardinality-rank conjecture]\label{conj:3-dim}
Let $k\geq 1$ be an integer. If the cardinality of every subset in a training set is at most $k$, then every optimal maximum likelihood marginal kernel for the training set has dimension at most $k$.
\end{conjecture}
This conjecture may be of independent interest outside the realm of maximum likelihood learning of DPPs.

\paragraph{Decoding a 3-coloring from an almost vector coloring.}
Because of Theorem~\ref{thm:dim3}, from now on we assume that the dimension of $Q$ is $3$. Therefore, for each $(u, v, (u,v)) \in E'$,   
\[
\Pr_{\bY \sim \calP_K}[\bY = \{u, v, (u,v)\}] = \det(K_{\{u, v, (u,v)\}}),
\]
where $K = Q^{\top} Q$ is the marginal kernel of a DPP $\calP$. To maximize the likelihood, we want to maximize the product of determinants of the above form. Since each $(u,v)$ occurs in only one example set, we can assume that $\chi_{(u,v)}$ is always taken to be orthogonal to $\chi_{u}$ and $\chi_{v}$. Thus, the likelihood contribution from $(u, v, (u,v))$ is maximized when $\chi_{u}^{\top} \chi_{v} = 0$; or equivalently, when $\chi_{u}$ and $\chi_{v}$ are orthogonal. We formally prove this correspondence between the ``orthogonality'' of the associated  $3$-vector-coloring of $G_{\phi}$ and the likelihood of the corresponding DPP with marginal kernel $K = Q^{\top} Q$, where the column of $Q$ corresponding to vertex $u$ is $\|q_{u}\|_2 \chi_{u}$:
\begin{align*}
\ell(K) 
&= 3\log{m}-\frac{1}{m}\sum_{(u,v)\in E(G_{\phi})}\left(\log(\deg_{G_{\phi}}(u))+\log(\deg_{G_{\phi}}(u))\right)-\log\left(\mathrm{orth}(G_{\phi}, \chi) \right) \\
&= \ell_{\textsc{yes}}(G_{\phi})-\log\left(\mathrm{orth}(G_{\phi}, \chi) \right),
\end{align*}
see \eqref{eqn:ell(K)-orth} in Section~\ref{sec:dim3-overview}. In other words, the orthogonality $\mathrm{orth}(G_{\phi}, \chi)$ of the $3$-vector-coloring characterizes how closely the log-likelihood of the kernel $K$ approaches the optimal value for $G_{\phi}$---a value that is attainable only when $G_{\phi}$ is $3$-colorable.

More importantly, we demonstrate that when the $3$-vector-coloring $\chi$ is nearly perfect (i.e., $\mathrm{orth}(G_{\phi}, \chi) \approx 1$), the continuous vectors $\{\chi(v)\}_{v\in V(G_{\phi})}$, which fuzzily encode a discrete $3$-coloring, can be successfully decoded. Specifically, we show that a proper $3$-coloring of $V(G_{\phi})$ can be recovered, provided a small fraction of ``noisy'' edges is removed.

\begin{theorem}[Soundness theorem]\label{thm:soundness}
Let $\ell_{\textsc{yes}}(G_{\phi})$ be as defined in Definition~\ref{def:ell-yes}, namely, the optimal log-likelihood assuming that $\phi$ were a YES instance (cf.~Theorem~\ref{thm:completeness}). Let $k$ and $\eps'$ be as those defined in Lemma~\ref{lem:BOT-graph} and let $n=|V(G_{\phi})|$ be the number of vertices in the BOT graph. Then there exists a constant $C=C(k,\eps')$ such that the following holds: If there exists a DPP marginal kernel $K$ for $G_{\phi}$ of rank $3$ satisfying 
\[
\ell(K)\leq \ell_{\textsc{yes}}(G_{\phi})+\frac{C}{\log^2{n}},
\]
then $G_{\phi}$ is $\eps'$-close to $3$-colorable. That is, there exists a set $E' \subset E(G_{\phi})$ with $|E'| \leq \eps'|E(G_{\phi})|$ such that $G_{\phi}\setminus E'$ is $3$-colorable. 
\end{theorem}
Figure~\ref{fig:soundness} shows how the main arguments in the proof of the soundness theorem fit together.

To see why Theorem~\ref{thm:soundness} implies our main theorem, Theorem~\ref{thm:main}, consider that we start the reduction with two families of instances $\phi_{Y}$ and $\phi_{N}$ in \textsc{Max-$3$SAT($k$)} which are NP-hard to distinguish. We then construct their corresponding BOT hypergraphs, $H_{\phi_{Y}}$ and $H_{\phi_{N}}$, and use the edge sets of these two hypergraphs as training sets of size $m$ for a DPP maximum likelihood learning algorithm. The log likelihood value of $\phi\in \phi_{Y}$ is $\ell_{\textsc{yes}}(G_{\phi})=\Theta(\log{N})$ by Theorem~\ref{thm:completeness}. Furthermore, both of the two reductions, as well as the value of $\ell_{\textsc{yes}}(G_{\phi})$, can be computed in polynomial time. What is the log likelihood value of $\phi\in \phi_{N}$? Well, by Theorem~\ref{thm:dim3}, if the marginal kernel of $G_\phi$ has log likelihood $\llhd(K) \leq \ell_{\textsc{yes}}(G_{\phi}) + \delta$ for some small enough $\delta > 0$, then there exists a marginal kernel $K'$ of dimension $3$ such that $\llhd(K') \le \ell_{\textsc{yes}}(G_{\phi}) + C_k\delta^{1/4}$. By Theorem~\ref{thm:soundness}, we must have $\delta=\Omega(\frac{1}{\log^2{N}})$ for our $\phi\in \phi_{N}$. That is, the log likelihood value of $\phi$ is $\ell_{\textsc{yes}}(G_{\phi})+\Omega(\frac{1}{\log^8{N}})$.
Now if there were a polynomial-time algorithm $\mathcal{A}$ which approximates the log likelihood value within a factor of $1-1/\Omega(\log^9{N})$, then $\mathcal{A}$ would be able to tell apart $\phi_{Y}$ from $\phi_{N}$, thus solving an NP-complete problem, simply by approximating the maximum log likelihood value on training data obtained from $H_{\phi_{Y}}$ and $H_{\phi_{N}}$, respectively.

\begin{figure}
\centering
\tikzset{
    block/.style={draw, fill=blue!20, rectangle, minimum height=3em, minimum width=6em, align=center},
    input/.style={coordinate},
    output/.style={coordinate}
}
\begin{tikzpicture}[node distance=0.80cm, auto, >=Latex]
    \node [block] (3SAT) {DPP with ``almost'' \\ optimal DPP kernel};
    
    \node [block, right=of 3SAT] (BOT) {``almost'' optimal \\ rank-$3$ kernel};  
    \draw [->] (3SAT) -- (BOT);

    \node [block, right=of BOT] (hypergraph) {``almost'' perfect $3$-vector- \\ coloring of $G'_{\phi}$};
    \draw [->] (BOT) -- (hypergraph);

    \node [block, right=of hypergraph] (DPP) {$3$-coloring \\ of $G''_{\phi}$};
    \draw [->] (hypergraph) -- (DPP);
\end{tikzpicture}
\caption{The overall argumentative structure of the proof of the soundness theorem ($G'_{\phi}$ is an intermediate graph between $G_{\phi}$ and $G''_{\phi}$, and $G''_{\phi}$ is $\eps'$-close to $G_{\phi}$)}
\label{fig:soundness}
\end{figure}

\section{Robustness of very strong expanders -- Proof of Lemma~\ref{lem:BOT-graph}} \label{sec:BOT}

\subsection{BOT graphs}\label{sec:BOT-graph-construction}
We now provide a detailed description and analysis of the construction in~\cite{BOT02}, 
as our hardness proof relies crucially on the properties of BOT graphs.

Recall that in the classic reduction from $3$-SAT to \textsc{$3$-Coloring}, there are three designated nodes
\textsc{True}, \textsc{False} and \textsc{Dummy} which form a triangle so that in any $3$-coloring of the graph, they
must be assigned to different colors. The reduction then connects the nodes in literal gadgets and clause gadgets
to these three nodes so that every literal is connected to the \textsc{Dummy} node, and 
any $3$-coloring of the graph can be decoded into a satisfying assignment simply 
by reading the color of every literal node.
On the other hand, any satisfying assignment of the variables can also be transformed 
into a valid $3$-coloring of the graph 
by assigning ``True'' literal nodes the same color as the \textsc{True} node 
and  assigning ``False'' literal nodes the same color as the \textsc{False} node. 

The basic idea of the construction in~\cite{BOT02} is to make $k$ ``local'' copies for each literal node
as well as for the three designated \textsc{True}, \textsc{False} and \textsc{Dummy} nodes,
so that each local copy of these nodes is used only \emph{once} in the clause gadgets.
The benefit of doing so is that deleting any single node or any single edge will affect (or ``destroy'')
a single clause, hence makes the reduction from $3$-SAT to \textsc{$3$-Coloring} (up to a constant factor) 
distance-preserving.
However, this introduces a new difficulty: we need to make sure that all $k$ copies of literal nodes
(\textsc{True}, \textsc{False} and \textsc{Dummy} nodes) are colored identically.
Such consistency is ensured by embedding these local copies on an expander
so that the graph is robust against deletions of a small number of edges or nodes, 
and at the same time the increase in the maximum degree is bounded.

Specifically, given a Boolean formula $\phi$ in $3$-CNF form with $n$ variables and $m$ clauses, and such that 
each literal appears in at most $k$ clauses, construct an undirected graph $G_{\phi}(V, E)$ as follows:
\begin{enumerate}
    \item \label{step:one} Create $2kn$ literal nodes: for each variable $x_i$, $1\leq i \leq n$, 
	create $k$ copies of it, $x_i^{(j)}$ for $j=1,\ldots, k$, 
	and $k$ copies of its negation, ${\bar x_i}^{(j)}$ for $j=1,\ldots, k$;
    \item For each of the $2kn$ literal nodes $x_i^{(j)}$ (resp. $\bar{x}_i^{(j)}$), create a block 
	 called the \emph{$x_i^{(j)}$-block} (resp. \emph{$\bar{x}_i^{(j)}$-block}) of $3$ color nodes  
	$T_i^{(j)}, F_i^{(j)}, D_i^{(j)}$ (resp. $\bar{T}_i^{(j)}, \bar{F}_i^{(j)}, \bar{D}_i^{(j)}$),
	which serve as local copies of the \textsc{True}, \textsc{False} and \textsc{Dummy} nodes.
	Add a triangle on $(T_i^{(j)}, F_i^{(j)}, D_i^{(j)})$ (resp. $(\bar{T}_i^{(j)}, \bar{F}_i^{(j)}, \bar{D}_i^{(j)})$) and edges between $x_j^{(i)}$ (resp.\ $\bar{x}_i^{(j)}$) and $D_i^{(j)}$ (resp.\ $\bar{D}_i^{(j)}$)
	  for each literal block;
    \item For each literal, add an edge between the literal node and its negation node, 
	so that literals and their negations will always be assigned to different colors. 
	These edges, as well as the $x_i^{(j)}$-block and the $\bar{x}_i^{(j)}$-block,
	are illustrated in Fig.~\ref{fig:literal_block};
	\item For each literal $x_i$ (resp. $\bar{x}_i$), use the \emph{equality gadget} as illustrated in
	Fig.~\ref{fig:equality_gadget} to connect between each pair of its $k$ copies. That is,
	create a ``$k$-clique'' among $\{x_i^{(j)}\}_{j\in [k]}$ (resp. $\{\bar{x}_i^{(j)}\}_{j\in [k]}$) such that
	each ``edge'' in the ``clique'' is actually an independent copy of the equality gadget;
	\item For each clause, add a \emph{clause gadget} to connect $3$ literals in that clause together with their 
	corresponding \textsc{True} nodes (in the blocks that are created in step~\ref{step:one})
	with $6$ auxiliary nodes as shown in Fig~\ref{fig:clause_gadget}. 
	Note that all literal nodes	and \textsc{True} nodes in the clause gadgets are used only once; and
	\item \label{step:six} Build a constant-degree expander graph of size $2kn$, with $x_i^{(j)}$-blocks and $\bar{x}_i^{(j)}$-blocks,
	$1\leq i \leq n$ and $1\leq j \leq k$, as its vertices.
	 Then, for any two blocks that are connected by an edge in the expander, 
	 connect their corresponding \textsc{True} nodes (resp. \textsc{False} nodes and \textsc{Dummy} nodes)
	 with an independent copy of the equality gadget. See Fig.~\ref{fig:BOT_graph} for an illustration.  
    \end{enumerate}

Note that for a bounded occurrence instance $\phi$ of $3$-CNF formula, the number of clauses is at most
$m\leq 2kn/3$. 

\usetikzlibrary{shapes, snakes}

\tikzstyle{r_node}=[draw,circle, minimum size=1cm, violet, fill={rgb:red,10;green,1;blue,1}, text=white,minimum width=25pt]
\tikzstyle{g_node}=[draw,circle, violet, fill={rgb:red,1;green,10;blue,1}, text=violet,minimum width=25pt]
\tikzstyle{b_node}=[draw,circle, violet, fill={rgb:white,5;blue,10}, text=white,minimum width=25pt]
\tikzstyle{w_node}=[draw,circle, violet, fill={white}, text=black,minimum width=25pt]
\tikzstyle{w2_node}=[draw,ellipse, violet, fill={white}, text=black,minimum width=25pt]
\tikzstyle{bl_node}=[draw,circle, violet, fill={black}, text=white,minimum width=25pt]

\tikzstyle{g3_node}=[draw,circle, minimum size=1.5cm, violet, fill={rgb:red,1;green,10;blue,1}, text=violet,minimum width=50pt]
\tikzstyle{b3_node}=[draw,circle, minimum size=1.5cm, violet, fill={rgb:white,5;blue,10}, text=white,minimum width=50pt]
\tikzstyle{w3_node}=[draw,circle, minimum size=1.5cm, violet, fill={white}, text=black,minimum width=50pt]
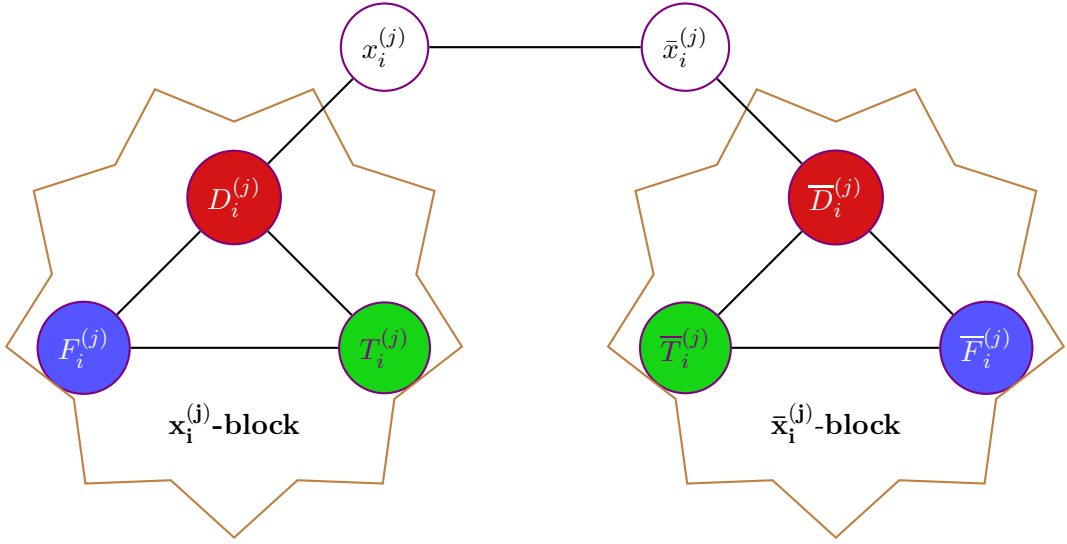
\begin{figure}
    \centering	
\begin{tikzpicture}[auto, thick]

\node[r_node] (Di)  at (-4,0) {$D_i^{(j)}$} ;
\node[b_node] (Fi)  at (-6,-2) {$F_i^{(j)}$} ;
\node[g_node] (Ti)  at (-2,-2) {$T_i^{(j)}$} ;

\node[r_node] (bDi)  at (4,0) {$\overline{D}_i^{(j)}$} ;
\node[b_node] (bFi)  at (6,-2) {$\overline{F}_i^{(j)}$} ;
\node[g_node] (bTi)  at (2,-2) {$\overline{T}_i^{(j)}$} ;

\node[star,star points=9,star point ratio=0.8,  minimum width=140pt, draw=brown ] at (-4,-1.45) {};
\node (x_block) [] at (-4, -3) {${x}_i^{(j)}$-\textrm{block}};
\node[star,star points=9,star point ratio=0.8,  minimum width=140pt, draw=brown ] at (4,-1.45) {};
\node (bx_block) [] at (4, -3) {$\bar{x}_i^{(j)}$-\textrm{block}};

\node[w_node] (xi)  at (-2,2) {$x_i^{(j)}$} ;
\node[w_node] (bxi)  at (2,2) {$\bar{x}_i^{(j)}$} ;
\foreach \source/\dest in {Di/Fi, Fi/Ti, Ti/Di, bDi/bFi, bFi/bTi, bTi/bDi, xi/bxi, xi/Di, bxi/bDi}
    \path (\source) edge (\dest);
\end{tikzpicture}
	
\caption{Literal blocks. This gadget enforces that the $j^{\text{th}}$ copy of literal $x_i$ and $\bar{x}_i$
will always be assigned opposite truth values, as long as the assignments of \textsc{True}, \textsc{False} 
and \textsc{Dummy} nodes in $x_i^{(j)}$-block and $\bar{x}_i^{(j)}$-block are consistent 
with their corresponding nodes in the rest of the graph.} 
\label{fig:literal_block}
\end{figure}

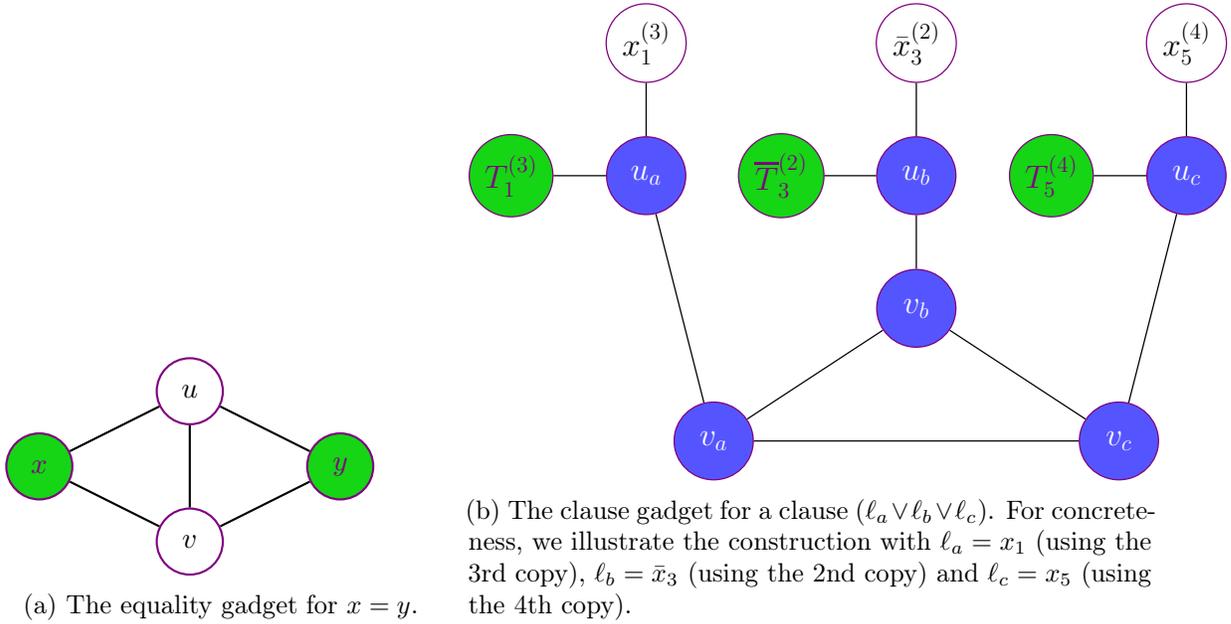
\begin{figure}
    \centering
    \begin{subfigure}[b]{0.3\textwidth}
        \begin{tikzpicture}[auto, thick]

		\node[g_node] (x)  at (0,0) {$x$} ;
		\node[g_node] (y)  at (4,0) {$y$} ;
		\node[w_node] (u)  at (2,1) {$u$} ;
		\node[w_node] (v)  at (2,-1) {$v$} ;

		\foreach \source/\dest in {x/u, x/v, y/u, y/v, u/v}
			\path (\source) edge (\dest);
		\end{tikzpicture}
        \caption{Equality gadget for $x=y$.}
        \label{fig:equality_gadget}
    \end{subfigure}
    ~ 
    \begin{subfigure}[b]{0.5\textwidth}
	\resizebox{3.5in}{2.25in}{%
        \begin{tikzpicture}[auto, thick]
		\node[g3_node] (T1)  at (-9,0) {\huge $T_{1}^{(3)}$} ;
		\node[g3_node] (T2)  at (-3,0) {\huge $\overline{T}_{3}^{(2)}$} ;
		\node[g3_node] (T3)  at (3,0) {\huge $T_{5}^{(4)}$} ;

		\node[w3_node] (l1)  at (-6,3) {\huge $x_{1}^{(3)}$} ;
		\node[w3_node] (l2)  at (0,3) {\huge $\bar{x}_{3}^{(2)}$} ;
		\node[w3_node] (l3)  at (6,3) {\huge $x_{5}^{(4)}$} ;

		\node[b3_node] (u1)  at (-6,0) {\huge $u_a$} ;
		\node[b3_node] (u2)  at (0,0) {\huge $u_b$} ;
		\node[b3_node] (u3)  at (6,0) {\huge $u_c$} ;

		\node[b3_node] (v1)  at (-4.5,-6) {\huge $v_a$} ;
		\node[b3_node] (v2)  at (0,-3) {\huge $v_b$} ;
		\node[b3_node] (v3)  at (4.5,-6) {\huge $v_c$} ;

		\foreach \source/\dest in {T1/u1, T2/u2, T3/u3, l1/u1, l2/u2, l3/u3, v1/v2, v2/v3, v3/v1, v1/u1, v2/u2, v3/u3}
			\path (\source) edge (\dest);
		\end{tikzpicture}
	}
        \caption{The clause gadget for a clause $(\ell_a \lor \ell_b \lor \ell_c)$. For concreteness, we illustrate
		the construction with $\ell_a=x_1$ (using the $3$rd copy), $\ell_b=\bar{x}_3$ (using the $2$nd copy) 
		and $\ell_c= x_5$ (using the $4$th copy).}
        \label{fig:clause_gadget}
    \end{subfigure}
    ~ 
    \caption{Two basic gadgets of BOT graphs.}\label{fig:gadget_1}
\end{figure}

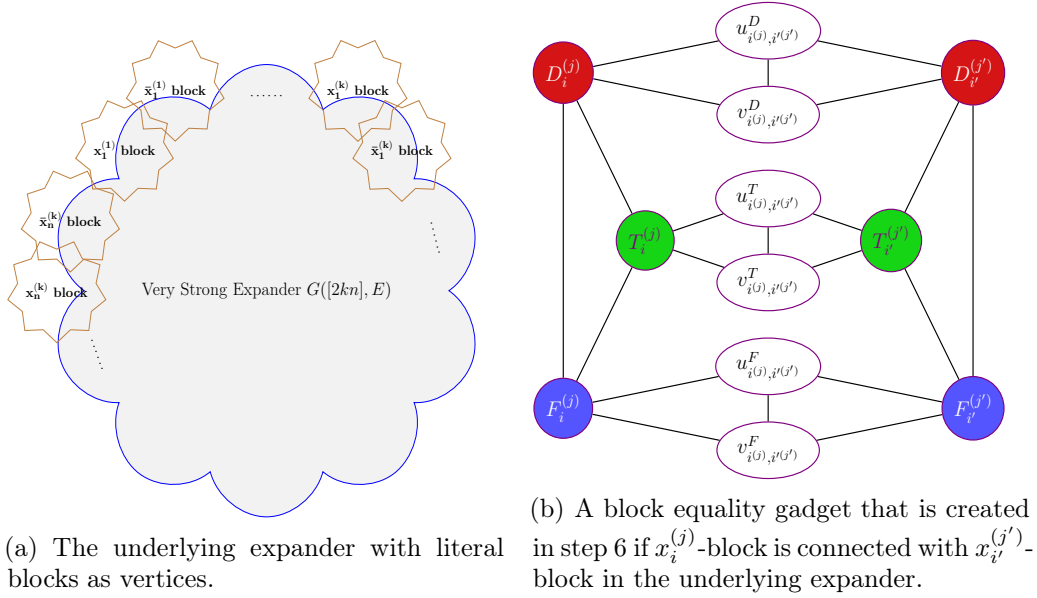
\begin{figure}
    \centering
    \begin{subfigure}[b]{0.40\textwidth}
	\resizebox{2.5in}{2.5in}{%
        \begin{tikzpicture}[auto, thick]
		\node[cloud, draw =blue,text=black,fill = gray!10, minimum width=50pt] (c) at (0,0) 
		     {\Large Very Strong Expander $G([2kn], E)$};

		\node[star,star points=9,star point ratio=0.8,  minimum width=5pt, draw=brown ] at (-4.2,4)  {\textbf{$\mathbf{x_1^{(1)}}$ block}};
		\node[star,star points=9,star point ratio=0.8,  minimum width=5pt, draw=brown ] at (-2.7,5.7)  {\textbf{$\mathbf{\bar{x}_1^{(1)}}$ block}};

		\node [] at (0, 5.5) {$\mathbf{\cdots \cdots}$};

		\node[star,star points=9,star point ratio=0.8,  minimum width=5pt, draw=brown ] at (4,4)  {\textbf{$\mathbf{\bar{x}_1^{(k)}}$ block}};
		\node[star,star points=9,star point ratio=0.8,  minimum width=5pt, draw=brown ] at (2.7,5.7)  {\textbf{$\mathbf{{x}_1^{(k)}}$ block}};

		\node [rotate=105] at (5,1.5) {$\mathbf{\cdots \cdots}$};
		\node [rotate=110] at (-5,-1.8) {$\mathbf{\cdots \cdots}$};

		\node[star,star points=9,star point ratio=0.8,  minimum width=5pt, draw=brown ] at (-5.8,2.0)  {\textbf{$\mathbf{\bar{x}_n^{(k)}}$ block}};
		\node[star,star points=9,star point ratio=0.8,  minimum width=5pt, draw=brown ] at (-6.2,0)  {\textbf{$\mathbf{{x}_n^{(k)}}$ block}};
		
		\end{tikzpicture}
	}
        \caption{The underlying expander with literal blocks as vertices.}
        \label{fig:expander}
    \end{subfigure}
    ~ 
    \begin{subfigure}[b]{0.40\textwidth}
	\resizebox{2.5in}{2.5in}{%
        \begin{tikzpicture}[auto, thick]
		\node[r_node] (Di)  at (-5,4) {\Large $D_i^{(j)}$} ;
		\node[b_node] (Fi)  at (-5,-4) {\Large $F_i^{(j)}$} ;
		\node[g_node] (Ti)  at (-3,0) {\Large $T_i^{(j)}$} ;

		\node[r_node] (bDi)  at (5,4) {\Large $D_{i'}^{(j')}$} ;
		\node[b_node] (bFi)  at (5,-4) {\Large $F_{i'}^{(j')}$} ;
		\node[g_node] (bTi)  at (3,0) {\Large $T_{i'}^{(j')}$} ;

		\node[w2_node]  (uT) at (0,1) {\Large $u^{T}_{i^{(j)},i'^{(j')}}$};
		\node[w2_node]  (vT) at (0,-1) {\Large $v^{T}_{i^{(j)},i'^{(j')}}$};

		\node[w2_node]  (uD) at (0,5) {\Large $u^{D}_{i^{(j)},i'^{(j')}}$};
		\node[w2_node]  (vD) at (0,3) {\Large $v^{D}_{i^{(j)},i'^{(j')}}$};

		\node[w2_node]  (uF) at (0,-3) {\Large $u^{F}_{i^{(j)},i'^{(j')}}$};
		\node[w2_node]  (vF) at (0,-5) {\Large $v^{F}_{i^{(j)},i'^{(j')}}$};

\foreach \source/\dest in {Di/Fi, Fi/Ti, Ti/Di, bDi/bFi, bFi/bTi, bTi/bDi, Ti/uT, uT/bTi, bTi/vT, vT/Ti, uT/vT, Di/uD, uD/bDi, bDi/vD, vD/Di, uD/vD,  Fi/uF, uF/bFi, bFi/vF, vF/Fi, uF/vF}
    \path (\source) edge (\dest);
		\end{tikzpicture}
	}
        \caption{A block equality gadget that is created in step~\ref{step:six} if $x_{i}^{(j)}$-block is connected 
		    with $x_{i'}^{(j')}$-block in the underlying expander.}
        \label{fig:expander_edge}
    \end{subfigure}
    ~ 
    \caption{The expander part of BOT graph.}\label{fig:BOT_graph}
\end{figure}

\subsection{Properties of BOT graphs and proof of Lemma~\ref{lem:BOT-graph}}
We first record the following correctness result of $G_{\phi}$ from~\cite{BOT02}, which follows a similar proof
of the classic reduction from $3$-SAT to \textsc{$3$-Coloring}.
\begin{prop}[\cite{BOT02}]\label{prop:BOT_correctness}
The $3$-CNF formula $\phi$ is satisfiable if and only if the graph $G_{\phi}$ is $3$-colorable.
\end{prop}

Next we do some simple calculations on the parameters of $G_{\phi}$.

\begin{prop}\label{prop:BOT_graph_basic}
Suppose the underlying expander is $d$-regular. Then $G_{\phi}$ satisfies the following:
\begin{itemize}
\item The maximum degree is $\max(2d+3, 2k+1)\leq 2\max(k,d)+3$;
\item $|V(G_{\phi})|=\Theta(nk)\cdot \max(k,d)$; and
\item $|E(G_{\phi})|=\Theta(nk)\cdot \max(k,d)$.
\end{itemize}
\end{prop} 

\begin{proof}
The maximum degree follows by noticing that (for simplicity we only consider nodes for the un-negated literals) 
$\deg(D_i^{(j)})=2d+3$, $\deg(F_i^{(j)})=2d+2$, $\deg(T_i^{(j)})\leq 2d+3$, and $\deg(x_i^{(j)})\leq 2k+1$.

As the lower bounds on the numbers of vertices and edges are straightforward, we prove the upper bounds only. Note that every equality gadget introduces $2$ auxiliary nodes and $5$ additional edges.
Similarly, every clause gadget introduces $6$ auxiliary nodes and $12$ additional edges.
The nodes in $G_{\phi}$ consist of $2nk$ literal nodes and $6nk$ literal block nodes and auxiliary nodes introduced
by gadgets. That is
\begin{align*}
|V(G_{\phi})|&=8nk+ 2n\cdot \frac{k(k-1)}{2} \cdot 2 +  3 \cdot \frac{2nk \cdot d}{2} \cdot 2 +6m \\
             &\leq 8nk+2n(k^2-k)+6nkd+4nk =O(nk)\cdot \max(k,d),
\end{align*}
where in the second-to-last step we use the bound $m\leq 2kn/3$.
Finally,
\begin{align*}
|E(G_{\phi})|&=nk(3+6)+ 2n\cdot \frac{k(k-1)}{2} \cdot 5 +  3 \cdot \frac{2nk \cdot d}{2} \cdot 5 +12m \\
             &\leq 9nk+5n(k^2-k)+15nkd+8nk =O(nk)\cdot \max(k,d),
\end{align*}
where the first term counts the edges in Fig.~\ref{fig:literal_block} for all $k$ copies of each \emph{variable},
the next two terms count the total number of edges in the equality constraints, and the last term counts
the total number of edges in the clause gadgets.
\end{proof}

We say an equality gadget is \emph{broken} if any of the $5$ edges in Fig.~\ref{fig:equality_gadget} is deleted,
and say a literal node $x_i^{(j)}$ (resp. $\bar{x}_i^{(j)}$) is \emph{isolated} if
either all its $k-1$ equality connections with other copies of $x_i$ (resp. $\bar{x}_i$) are all broken,
or all its $d$ equality connections with other literal nodes in the expander graph are all broken.
We say a literal pair $(x_i^{(j)}, \bar{x}_i^{(j)})$ is \emph{damaged} if
any of the edges within the two literal blocks in Fig.~\ref{fig:literal_block} is deleted,
or either $x_i^{(j)}$ node or $\bar{x}_i^{(j)}$ node becomes isolated.

Let $E' \subset E(G_{\phi})$ be a subset of the edges of graph $G_{\phi}$.
What happens if we delete all the edges in $E'$?
Define the \emph{survived subgraph induced by $E'$} of $G_{\phi}$, denoted $G'_{\phi, E'}$ as follows.
Delete all literal pairs that are damaged from deleting edges in $E'$, 
and delete all the edges connected with damaged pairs;
delete all clause gadgets that contain one or more damaged literals. 
We make the following two simple observations on a survived subgraph $G'_{\phi, E'}$.

\begin{claim}\label{claim:subgraph_colorable}
Let $\phi'$ be the $3$-CNF formula by keeping only the survived clauses in $G'_{\phi, E'}$.
Then $\phi'$ is satisfiable if and only if $G'_{\phi, E'}$ is $3$-colorable.
\end{claim}
\begin{proof}
This follows directly by noticing that $G'_{\phi, E'}$ can be regarded as the BOT graph constructed
from the formula $\phi'$.
\end{proof}

\begin{claim}\label{claim:survived_clauses}
If we delete a subset of edges $E'$, then the number of survived clauses in $\phi'$ is at least
$m-2|E'|$.
\end{claim}
\begin{proof}
This is because deleting an edge causes at most one literal pair to be damaged, and thus affects at most
two clauses.
\end{proof}

Our main goal of this Section is the prove the following lemma on the
gap-preserving reduction from \textsc{Max-$3$SAT($k$)} to \textsc{$3$-Coloring} of bounded-degree graphs
of Bogdanov, Obata and Trevisan~\cite{BOT02}.
\begin{replemma}{lem:BOT-graph}[restatement]
There are absolute constants $d$ and $\eps'>0$ such that the following holds.
For infinitely many integers $n$, there are two families of degree-$d$ bounded graphs $G_{\phi, Y}$ and $G_{\phi, N}$ of size $n$ 
such that: every $G_\phi \in G_{\phi, Y}$ is $3$-colorable; no $1-\eps'$ fraction of the edges of any $G_\phi \in G_{\phi, N}$ is $3$-colorable; and yet when given either $G_{\phi} \in G_{\phi, Y}$ or $G_{\phi} \in G_{\phi, N}$, it is NP-hard to distinguish which family $G_\phi$ is from.
\end{replemma}

\begin{proof}
Recall that, by Lemma~\ref{lem:max-3sat}, there is an integer $k$ and constant $\eps(k)>0$ such that there are two families of instances $\phi_{Y}$ and $\phi_{N}$ in \textsc{Max-$3$SAT($k$)} of size $n$ each such that every $\phi \in \phi_{Y}$ is satisfiable, any truth assignment can satisfy at most $1-\eps(k)$ fraction of the clauses of any $\phi \in \phi_{N}$, and it is NP-hard to distinguish between these two families. To translate this NP hardness of approximation to \textsc{$3$-Coloring}, we use the very strong expander given in Theorem~\ref{thm:strong.expander} to construct the corresponding families of BOT graphs $G_{\phi, Y}$ and $G_{\phi, N}$ for $\phi\in\phi_{Y}$ and $\phi\in\phi_{N}$, respectively (we just need to make sure that the number of vertices of the very strong expander is at least $2nk$; if the number of vertices is larger, then we can simply make some of the literals more than $k$ copies). Let $n$ be the number of variables and $m$ be the number of clauses of the Boolean formula, where $m\geq n$. Define $\eps'=\frac{\eps \cdot n}{2|E(G_{\phi, N})|}$. Since both $k$ and $d$ are constants, by Proposition~\ref{prop:BOT_graph_basic}, $\eps'$ is an absolute constant depending on $\eps$, $k$ and the degree $d$ of the underlying expander graph. Clearly every $G_\phi \in G_{\phi, Y}$ is $3$-colorable. It is also not hard to prove that any $G_\phi \in G_{\phi, N}$ is not $(1-\eps')$ $3$-colorable: Indeed, suppose that $G_\phi\in G_{\phi, N}$ is $(1-\eps')$ $3$-colorable. This implies that we can delete a subset of edges from $G_\phi$ of size at most $\eps'|E(G_{\phi, N})|$, and obtain a survived subgraph with at least $m-\eps\cdot n\geq (1-\eps)m$ clauses (by Claim~\ref{claim:survived_clauses}) which is $3$-colorable. By Claim~\ref{claim:subgraph_colorable}, the $3$-coloring of the survived subgraph can be decoded into a satisfying assignment for the (at least $(1-\eps)m$) survived clauses of $\phi\in \phi_{N}$, a contradiction. Finally the constant degree bound of the constructed BOT graphs follows from the first item of Proposition~\ref{prop:BOT_graph_basic}.
\end{proof}

\section{Structure of max-likelihood solutions -- Proof of Theorem~\ref{thm:diagonal}}\label{sec:diagonal}
For convenience of the reader, we repeat Theorem~\ref{thm:diagonal} here.
\begin{reptheorem}{thm:diagonal}[restatement]
Let $K$ be a marginal kernel with likelihood $\llhd(K)$.  
Then there exists a marginal kernel $K'$ with 
$\llhd(K')\leq \llhd(K)$ such that the diagonal of $K'$ (indexed by vertices and edges of $G_\phi$) satisfies
\[K'_{ii}=
\begin{cases}
\frac{\deg_{G_\phi}(u)}{m} & \mathrm{for \;} i=u \in V(G_\phi); \\
\frac{1}{m} & \mathrm{for \;} i=(u,v) \in E(G_\phi).
\end{cases}
\]
\end{reptheorem}
The proof follows from the following slightly stronger lemma, which implies that
we can assume that the diagonal of a max likelihood kernel $K$ (indexed by vertices and edges of $G_\phi$) satisfies
that $K_{ii}$ is equal to the normalized frequency of element $i$ in the training set.

\begin{lemma}\label{lem:BMRU_diagonals}
Let $[N]$ be the ground set, $\mathcal{S}=\{X_j\}_{j=1}^{m}$ be a set of  $m$ example sets with 
each $X_j \subseteq [N]$.
Let $D$ be the empirical distribution induced by the set $\mathcal{S}$: namely, for every 
$X\subseteq [N]$, $D(X)=|\{X \in \mathcal{S}\}|/m$ if $X$ is in $\mathcal{S}$, and $D(X)$ is zero otherwise.
Then there is a maximum likelihood marginal kernel $K$ which satisfies that 
the diagonal entry of $K$ at each element in $[N]$
is equal to the element's empirical frequency in $\mathcal{S}$. That is,
$K_{ii} = \displaystyle\sum_{X\in \mathcal{S}: X \ni i} D(X)$ for every $i\in [N]$.
\end{lemma}

The rest of this section is devoted to a proof of Lemma~\ref{lem:BMRU_diagonals}.

\subsection{The $L$-ensemble case}
We start by assuming that all of $K$'s eigenvalues are in $(0,1)$, and
follow a similar strategy as of the proof of Proposition $13$ in~\cite{BMRU17}.    
Note that in this case the DPP is an $L$-ensemble, where $L = K(I-K)^{-1}$.  If $K$ is a maximum likelihood DPP, 
then the corresponding $L$ is a critical point of the log-likelihood function, 
and accordingly the directional derivative of the log-likelihood function in every direction is zero.  
Recall that for any matrices $M,H \in \R^{N\times N}$, the directional derivative of the log-determinant at $M$ in the ``direction'' $H$ (viewing $M$ and $H$ as vectors) is given by $\frac{\partial}{\partial x}[\log\det(M+xH)]=\mathrm{Tr}(M^{-1}H)$ (see, e.g., Example A.3 of~\cite{BV04}).
Thus, in the direction of any matrix $H \in \R^{N\times N}$, the log-likelihood satisfies

\[
\sum_{X\subseteq [N]} D(X) \mathrm{Tr}(L_X^{-1}H_X) - \mathrm{Tr}((I+L)^{-1}H) = 0,
\] 
where $L_X$ and $H_X$ denote the respective principal submatrices of $L$ and $H$ indexed by the subset $X$ of $[N]$.

Fix $t_1,t_2,\ldots,t_n \in \R$, and let $T$ be the diagonal matrix with diagonal $t_1,t_2,\ldots,t_n$.  
Setting $H = \frac12LT + \frac12 TL$ (this ensures that we consider symmetric kernels), 
we get $H_X = \frac12L_X T_X + \frac12T_X L_X$.  Since $(I+L)^{-1}$ and 
$L$ commute, and the trace is invariant under cyclic permutations, we get

\[
\sum_{X\subseteq [N]} D(X) \sum_{j \in X} t_j = \mathrm{Tr}(KT) = \sum_{j=1}^N K_{jj}t_j
\]

Fix $i \in [N]$, and consider the setting where $t_i = 1$ and $t_j = 0$ for all $j \neq i$.  
The above equation becomes

\begin{equation}\label{eqn:BMRU}
\sum_{X\subseteq [N]:X \ni i} D(X) = K_{ii},
\end{equation}
completing the proof of the lemma for the $L$-ensemble case.

\subsection{Perturbing the training example distribution}
To deal with the non $L$-ensemble case, we apply a continuity argument: 
specifically, we employ a Lipschitz property of the log likelihood function with respect to small perturbations
of the training example distribution.
In the rest of the proof we think $N$ (hence both $n$ and $m$) as a fixed constant and let
the small quantities such as $\eps$ tend to zero independent of $N$.

Let $D$ be the original empirical distribution induced by the sample set $\mathcal{S}$ of size $m$.
Without loss of generality, we assume\footnote{For if $D(\emptyset)\neq 0$, any optimal DPP kernel for $D$ 
is necessarily an $L$-ensemble. We further assume $D([N])=0$ to simplify exposition. 
It is easy to see
that the there are two parts in the definition of $K^{*}_{\eps}$ below: the rescaling part is to 
deal with DPP kernel singularity at $\emptyset$, and the spectrum shifting part is to deal with DPP kernel
singularity at $[N]$. One can therefore keep only the rescaling part in $K^{*}_{\eps}$ when
$D([N])$ is bounded away from zero.} both $D(\emptyset)=0$ and $D([N])=0$.
We perturb $D$ by adding very tiny probabilities at $\emptyset$ and $[N]$ so that the corresponding correlation
kernel is an $L$-ensemble.\footnote{Making the empirical distribution to be non-zero at $\emptyset$ ensures that
the kernel is an $L$-ensemble; however, in order to apply the matrix derivative argument in~\cite{BMRU17},
we need to make the distribution to be non-zero at $[N]$ as well.} Formally, define
$D_{\eps}$ to be a distribution on $2^{[N]}$ such that for any $X \subseteq [N]$
\begin{equation}\label{eqn:perturbed-dist}
D_{\eps}(X)=
\begin{cases}
(1-\eps)D(X) & \text{if $X\in \mathcal{S}$; }\\
\eps/2 & \text{if $X=\emptyset$; }\\
\eps/2 & \text{if $X=[N]$; }\\
0 & \text{otherwise.}
\end{cases}
\end{equation}
We stress that the small quantity $\eps$ is independent of the DPP learning instance under consideration.

In the following, we will use $K^{*}$ to denote a (non $L$-ensemble) optimal marginal kernel 
for the original training distribution $D$,
and use $K_{\eps}$ to denote an optimal marginal kernel for the perturbed training distribution $D_{\eps}$.
Since $D_{\eps}(\emptyset)> 0$ and $D_{\eps}([N])>0$, 
$K_{\eps}$ is necessarily an $L$-ensemble and hence by our previous argument, 
the diagonals of $K_{\eps}$ satisfy \eqref{eqn:BMRU}.
In particular, for every $i\in [N]$
\begin{align}\label{eqn:perturbed-diagonal}
(K_{\eps})_{ii} = \frac{\eps}{2}+\sum_{X\in \mathcal{S}: X \ni i} (1-\eps)D(X), 
\end{align}
which approaches $\sum_{X\in \mathcal{S}: X \ni i} D(X)$ when $\eps$ tends to zero.

To make the dependency of the maximum log likelihood estimator on the empirical distribution explicit,
we write $\ell(K,D)$ for the log likelihood function of DPP kernel $K$ on training distribution $D$. Therefore,  
\begin{align}\label{eqn:likelihood-original}
\ell(K^{*}, D)=\sum_{X\in \mathcal{S}}D(X) \log \left(\frac{1}{\Pr_{\bY \sim \calP_{K^{*}}}[\bY=X]} \right),
\end{align}
and
\begin{align}\label{eqn:likelihood-perturbed}
\ell(K_{\eps}, D_{\eps})= & \frac{\eps}{2}\log \left(\frac{1}{\Pr_{\bY \sim \calP_{K_{\eps}}}[\bY=\emptyset]} \right)
+\sum_{X\in \mathcal{S}}(1-\eps)D(X) \log \left(\frac{1}{\Pr_{\bY \sim \calP_{K_{\eps}}}[\bY=X]} \right) \nonumber\\
   &+\frac{\eps}{2}\log \left(\frac{1}{\Pr_{\bY \sim \calP_{K_{\eps}}}[\bY=[N]]} \right).
\end{align}

The main goal in this subsection is to prove the following inequality, which
says that the $L$-ensemble marginal kernel $K_{\eps}$, 
which is defined to be optimal for the perturbed training distribution $D_{\eps}$, 
is also close to being optimal for the original training distribution $D$.

\begin{prop}\label{prop:perturbed}
Let $D$, $K^{*}$, $D_{\eps}$ and $K_{\eps}$ be defined as before. Then for every $\eps>0$ that is small enough
\begin{align}\label{eqn:perturbed-near-optimal}
\ell(K_{\eps}, D) \leq \ell(K^{*}, D)+O(\eps \log{\frac{1}{\eps}}),
\end{align}
where the hidden constant in the $O(\eps \log{\frac{1}{\eps}})$ term depends only on $N$.
\end{prop}

\begin{proof}
Let $I_{N}$ denote the $N\times N$ identity matrix. Define a new marginal kernel $K^{*}_{\eps}$ as 
\[
K^{*}_{\eps}=(1-\eps)K^{*}+\frac{\eps}{2}I_{N}.
\]

Clearly, since $K^{*}$ is PSD, $\lambda_{N}(K^{*}_{\eps})\geq \eps/2$.  
It is also not hard to verify that $\lambda_{1}(K^{*}_{\eps})=(1-\eps)\lambda_{1}(K^{*})+\eps/2\leq 1-\eps/2$.
It follows that $K^{*}_{\eps}$ is indeed a DPP marginal kernel.

We will prove the following sequence of inequalities
\begin{align}
\ell(K^{*}_{\eps}, D) &\leq \ell(K^{*}, D)+O(\eps). \label{eqn:seq1} \\
\ell(K^{*}_{\eps}, D_{\eps}) &\leq \ell(K^{*}_{\eps}, D)+O(\eps \log{\frac{1}{\eps}}). \label{eqn:seq2} \\
\ell(K_{\eps}, D_{\eps}) &\leq \ell(K^{*}_{\eps}, D_{\eps}). \label{eqn:seq3} \\
\ell(K_{\eps}, D) &\leq \ell(K_{\eps}, D_{\eps})+O(\eps). \label{eqn:seq4}
\end{align}

First of all, let's state and prove two useful bounds.

\begin{fact}\label{fact:trivial-kernel}
Let $D'$ be an arbitrary empirical distribution over the ground set $[N]$,
and let $\ell^{*}(D')$ denote the optimal value of the DPP maximum log likelihood estimator
for $D'$. Then $\ell^{*}(D')\leq N\log{2}$.
\end{fact}
\begin{proof}
This is because the trivial diagonal kernel $K_{I}:=\frac{1}{2}I_{N}$, 
which corresponds to the uniform distribution over all $2^{[N]}$ subsets, 
is a legal DPP marginal kernel for any distribution $D'$ over $[N]$, 
and its log likelihood estimator is easily seen to be $N\log{2}$.
\end{proof}

\begin{fact}\label{fact:minimum-prob}
Let $\mathcal{S}=\{X_j\}_{j=1}^{m}$ be a set of $m$ example sets with each $X_j \subseteq [N]$, the ground set. 
Let $K$ be any maximum likelihood marginal kernel of the empirical distribution induced by 
$\mathcal{S}$. Then, for any $X_j \in \mathcal{S}$,
\[
\Pr_{\bY \sim \calP_K}[\bY=X_j]\geq e^{-\log{2}\cdot mN} > e^{-\log{2} \cdot N^{2}}.
\]
\end{fact}
\begin{proof}
This follows by combining Fact~\ref{fact:trivial-kernel} with 
the fact that in \eqref{eqn:likelihood-original}, $D(X_j)\geq 1/m$.
\end{proof}

\subsubsection{Proof of inequality~\eqref{eqn:seq1}} 

For any DPP marginal kernel $K$ and $X\subseteq [N]$, it is well-known that using
inclusion-exclusion principle, we can express
$\Pr_{\bY \sim \calP_K}[\bY=X]$ in terms of $K$ as 
\begin{align*}
\Pr_{\bY \sim \calP_K}[\bY=X]
& =\sum_{T\supseteq X} (-1)^{|T \setminus X|} \det(K_{T})= (-1)^{|X|}\sum_{T\supseteq X} (-1)^{|T|} \det(K_{T})\\
& = (-1)^{|X|}\det(I_{\bar{X}}-K) =|\det(K-I_{\bar{X}})|,
\end{align*}
where $I_{\bar{X}}$ stands for the diagonal matrix whose $(i,i)$-entry is $1$ if
$i \in \bar{X}$ and is $0$ otherwise.

Now fix any $X\in \mathcal{S}$. Then $\Pr_{\bY \sim \calP_{K^*}}[\bY=X]=|\det(K^{*}-I_{\bar{X}})|=|\det(A)|$,
where $A:=K^{*}-I_{\bar{X}}$.
By Fact~\ref{fact:minimum-prob}, $|\det(A)|\geq e^{-mN\log{2}}$.
Similarly, 
\[
\Pr_{\bY \sim \calP_{K^{*}_{\eps}}}[\bY=X]=|\det(K^{*}_{\eps}-I_{\bar{X}})|=
|\det(A+\eps E)|,
\]
where $E:=\frac{1}{2}I_{N}-K^{*}$.

We use the following result to bound the difference between $\det(A+\eps E)$ and $\det(A)$.
\begin{theorem}[\cite{IR08}] \label{thm:IR08}
Let $A$ and $E$ be two $N\times N$ complex matrices. Let $\sigma_1\geq \cdots \geq \sigma_N$ be the singular
values of $A$, and suppose that the spectral norm of $E$ satisfies $\lVert E \rVert_{2} <1$. Then
\[
\left| \det(A+E) -\det(A)\right| \leq s_{N-1}\lVert E \rVert_{2} +O(\lVert E \rVert_{2}^{2}),
\]
where $s_{N-1}$ is the $(N-1)^{\text{st}}$ elementary symmetric function in the singular values of $A$
and is upper bounded by $N\sigma_1 \cdots \sigma_{N-1}$.
\end{theorem}

Since $A=K^{*}-I_{\bar{X}}$ is symmetric, 
$\{\sigma_1(A), \ldots, \sigma_N(A)\}=\{|\lambda_1(A)|, \ldots, |\lambda_N(A)|\}$.
To bound the singular values of $A$, 
we recall Weyl's theorem on the changes to eigenvalues of a Hermitian matrix that is perturbed.
Specifically, the theorem states that if $A, B$ and $C$ are Hermitian matrices of size $n\times n$, $C=A+B$,
then
\begin{align*}
\lambda_{1}(A)+\lambda_{n}(B) &\leq \lambda_{1}(C) \leq \lambda_{1}(A)+\lambda_{1}(B) \\
                       \cdots & \cdots  \cdots \\
\lambda_{n}(A)+\lambda_{n}(B) &\leq  \lambda_{n}(C) \leq \lambda_{n}(A)+\lambda_{1}(B)
\end{align*}
Now all the eigenvalues of $K^{*}$ are in $[0,1]$, and all the eigenvalues of $-I_{\bar{X}}$ are in $\{-1, 0\}$,
therefore, all the eigenvalues of $A=K^{*}-I_{\bar{X}}$ are in $[-1,1]$. It follows that 
$\sigma_1, \ldots, \sigma_N \in [0,1]$ and $s_{N-1}\leq N$.

Since $E=\frac{1}{2}I_{N}-K^{*}$, its eigenvalues are just the corresponding eigenvalues of $-K^{*}$ 
but shifted by $\frac{1}{2}$, hence all eigenvalues of $E$ are in $[-\frac{1}{2}, \frac{1}{2}]$
and $\lVert E \rVert_{2}\leq \frac{1}{2}$.   

Now we can plug in all these bounds into Theorem~\ref{thm:IR08} to obtain
\begin{align*}
\left|\Pr_{\bY \sim \calP_{K^{*}}}[\bY=X] - \Pr_{\bY \sim \calP_{K^{*}_{\eps}}}[\bY=X]\right|
& \leq \left| \det(K^{*}-I_{\bar{X}}) - \det\left(K^{*}-I_{\bar{X}}+\eps(\frac{I_N}{2}-K^*)\right)\right| \\
& \leq \frac{\eps N}{2}+O(\frac{\eps^2}{4}) <\eps N,
\end{align*}
for all small enough $\eps$.

It follows that, for all small enough $\eps$, 
\[
\Pr_{\bY \sim \calP_{K^{*}_{\eps}}}[\bY=X] \geq \Pr_{\bY \sim \calP_{K^{*}}}[\bY=X]-\eps N
\geq \Pr_{\bY \sim \calP_{K^{*}}}[\bY=X](1-Ne^{\log{2} \cdot N^{2}}\eps)
\geq \Pr_{\bY \sim \calP_{K^{*}}}[\bY=X] e^{-O(\eps)},
\]
where the last step follows from the inequality that $1-x \geq e^{-\frac{3}{2}x}$ for all $0\leq x \leq 1/2$.

Since this lower bound on DPP probability of kernel $K^{*}_{\eps}$ holds for every $X \in \mathcal{S}$, plugging it
into the log likelihood function~\eqref{eqn:likelihood-original} for the empirical distribution $D$
proves inequality~\eqref{eqn:seq1}.

\subsubsection{Proof of inequality~\eqref{eqn:seq2}} 
Since $D(\emptyset)=D([N])=0$, we have
\begin{align*}
\ell(K^{*}_{\eps}, D_{\eps}) 
&= (1-\eps)\ell(K^{*}_{\eps}, D)
+ \frac{\eps}{2}\log \left(\frac{1}{\Pr_{\bY \sim \calP_{K^{*}_{\eps}}}[\bY=\emptyset]} \right)
+ \frac{\eps}{2}\log \left(\frac{1}{\Pr_{\bY \sim \calP_{K^{*}_{\eps}}}[\bY=[N]]} \right) \\
&\leq \ell(K^{*}_{\eps}, D)
+ \frac{\eps}{2}\log \left(\frac{1}{\Pr_{\bY \sim \calP_{K^{*}_{\eps}}}[\bY=\emptyset]} \right)
+ \frac{\eps}{2}\log \left(\frac{1}{\Pr_{\bY \sim \calP_{K^{*}_{\eps}}}[\bY=[N]]} \right).
\end{align*}
We can lower bound $\Pr_{\bY \sim \calP_{K^{*}_{\eps}}}[\bY=\emptyset]$
and $\Pr_{\bY \sim \calP_{K^{*}_{\eps}}}[\bY=[N]]$, using the bounds on $\lambda_{1}(K^{*}_{\eps})$ 
and $\lambda_{N}(K^{*}_{\eps})$ obtained right after the definition of $K^{*}_{\eps}$, as follows.
\[
\Pr_{\bY \sim \calP_{K^{*}_{\eps}}}[\bY=\emptyset]=\det(I-K^{*}_{\eps})=
\prod_{i=1}^{N}\left(1-\lambda_{i}(K^{*}_{\eps}) \right) \geq (\frac{\eps}{2})^N.
\] 

Similarly,
\[
\Pr_{\bY \sim \calP_{K^{*}_{\eps}}}[\bY=[N]]=\det(K^{*}_{\eps})
=\prod_{i=1}^{N}\lambda_{i}(K^{*}_{\eps})\geq (\frac{\eps}{2})^N.
\]

Therefore,
\[
\ell(K^{*}_{\eps}, D_{\eps}) \leq \ell(K^{*}_{\eps}, D)+N\eps\log{\frac{2}{\eps}}
=\ell(K^{*}_{\eps}, D)+O(\eps \log{\frac{1}{\eps}}).
\]
\subsubsection{Proof of inequality~\eqref{eqn:seq3}} 
This follows directly from the fact that $K_{\eps}$ is an optimal kernel for the sample distribution $D_{\eps}$.

\subsubsection{Proof of inequality~\eqref{eqn:seq4}} 
Recall that
\begin{align*}
\ell(K_{\eps}, D_{\eps})= &\frac{\eps}{2}\log \left(\frac{1}{\Pr_{\bY \sim \calP_{K_{\eps}}}[\bY=\emptyset]} \right)
+\sum_{X\in \mathcal{S}}(1-\eps)D(X) \log \left(\frac{1}{\Pr_{\bY \sim \calP_{K_{\eps}}}[\bY=X]} \right) \\
 &+\frac{\eps}{2}\log \left(\frac{1}{\Pr_{\bY \sim \calP_{K_{\eps}}}[\bY=[N]]} \right).
\end{align*}
Therefore,
\allowdisplaybreaks
\begin{align*}
\ell(K_{\eps}, D)
& = \sum_{X\in \mathcal{S}}D(X) \log \left(\frac{1}{\Pr_{\bY \sim \calP_{K_{\eps}}}[\bY=X]} \right)\\
& \leq \frac{1}{1-\eps} \ell(K_{\eps}, D_{\eps}) \\
& \leq (1+2\eps)\ell(K_{\eps}, D_{\eps}) 
\qquad\text{(as long as $\eps\leq 1/2$)} \\
& \leq \ell(K_{\eps}, D_{\eps}) +(2\log{2})N\eps \\
& = \ell(K_{\eps}, D_{\eps})+O(\eps),
\end{align*}
where in the second last step we use the fact that $K_{\eps}$ is an optimal DPP kernel for $D_{\eps}$
and the bound in Fact~\ref{fact:trivial-kernel}.

\subsubsection{Putting everything together}
Adding inequalities \eqref{eqn:seq1}, \eqref{eqn:seq2}, \eqref{eqn:seq3} and \eqref{eqn:seq4} together 
yields $\ell(K_{\eps}, D)\leq \ell(K^{*}, D)+O(\eps \log{\frac{1}{\eps}})$, 
which completes the proof of Proposition~\ref{prop:perturbed}.
\end{proof}

\subsection{An optimal kernel as a limiting matrix}
Let $\veps$ be a small enough positive number so that Proposition~\ref{prop:perturbed} holds.
Define an infinite decreasing sequence $\{\eps_1, \ldots, \eps_k, \ldots\}$,
where $\eps_k=\veps/k$.
Define an infinite sequence of distributions $\{D_{\eps_1}, \ldots, D_{\eps_k}, \ldots\}$
as in~\eqref{eqn:perturbed-dist} for each $\eps_k$.
Correspondingly, let $\{K_{\eps_1}, \ldots, K_{\eps_k} \ldots \}$
be a sequence of optimal DPP kernels for distributions $\{D_{\eps_1}, \ldots, D_{\eps_k}, \ldots\}$.
Note that, since there can be more than one optimal DPP kernel for each distribution,
such an infinite sequence of kernels is in general not unique. We just fix one such sequence.

If we view\footnote{Correspondingly, the underlying matrix norm is the Frobenius norm.} 
the set of all (symmetric) $N\times N$ positive semidefinite matrices whose maximum eigenvalues are bounded
by $1$ as a subset $\mathcal{P}$ of $\R^{N^2}$, then for any $M\in \mathcal{P}$,
\[
\sum_{i, j=1}^{N}|M_{i,j}|^2 = \sum_{k=1}^{N}\sigma^{2}_{k}(M)= \sum_{k=1}^{N}\lambda^{2}_{k}(M)\leq N,
\] 
where $\sigma_{k}(M)$ denotes the $k^{\text{th}}$ singular value of matrix $M$.
It follows that $\mathcal{P}$ is a bounded set.

Consider $\{K_{\eps_1}, \ldots, K_{\eps_k}, \ldots \}$, which is an infinite sequence in $\mathcal{P}$.
Recall the Bolzano-Weierstrass theorem, which states that 
every bounded infinite sequence in a finite-dimensional Euclidean space $\R^{N^2}$ has a convergent subsequence.
Therefore there exists an infinite sequence of indices $i_1, \ldots, i_k, \ldots$ of $\mathbb{N}$
such that the infinite subsequence of matrices 
$\{K_{\eps_{i_1}}, \ldots, K_{\eps_{i_k}}, \ldots \}$ converge in $\mathcal{P}$.

Let $K_{\infty}$ be the limiting matrix of the converging sequence 
$\{K_{\eps_{i_1}}, \ldots, K_{\eps_{i_k}}, \ldots \}$. That is,
\[
K_{\infty} := \lim_{k\to \infty}K_{\eps_{i_k}}.
\]

Since each of the marginal kernel in $\{K_{\eps_{i_1}}, \ldots, K_{\eps_{i_k}}, \ldots \}$ 
satisfies the diagonal entry condition in~\eqref{eqn:perturbed-diagonal},
so does $K_{\infty}$. That is, for every $j\in [N]$
\[
(K_{\infty})_{jj} = 
\lim_{k\to \infty} (K_{\eps_{i_k}})_{jj} 
=\lim_{\eps \to 0} \left( \frac{\eps}{2}+\sum_{X\in \mathcal{S}: X \ni j} (1-\eps)D(X) \right)
=\sum_{X\in \mathcal{S}: X \ni j}D(X).
\]

Moreover, by Proposition~\ref{prop:perturbed}, 
\[
\ell(K_{\infty}, D)= \lim_{k\to \infty}\ell(K_{\eps_{i_k}}, D)
\leq \ell(K^{*}, D)+ \lim_{k\to \infty} O(\eps_{i_k} \log{\frac{1}{\eps_{i_k}}})
=\ell(K^{*}, D),
\]
which shows that $K_{\infty}$ is an optimal DPP kernel for the original sample distribution $D$.
This completes the proof of Lemma~\ref{lem:BMRU_diagonals}.

\section{An optimal kernel for $3$-colorable graphs -- Proof of Theorem~\ref{thm:completeness}}
\label{sec:completeness-proof}

We now prove Theorem~\ref{thm:completeness}, which states that if BOT graph $G_{\phi}$ is $3$-colorable (or equivalent, the corresponding $3$-CNF formula $\phi$ is satisfiable), then there exists a rank-$3$ optimal DPP marginal kernel $K$ whose log likelihood value is
$\ell(K)=\ell_{\textsc{yes}}(G_{\phi})=3\log{m}-\frac{1}{m}\sum_{(u,v)\in E(G_{\phi})}\left(\log(\deg_{G_{\phi}}(u))+\log(\deg_{G_{\phi}}(v))\right)$.

\begin{proof}
From Theorem~\ref{thm:diagonal}, 
we can assume that the diagonal of $K$ (indexed by vertices and edges of $G_\phi$) satisfies

\[K_{ii}=
\begin{cases}
\frac{\deg_{G_\phi}(u)}{m} & \mathrm{for \;} i=u \in V(G_\phi) \\
\frac{1}{m} & \mathrm{for \;} i=(u,v) \in E(G_\phi) \\
\end{cases}.
\]

We first prove a lower bound on the log likelihood function of any marginal kernel $K$, by means of the Hadamard inequality.
\begin{lemma}[Hadamard's inequality, see e.g. Theorem 7.8.1 in ~\cite{HJ12}]
\label{lem:det-leq-prod-main-diagonal}
For every positive semidefinite matrix $B$, $\det(B) \leq \prod_i B_{ii}$, with equality if and only if $B$ is diagonal.
\end{lemma}

Therefore, for every example $T = (u, v, (u,v)) \in E'(H_{\phi})$, we 
can upper bound the probability that marginal kernel $K$ outputs $T$ as
\[
\Pr_{\bY \sim \calP_K}[\bY = T] \leq \Pr_{\bY \sim \calP_K}[T \subseteq \bY] = \det({K_T}) 
\leq \frac{\deg_{G_\phi}(u)}{m}\frac{\deg_{G_{\phi}}(v)}{m}\frac{1}{m} = \frac{\deg_{G_\phi}(u)\deg_{G_\phi}(v)}{m^3},
\]
\noindent
The average log likelihood of the DPP associated with any marginal kernel $K$ thus satisfies
\begin{align*}
\llhd(K) &\geq 3 \log m - \frac{1}{m} \displaystyle\sum_{(u,v)\in E(G_\phi)} 
\left( \log \deg_{G_\phi}(u) + \log \deg_{G_\phi}(v) \right). 
\end{align*}

Next, for any $3$-colorable BOT graph $G_{\phi}$, we construct a rank-$3$ marginal kernel $K$
with matching log likelihood function, hence proving the optimality of the kernel.
The kernel $K$ is constructed by the natural $3$-dimensional embedding induced by the vertex coloring of $G_{\phi}$.
 
Let $\chi : V(G_\phi) \to \{1,2,3\}$ be a proper $3$-coloring of $G_{\phi}$.
We extend the coloring function $\chi$ to include the edge set $E(G_\phi)$ in the natural way:
for every $(u, v)\in E(G_\phi)$, let $\chi((u,v))=\{1,2,3\}\setminus \{\chi(u), \chi(v)\}$.
Since $\chi$ is a proper coloring of $G_\phi$, such an extended definition of $\chi$ is unambiguous.
 
Let $\mathbf{e}_1=\begin{pmatrix}1 \\ 0\\ 0 \end{pmatrix}$, 
$\mathbf{e}_2=\begin{pmatrix}0 \\ 1\\ 0 \end{pmatrix}$, and 
$\mathbf{e}_3=\begin{pmatrix}0 \\ 0\\ 1 \end{pmatrix}$. 
Let $N=|V(G_\phi)| + |E(G_\phi)|$.
Let $Q$, the embedding matrix of coloring $\chi$, 
be a $3\times N$ matrix whose columns are indexed by the vertices and edges of $G_\phi$:
\[
Q=\begin{pmatrix} 
\vertbar & \vertbar &        & \vertbar  \\
 q_1     &     q_2  & \cdots &  q_N  \\
 \vertbar & \vertbar &        & \vertbar 
\end{pmatrix},
\]
Each $q_i$ is a $3$-dimensional column vector set as follows:
\begin{itemize}
\item If $i=u\in V(G_\phi)$, then set $q_i=\sqrt{\frac{\deg_{G_\phi}(u)}{m}}\mathbf{e}_{\chi(u)}$;
\item if $i=(u,v) \in E(G_\phi)$, then set $q_i = \sqrt{\frac{1}{m}}\mathbf{e}_{\chi((u,v))}$. 
\end{itemize}
Finally, define the marginal kernel as $K=Q^\top Q$.

Clearly, $\mathrm{rank}(K)=3$. Since $K$ is a Gram matrix, it is positive semidefinite and 
the diagonals of $K$ satisfy the conditions specified in Theorem~\ref{thm:diagonal}.
Next we prove that $\lambda_1(K)$ --- the largest eigenvalue of $K$ --- is equal to $1$, which would imply that
$K$ is indeed a DPP marginal kernel.

By Courant-Fischer's variational formulation of the eigenvalues, 
\[
\lambda_1(K)=\sup_{x\in \R^N, \text{ }\lVert x \rVert_{2} =1}x^{\top} K x = 
\sup_{x\in \R^N, \text{ }\lVert x \rVert_{2} =1}\lVert Qx \rVert^{2}.
\]
For notational convenience, for each column vector in $Q$, let $q_i=\sqrt{p_i}\mathbf{e}_{j}$ where
$j\in \{1,2,3\}$ is the color of node $i$ in hypergraph $H_{\phi}$.
That is, $p_i=\frac{\deg_{G_\phi}(i)}{m}$ if node $i$ corresponds to a vertex in $G_{\phi}$
and $p_i=\frac{1}{m}$ if node $i$ corresponds to an edge in $G_{\phi}$.
Let $W_1=\sum_{i\in [N]: \chi(\text{node $i$})=1}x_{i}\sqrt{p_{i}}$,
and similarly define $W_2$ and $W_3$. 
Also define $P_1$ (respectively $P_2$ and $P_3$) to be 
$W_1=\sum_{i\in [N]: \chi(\text{node $i$})=1}p_{i}$.
An important observation is that, since $\chi$ is a proper $3$-coloring of $G_\phi$, $P_1 = P_2 = P_3 = 1$.
 
Clearly
\[
Qx=\begin{pmatrix} W_1 \\ W_2 \\ W_3 \end{pmatrix},
\]
and by the Cauchy-Schwarz inequality 
\[
W_1^2 \leq \left(\sum_{i\in [N]: \chi(\text{node $i$})=1}p_{i} \right)
\left( \sum_{i\in [N]: \chi(\text{node $i$})=1}x_{i}^{2}\right)=\sum_{i\in [N]: \chi(\text{node $i$})=1}x_{i}^{2}.
\]
Similar inequalities hold for $W_2$ and $W_3$. Therefore,
\begin{align*}
\lambda_1(K)&= \sup_{x\in \R^N, \text{ }\lVert x \rVert_{2} =1}\lVert Qx \rVert_{2}^{2}=W_1^2+W_2^2+W_3^2 \\
 &\leq \sum_{i\in [N]: \chi(\text{node $i$})=1}x_{i}^{2}+ \sum_{i\in [N]: \chi(\text{node $i$})=2}x_{i}^{2}
+ \sum_{i\in [N]: \chi(\text{node $i$})=3}x_{i}^{2} \\
 &=\sum_{i\in [N]}x_i^2=1.
\end{align*}
Furthermore, it is easy to see that we can choose $q_i$'s such that $q_i$ is proportional to $\sqrt{p_i}$ 
for every $i\in [N]$, thus making the Cauchy-Schwarz inequalities to be equalities. This shows that
the spectral norm of $K$ is indeed one.\footnote{We remark that, for $G_{\phi}$ that is not
$3$-colorable, a marginal kernel constructed from a $3$-coloring may fail to be an optimal DPP kernel for two reasons.
Firstly, such a coloring necessarily make some edges in $G_{\phi}$ monochromatic, and 
hence the corresponding DPP probabilities vanish. 
Secondly, it may incurs some discrepancy among three colors. 
Since $P_1+P_2+P_3=3$ always holds, the discrepancy then implies $\lambda_1(K)=\max\{P_1,P_2,P_3\}>1$, 
and consequently we need to scale down the matrix $K$ by a factor larger than one in order to make it a DPP kernel.
This in turn implies that kernel such constructed can not be an optimal one, as indicated by Theorem~\ref{thm:diagonal}.
}

It is well-known that, if $q_1, \ldots, q_k$ are $k$ vectors in an inner product space 
and form the $k\times k$ Gram matrix $A$ whose $(i, j)$ entry is the inner product between $q_i$ and $q_j$, 
then $\det(A)$ is equal to the square of the volume of the $k$-dimensional parallelepiped spanned by $q_1, \ldots, q_k$.
Now by our construction, for every example $T = (u, v, (u,v)) \in E'(H_{\phi})$,
the $u^{\text{th}}$, $v^{\text{th}}$, and $(u,v)^{\text{th}}$ column vectors of $Q$ are pairwise orthogonal.  
It follows that, since $\mathrm{rank}(K)=3$, 
\[
\Pr_{\bY \sim \calP_K}[\bY = T] = \Pr_{\bY \sim \calP_K}[T \subseteq \bY]
= \det(K_T) = \frac{\deg_{G_\phi}(u)}{m}\frac{\deg_{G_{\phi}}(v)}{m}\frac{1}{m}.
\]
Consequently, 
\begin{align*}
\llhd(K) 
&= -\frac{1}{T'} \sum_{t=1}^{T} \log  \Pr_{\bY \sim \calP_K}[\bY=X_t] \\
&= -\frac{1}{m} \sum_{(u,v)\in E(G_\phi)}\log\left(  \frac{\deg_{G_\phi}(u)}{m}\frac{\deg_{G_{\phi}}(v)}{m}\frac{1}{m} \right) \\
&= 3 \log m - \frac{1}{m} \displaystyle\sum_{(u,v)\in E(G_\phi)} 
\left( \log \deg_{G_\phi}(u) + \log \deg_{G_\phi}(v) \right), 
\end{align*}
which completes the proof of the theorem.
\end{proof}

\section{Good rank-$3$ kernels exist -- Proof of Theorem \ref{thm:dim3}}\label{sec:rank-3}

In this section we prove Theorem \ref{thm:dim3}, which shows that, if the log likelihood value is already close to
optimal, then there exists correspondingly a \emph{rank-$3$} kernel whose log likelihood value is also close (with worse
proximity parameter) to optimal.
\begin{reptheorem}{thm:dim3}[restatement]
Let $G_{\phi}$ be a BOT graph with maximum degree at most $k$ and let $\ell_{\textsc{yes}}(G_{\phi})$ be as defined in Definition~\ref{def:ell-yes}. Then there is a constant $C_k$ depending only on $k$ such that the following holds. Let $K$ be an optimal marginal kernel for $G_{\phi}$ with likelihood $\llhd(K) \leq \ell_{\textsc{yes}}(G_{\phi}) + \delta$ for some $0 < \delta \leq 1/(128k)^2$, then there exists a marginal kernel $K'$ for $G_{\phi}$ of dimension $3$ such that $\ell(K') \leq \ell_{\textsc{yes}}(G_{\phi})+C_{k}\delta^{1/4}$.
\end{reptheorem}

\subsection{Proof overview}\label{sec:dim3-overview}
Without loss of generality, we can assume that the diagonal of $K$ is consistent with Theorem~\ref{thm:diagonal}, 
namely its diagonal entries are 
\[
K_{ii}=
\begin{cases}
\frac{\deg_{G_\phi}(u)}{m} & \mathrm{for \;} i=u \in V(G_\phi) \\
\frac{1}{m} & \mathrm{for \;} i=(u,v) \in E(G_\phi) \\
\end{cases}.
\]
	
By our construction of $G_{\phi}$, we may assume that every element $i \in [m+n]$ occurs in at most $k$ 
training examples in $\calT$ for some absolute constant $k$.
	
Since $K$ is PSD, there exists a matrix $Q$ such that $K = Q^{\top}Q$.  Let $q_1,q_2,\ldots,q_N$ denote the columns of $Q$, 
and for a set $T \subseteq [N]$, define $Q_T = \mathrm{span}(\{q_i : i \in T \})$.
Recall that we need to compare 
\[
\ell_{\textsc{yes}}(G_{\phi})= 3 \log m - \frac{1}{m} \displaystyle\sum_{(u,v)\in E(G_\phi)} 
\left( \log \deg_{G_\phi}(u) + \log \deg_{G_\phi}(v) \right),
\]
which is the value attainable by an embedding with a ``perfect'' vector coloring (in which the coloring vectors of the vertices in each set are pairwise orthogonal), with the actual log-likelihood attained by $K$, $\ell(K)$. Recall that the determinant of each $K_X$ corresponds to the squared-volume of the parallelepiped spanned by the columns of $Q$ in the set $X=\{u,v,(u,v)\}$, where fixing the diagonal of $K$ fixes the lengths of $q_u$, $q_v$, and $q_{(u,v)}$ ($\|q_u\|_{2}^2=\deg_{G_\phi}(u)/m$, $\|q_{(u,v)}\|_{2}^2=1/m$, etc.). The volume is at least the product of the lengths times the area of the parallelogram spanned by the unit vectors $q_u/\|q_u\|_{2}$ and $q_v/\|q_v\|_{2}$, which is $\sin\theta(q_u, q_v)$. (This is tight if $q_{(u,v)}$ is orthogonal to both $q_u$ and $q_v$.) Thus, we have
\begin{align}\label{eqn:ell(K)-orth}
\ell(K)\geq 3\log{m}-\frac{1}{m}\sum_{(u,v)\in E(G_{\phi})}
\left(\log(\deg_{G_{\phi}}(u))+\log(\deg_{G_{\phi}}(v))+\log(\sin^2\theta(q_u, q_v)\right).
\end{align}
Since ``vertex'' $(u,v)$ in the hypergraph $H_{\phi}$ is connected only with vertices $u$ and $v$,
and our goal is to maximize the likelihood of $K$ (hence minimize the log likelihood $\ell(K)$),
from now on we assume that $q_{(u,v)}$ is always taken to be orthogonal to both $q_{u}$ and $q_{v}$, 
for every $(u,v)\in E(G_{\phi})$.

It is more convenient to work with likelihood functions, which are $L(K):=\exp(\ell(K))$ and 
$L_{\textsc{yes}}(K):=\exp(\ell_{\textsc{yes}}(G_{\phi}))$. 
After simple manipulations, we have
\begin{align}\label{ineq:near-optimal-assumption}
\frac{L_{\textsc{yes}}(K)}{L(K)} 
= \frac{\prod_{(u, v) \in E(G_{\phi})} \Pr_{\bY \sim \calP_K}[\bY=\{u, v, (u,v)\}]}
   {\prod_{(u, v) \in E(G_{\phi})}\|q_u\|_{2}^2 \cdot \|q_v\|_2^2 \cdot \|q_{(u,v)}\|_2^2}
=\prod_{(u, v) \in E(G_{\phi})} \sin^2 \theta(q_u, q_v) \geq \exp(-\delta m),
\end{align}
by our assumption that $\llhd(K) \leq \ell_{\textsc{yes}}(G_{\phi}) + \delta$.

That is, the ratio between the likelihood of a marginal kernel $K$ and
the likelihood of a perfect coloring kernel is just the product of $\sin^2\theta(q_u, q_v)$, 
over all edges $(u,v)\in E(G_{\phi})$.

To construct a good dimension-$3$ kernel $K'$ from kernel $K$,
our basic idea is to project each column vector of $Q$ onto a subspace of dimension $3$ (so that
the dimension of the new kernel is at most $3$), and show that the likelihood does not decrease too much. 
However, there are several issues we need to cope with. First, how to find such a subspace?  
Second, there can be pairs of column vectors $q_u$ and $q_v$ of $Q$ such that 
the angle $\theta(u,v)$ between them is small, 
and this angle may become even smaller or even zero upon projection.  
Last, there can be column vectors whose projections onto the subspace have extremely small or even zero lengths.  

Since the log likelihood function of $K$ is close to optimal, there exists an edge $(u,v)\in E(G_{\phi})$
such that $q_u$ and $q_v$ are almost orthogonal (recall that we always take $q_{(u,v)}$ to be orthogonal to 
the plane spanned by $q_u$ and $q_v$) by a simple probabilistic argument. 
The subspace spanned by $\{q_u, q_v, q_{(u,v)}\}$ will be 
the dimension-$3$ subspace onto which we project each column vector of $Q$.
To tackle the other two issues mentioned above, 
we employ a simple counting method and the negative association property of conditional DPPs
to bound the numbers of column vectors involved in these two bad situations,
and apply a greedy algorithm to assign new directions for these vectors while keep their norms unchanged.
This allows us to lower bound the probabilities of seeing these ``bad edges'' under the new DPP kernel.
Finally, we lower bound the likelihood function of the new kernel on the training set 
by providing an upper bound on the scaling factor 
needed to ensure that the new kernel's eigenvalues are bounded by $1$.

\subsection{Finding a good dimension-$3$ subspace}
Let $K=Q^{\top}Q$ be a marginal kernel such that $L(K) \geq \exp(-\delta m) L_{\textsc{yes}}(K)$.
If $\mathrm{rank}(Q)\leq 3$, we are done. So we assume that $\mathrm{rank}(Q)>3$.  
By geometric averaging, there exists a size-$3$ subset $S=\{u, v, (u,v)\}$ such that 
\[
\Pr_{\bY \sim \calP_K}[\bY = S] \geq \exp(-\delta) \prod_{i\in S}\|q_i\|_{2}^{2}.
\]
In the following, we use $\bar{S}:=[N]\setminus S$ to denote the complement of $S$.

Denote the dimension-$3$ subspace spanned by the embedding vectors of these three elements by $V$, namely
\[
V=\mathrm{span}\{q_u, q_v, q_{(u,v)}\}.
\]
Then $V$ will the subspace onto which we project all column vectors of $Q$.
More importantly, we will show that the average length orthogonal to $V$ of the column vectors of $Q$ is small.

\begin{lemma}\label{lem:sum_projection_weights}
Let $V$ be a dimension-$3$ subspace as constructed above, then 
\[
\sum_{i \in \bar{S}} \| \proj_{V^{\perp}} q_i \|_2^2 \leq \delta,	
\]
where $V^{\perp}$ denotes the dimension $\mathrm{rank}(Q)-3$ subspace orthogonal to $V$.
\end{lemma}

\begin{proof}
Let $\bY$ be a random set distributed according to the DPP kernel $K$, and 
let $\bY_{S}$ be a random set which is distributed according kernel $K$, 
conditioned on  $S$ occurring in $\bY_{S}$.
It is well-known that conditioning a DPP on the event that all of the elements in a fixed set are observed gives rise to
another DPP distribution (over the ground set $\bar{S}$) (see, e.g. \cite{Kul12}). 
Therefore if we define $\bZ:=\bY_{S}\setminus S$ as the random set 
containing elements outside $S$ which appear in $\bY_{S}$,
then $\bZ$ is distributed as a DPP.\footnote{
In fact, $\bZ$ is distributed as a DPP with marginal kernel 
$K-K_{\bar{S} S} K_S^{-1} K_{S\bar{S}}$~\cite{BR05}, 
where $K_{AB}$ denotes the submatrix of $K$ consisting of the rows in $A$ and columns in $B$. 
However, the actual form of the DPP kernel is unimportant for our argument here.} 

By our choice of subset $S$, 
\begin{align*}
& \qquad \exp(-\delta) \prod_{i\in S}\|q_i\|_{2}^{2} \leq \Pr_{\bY \sim \calP_K}[\bY = S] \\ 
&=\Pr_{\bY \sim \calP_K}[S \subseteq \bY] \cdot \Pr[\bY_{S} = S] 
   = \Pr_{\bY \sim \calP_K}[S \subseteq \bY] \cdot \Pr[\bZ=\emptyset] \\
&=\det(K_{S}) \cdot \Pr[\bZ=\emptyset]. 
\end{align*}
Since $\det(K_{S})\leq \prod_{i\in S}\|q_i\|_{2}^{2}$, it follows that $\Pr[\bZ = \emptyset] \geq \exp(-\delta)$.

Recall that a set of random variables $\bX_1,\bX_2,\ldots,\bX_N$ is said to be \emph{negatively associated} (NA) 
if for any two disjoint index subsets $S, T \subset [N]$ and any two functions $f, g$ 
that depend only on variables in subsets $S$ and $T$ respectively, and are either 
both monotone increasing or both monotone decreasing, we have 
\[
\EX\left[f(\bX_i: i \in S)\cdot g(\bX_j: j \in T) \right] \leq 
\EX\left[f(\bX_i: i \in S)\right] \cdot \EX\left[ g(\bX_j: j \in T) \right].
\]
Lyons~\cite{Lyo03} proved that the indicator random variables of elements in the set are negatively associated
if the probability distribution on the set is a DPP.

\begin{claim}	\label{claim:na-expectation}
If $\bX_1,\bX_2,\ldots,\bX_N$ are negatively associated $\{0,1\}$-valued random variables such that 
$\Pr[\bigwedge_{i=1}^N (\bX_i = 0)] \geq \exp(-\alpha)$, then $\EX[\sum_{i=1}^{N} \bX_i] \leq \alpha$.
\end{claim} 
\begin{proof}
First of all, define $f_i(\bX_i):=\mathbbm{1}{(\bX_i = 0)}$ for every $i \in [N]$. Note that
for disjoint subsets $S$ and $T$, $f:=\prod_{i\in S}f_i$ and $g:=\prod_{i\in T}f_i$
depend on disjoint subsets of variables and are both monotone decreasing. Therefore, inductively applying the 
negative association property of $\bX_1,\bX_2,\ldots,\bX_N$, 
we have $\EX\left[\prod_{i\in [N]}f_i(\bX_i)\right] \leq \prod_{i\in [N]}\left(\EX[f_i(\bX_i)]\right)$, or equivalently
\[
\Pr[\bigwedge_{i=1}^N (\bX_i = 0)] \leq \prod_{i=1}^N \Pr[\bX_i = 0].
\]
Therefore, 
\begin{align*}
\exp(-\alpha) &
\leq \Pr[\bigwedge_{i=1}^N (\bX_i = 0)] \leq \prod_{i=1}^N \Pr[\bX_i = 0] \\
&= \prod_{i=1}^N (1 - \Pr[\bX_i = 1]) \leq \prod_{i=1}^N \exp(-\Pr[\bX_i = 1]) \\
&= \exp(-\sum_{i=1}^{N} \Pr[\bX_i = 1]) = \exp(-\EX[\sum_{i=1}^{N} \bX_i]).
\end{align*}
Hence the claim follows.  
\end{proof}
	
\begin{claim}\label{claim:z-expectation}
Let $\bZ$ be the random set containing elements outside $S$ which appear in $\bY_{S}$ as before, 
then $\EX[|\bZ|] \leq \delta$.
\end{claim}
	
\begin{proof}
Define $\bX_i$ to be the $\{0,1\}$-indicator random variable corresponding to the event 
$i \in \bZ$ for each $i \in \bar{S}$.  
Then applying Claim~\ref{claim:na-expectation} to these variables, 
and observing that $|\bZ| = \sum_{i \in \bar{S}} \bX_i$ proves the claim.
\end{proof}
	
Now we have
\[
\EX[|\bZ|] = \sum_{i \in \bar{S}} \Pr[i \in \bZ] = 
\sum_{i \in \bar{S}} \frac{\Pr_{\bY \sim \calP_K}[S \cup \{i\} \subseteq \bY]}{\Pr_{\bY \sim \calP_K}[S \subseteq \bY]} =
\sum_{i \in \bar{S}} \frac{\det(K_{S \cup \{i\}})}{\det(K_S)},
\] 
and interpreting the determinants as squared volumes of parallelepipeds, we have
\[
\sum_{i \in \bar{S}} \frac{\det(K_{S \cup \{i\}})}{\det(K_S)} = 
\sum_{i \in \bar{S}} \| \proj_{V^{\perp}} q_i \|_2^2
\]
and applying Claim~\ref{claim:z-expectation} completes the proof.
\end{proof}

\subsection{Bounding the number of ``bad'' vertices}
Let $\eps_0<\frac18$ be a fixed constant. Define
\[
B_e = \{ (u, v) \in E(G_{\phi}) : \sin^2\theta(q_u, q_v) < \eps_0 \}.
\]
	
\begin{lemma}\label{lem:b0-small}
The size of $B_e$ can be upper bounded as $|B_e| \leq \delta m/2$.
\end{lemma}

\begin{proof}
If we rewrite \eqref{ineq:near-optimal-assumption} as
\[
\exp(-\delta m) L_{\textsc{yes}}(K) \leq L(K) \leq  L_{\textsc{yes}}(K) \prod_{(u, v) \in E(G_{\phi})} \sin^2\theta(q_u, q_v),	
\]
then
\[
\exp(-\delta m)  \leq \prod_{(u, v) \in E(G_{\phi})} \sin^2\theta(q_u, q_v) < {\eps_0}^{|B_e|}.
\]
Taking logarithms on both sides yields $- \delta m \leq  |B_e| (\log \eps_0)\leq  |B_e| \log (1/8)$, 
so $|B_e| \leq \dfrac{\delta m}{ \log{8}} \leq \delta m/2$.
\end{proof}

Define $B_1$ to be the set of ``bad'' vertices of the first kind;
that is, the set of vertices whose embedding vectors that are too close to some of their 
neighboring (in graph $G_{\phi}$) embedding vectors:
\[
B_1 = \{ u \in V(G_{\phi}) : \exists v \mathrm{\; such \; that \;} (u,v)\in E(G_{\phi}) \mathrm{\; and \;} \sin^2\theta(q_u, q_v) < \eps_0 \}.
\]
	
Then clearly, 
\begin{align}\label{ineq:bound-on-B1}
|B_1| \leq 2|B_e| \leq \delta m.
\end{align}

\subsection{Assigning projections for ``bad vectors''}	
Define $B_2$ to be the set of ``bad'' vertices of the second kind;
these are vertices whose embedding vectors' (relative) norms along $V^{\perp}$ are not small (consequently,
the norms of the projections onto $V$ of such vectors  are not large enough):  
\[
B_{2} = \{i \in [N] : \| \proj_{V^{\perp}} q_i \|_2^2 \geq \sqrt{\delta} \|q_i\|_2^2 \}.
\]

Let $B:=B_1 \cup B_2$ be the set of all ``bad'' vertices (correspondingly, ``bad'' embedding vectors). 
We will call a vertex (and its corresponding embedding vector) ``good'' if it is not ``bad''.

First of all, since $\|q_i\|_2^2 \geq \frac{1}{m}$ for each $i\in [N]$, by Lemma~\ref{lem:sum_projection_weights},
the size of $B_2$ can be bounded by
\begin{align}\label{ineq:bound-on-B2}
|B_2| \leq \sqrt{\delta} m.
\end{align}
Combined with \eqref{ineq:bound-on-B1}, this gives
\begin{align}\label{ineq:bound-on-B}
|B| \leq 2\sqrt{\delta} m.
\end{align}

For every ``good'' vertex $u \notin B$, let 

\begin{tcolorbox}[
    title=Assigning embedding vectors for ``good'' vertices,   
    halign title=center,        
    colback=white,              
    colframe=blue!75!black,              
    width=0.7\textwidth,        
    center,                     
    sharp corners,               
	toptitle=2mm,    
    bottomtitle=2mm, 
    top=-2mm,         
    bottom=5mm       
]
\begin{align*}
q_u' \leftarrow \proj_{V} q_u 
\end{align*}
\end{tcolorbox}

What about ``bad'' vertices? For each $v\in B$, we will use the greedy algorithm outlined in 
Fig.~\ref{assign_projections} to find a good direction $\mathbf{z}$ 
(as the direction of its new embedding vector $q'_v$) so that the angle between $\mathbf{z}$ and 
the new embedding vector of any neighbor (in graph $G_{\phi}$) of $v$  is not too small. 
Note that in the description of the algorithm, we assume that the the projection subspace $V=\R^3$
and use $\theta(x,y)$ to denote the angle between any two vectors $x$ and $y$ in $\R^3$.

\begin{claim}\label{claim:bound-tau-k}
There is an absolute constant $\tau_k$ depending on $k$ only such that, for all $(u,v)\in E(G_{\phi})$
with at least one of $u$ and $v$ ``bad'', the angle between their new embedding vectors satisfies that
$\sin^2 \theta(q'_u, q'_v)\geq \tau_k$.
\end{claim}	
\begin{proof}
Recall that every element $u \in V(G_{\phi})$ has at most $k$ neighbors in $V(G_{\phi})$.
Then we can take $\tau_k$ to be the solution to the following optimization problem 
\[
\displaystyle\min_{x_1,x_2,\ldots,x_k \in S^2} \max_{y \in S^2} \min_i \sin^2 \theta(x_i,y),
\]
Observe that $\tau_k$ is an absolute constant between $0$ and $1$.
Indeed, one can obtain a simple lower bound on $\tau_k$ as follows: 
let $r_k$ be the maximum radius of a spherical cap such that $k+1$ identical such caps
can be placed disjointly on the surface of a unit sphere. 
Now for any configuration of $x_1,x_2,\ldots,x_k \in S^2$, at least one cap will contain no $x_i$.
Letting $y$ be the center of such a cap is a certificate that $\tau_k \geq r_{k}^2$. 
\end{proof}

\begin{figure}[htbp] 
    \centering
    \begin{tcolorbox}[
        enhanced,
        title=Assigning embedding vectors for ``bad'' vertices,
        halign title=center,        
		colback=white,              
		colframe=blue!75!black,              
		width=0.8\textwidth,        
		center,                     
		sharp corners,               
		toptitle=2mm,    
		bottomtitle=2mm, 
		top=0mm,         
		bottom=5mm       
    ]
 \begin{tabbing}
    \quad \= \quad \= \quad \= \quad \= \quad \kill 
    1. \> order all $n$ vertices arbitrarily \\
    2. \> \textbf{for each} $v \in V(G_{\phi})$ \textbf{do} \\
    \> \> \textbf{if} $v$ is ``good'' \\
    \> \> \> \textbf{continue} (that is, $q'_v \leftarrow \proj_{V}q_v$) \\
    \> \> \textbf{else} ($v$ is ``bad'') \\
    \> \> \> let $\mathcal{N}$ be the set of neighboring vertices of $v$ appearing before $v$ \\
    \> \> \> let $\mathbf{z} = \text{argmax}_{y \in S^2} \min_{u \in \mathcal{N}} \sin^2 \theta(u,y)$ \\
    \> \> \> \textit{\small (any such maximal direction works)} \\
    \> \> \> assign $q'_v \leftarrow \|q_v\|_{2}\mathbf{z}$
\end{tabbing}
    \end{tcolorbox}
    \caption{Algorithm for computing the new embedding vector for ``bad'' vertices.} \label{assign_projections}
\end{figure}

Once we set new embedding vectors for every vertex $v\in V(G_{\phi})$, let 
\[
Q'=\begin{pmatrix} 
\vertbar & \vertbar &        & \vertbar  \\
 q'_1     &     q'_2  & \cdots &  q'_N  \\
 \vertbar & \vertbar &        & \vertbar 
\end{pmatrix},
\]
and set 
\[
K' = \beta(Q')^{\top}Q',
\]
where $\beta$ is a scaling factor to be determined later to make the eigenvalues of $K'$ at most $1$.

\subsection{Bounding the likelihood value of the new kernel}	
To lower bound the likelihood value of the new kernel $K'$ on the training set, 
we compare the probabilities of kernels $K$ and $K'$ seeing an arbitrary subset $T:=\{u,v, (u,v)\}$ in the training set.  
We distinguish between two cases.

\paragraph{Both $u$ and $v$ are ``good'' vertices.}
Since $u$ is ``good'', the angle $\theta'_u$ between $q_u$ and its projection $q_u'$ is small.
Indeed, since $\|q_u'\|_2^2 > (1-\sqrt{\delta}) \| q_u \|_2^2$, we have $\sin{\theta'_u} < \delta^{1/4}$,
or $\theta'_u <\sin^{-1}\delta^{1/4} \leq \frac{2}{\pi}\delta^{1/4} < \delta^{1/4}$. 
It follows that, if both $u$ and $v$ are ``good'' vertices, 
then the angle $\theta(q_u', q_v')$ between their projected vectors satisfies
\begin{align}\label{ineq:bound-on-new-angle}
\theta(q_u', q_v')>\eps_{0}-2\delta^{1/4},
\end{align}
by the triangle inequality.

Therefore, we have
\allowdisplaybreaks
\begin{align*}
& \; \quad \frac{\Pr_{\bY' \sim \calP_{K'}}[\bY' = \{u, v, (u, v)\}]}{\Pr_{\bY \sim \calP_K}[\bY = \{u, v, (u, v)\}]}	
 \geq \frac{\Pr_{\bY' \sim \calP_{K'}}[\bY' = \{u, v, (u, v)\}]}{\Pr_{\bY \sim \calP_K}[\{u, v, (u, v)\}\subseteq \bY]} \\
& = \frac{\det(K'_{T})}{\det(K_T)} \qquad \text{(since $K'$ is a dimension-$3$ DPP marginal kernel)} \\
& = \beta^3 \dfrac{\|q'_u\|^{2}_{2} \|q'_v\|^{2}_{2} \sin^2 \theta(q_u', q_v')}
    {\|q_u\|^{2}_{2}\|q_v\|^{2}_{2}\sin^2\theta(q_u, q_v)} \\
& \geq \beta^3 (1-\sqrt{\delta})^{2} \dfrac{\sin^2(\eps_0-2\delta^{1/4})}{\sin^2 \eps_0}  
   \qquad \text{($\|q'_i\|^{2}_{2}/\|q_i\|^{2}_{2}>1-\sqrt{\delta}$ 
   for ``good'' embedding vectors and \eqref{ineq:bound-on-new-angle})}\\
& \geq \beta^3 (1-\sqrt{\delta})^2 \left(1-\frac{2\delta^{1/4}}{\sin^2 \eps_0}\right) 
\qquad \text{(using the fact that $|\frac{\mathrm{d}}{\mathrm{d}x}(\sin^2 x)| \leq 1$ for all $x$)}\\
& \geq  \beta^3 (1-\sqrt{\delta}) \left(1-\frac{3\delta^{1/4}}{\sin^2 \eps_0}\right),
\end{align*}
where the factor $\beta^3$ comes from the scaling factor $\beta$ in our definition of $K'$, 
and the fact that these probabilities are determinants of $3 \times 3$ principal submtrices of $\beta (Q')^{\top}Q'$.

\paragraph{At least one of $u$ and $v$ is ``bad''.}
Using Claim~\ref{claim:bound-tau-k}, we have for such edges 
\begin{align*}
& \; \quad \frac{\Pr_{\bY' \sim \calP_{K'}}[\bY' = \{u, v, (u, v)\}]}{\Pr_{\bY \sim \calP_K}[\bY = \{u, v, (u, v)\}]}	
   \geq \frac{\Pr_{\bY' \sim \calP_{K'}}[\bY' = \{u, v, (u, v)\}]}{\Pr_{\bY \sim \calP_K}[\{u, v, (u, v)\}\subseteq \bY]} 
   = \frac{\det(K'_{T})}{\det(K_T)} \\
& = \beta^3 \dfrac{\|q'_u\|^{2}_{2} \|q'_v\|^{2}_{2} \sin^2 \theta(q_u', q_v')}
    {\|q_u\|^{2}_{2}\|q_v\|^{2}_{2}\sin^2\theta(q_u, q_v)} \\
&\geq \beta^3 (1-\sqrt{\delta})\tau_k,
\end{align*}
since at least one of $u$ and $v$ is a ``bad'' vertex, and $\|q'_i\|^{2}_{2}/\|q_i\|^{2}_{2}=1$
for such vertices.

\paragraph{Putting the two cases together.}	
\allowdisplaybreaks
\begin{align*}
& \; \quad \frac{L(K')}{L(K)} = \frac{\prod_{(u,v)\in E(G_{\phi})}\Pr_{\bY' \sim \calP_{K'}}[\bY' = \{u, v, (u, v)\}]}
{\prod_{(u,v)\in E(G_{\phi})}\Pr_{\bY \sim \calP_{K}}[\bY = \{u, v, (u, v)\}]} \nonumber\\
&\geq \displaystyle\prod_{\substack{(u, v) \in E(G_{\phi}) \\ u \notin B \text{ and } v \notin B}} 
      \beta^3 (1-\sqrt{\delta})\left(1-\dfrac{3\delta^{1/4}}{\sin^2\eps_0}\right) 
	\displaystyle\prod_{\substack{(u, v) \in E(G_{\phi}) \\ u \in B \text{ or } v \in B }} 
	\beta^3 (1-\sqrt{\delta})\tau_k\nonumber\\
&\geq \beta^{3m}(1-\sqrt{\delta})^{m}\left(1-\dfrac{3\delta^{1/4}}{\sin^2 \eps_0}\right)^{m-k|B|} \tau_k^{k|B|} \nonumber\\
&\geq \beta^{3m}(1-\sqrt{\delta})^{m}\left(1-\dfrac{3\delta^{1/4}}{\sin^2\eps_0}\right)^{(1-2k\sqrt{\delta})m} 
   \tau_k^{2k\sqrt{\delta} m} \nonumber\\
&\geq \beta^{3m} \exp\left[-\left(\frac{3\delta}{2}+ \dfrac{3\delta^{1/4}}{\sin^2\eps_0}(1-2k\sqrt{\delta})
    +2k\log(\frac{1}{\tau_{k}})\sqrt{\delta}\right)m \right] \\  
&\geq \beta^{3m} \exp\left(-C''_{k} \delta^{1/4} m\right),	
\end{align*}
where $C''_{k}$ is some constant depending only on $k$ and in the second last step we use the inequality
$1-x\geq \exp(-3x/2)$ for all $0\leq 1/2 \leq x$.
It follows that 
\begin{align}\label{ineq:likelihood-bound}
\quad \frac{L(K')}{L_{\textsc{yes}}(K)}\geq \beta^{3m}\exp\left(-C''_{k} \delta^{1/4} m\right)\cdot \exp(-\delta m)
\geq \beta^{3m} \exp\left(-C'_{k} \delta^{1/4} m\right),
\end{align}	
for some absolute constant $C'_{k}$ depending on $k$ only.

\subsection{Bounding the scaling factor $\beta$}	
We finish the proof by providing a lower bound on $\beta$.  
First, we consider $K'' = (Q'')^{\top}Q''$, where the $i^{\text{th}}$ column of $Q''$ is $\proj_{V} q_i$ for
\emph{every} $i \in [N]$.  Recall that the spectral norm of a matrix $A$, 
$\|A\|_2$, is the largest singular value of $A$, 
we have $\|Q''\|_2 \leq \|Q\|_2 \leq 1$, 
as $Q''$ is a projection onto a subspace spanned by a subset of its columns.  
By our bound on the number of ``bad'' embedding vectors \eqref{ineq:bound-on-B},
at most $2\sqrt{\delta} m$ of the columns of $Q'$ are not the same projection as in $Q''$.  
Each column is replaced by a vector of the same length, which is at most $\sqrt{k/m}$. 
We use the Frobenius norm of $Q''-Q'$ to upper bound its spectral norm as follows.
For every ``bad'' vertex $i$,
\[
\|q_i-q'_i\|_2^{2}=\|q_i\|_2^{2}+\|q'_i\|_2^{2}-2\langle q_i, q'_i \rangle \leq 4\|q_i\|_2^{2} \leq 4k/m.
\]
Therefore,
\[
\|Q''-Q'\|_{2}^2 \leq \|Q''-Q'\|^2_F=\sum_{i=1}^{N}\|q_i-q'_i\|_2^{2} 
=\sum_{i\in B}\|q_i-q'_i\|_2^{2}
\leq 2\sqrt{\delta} m \cdot 4k/m 
=8k\sqrt{\delta}.
\]
Now, by the triangle inequality, we have
\[
\|Q'\|_2 \leq \|Q''\|_2 + \|Q''-Q'\|_2 \leq 1 + \sqrt{8k \sqrt{\delta}}. 
\]  
In order for $K' = \beta Q^{\top}Q$ to be the marginal kernel of a DPP, 
we need to ensure that $\|K'\|_2 \leq 1$.  As $\|K'\|_2 = |\beta| \|Q^{\top}Q\|_2 = |\beta| \|Q\|_2^2$, 
we can take $\beta = \frac{1}{(1+\sqrt{8k\sqrt{\delta}})^2} \geq 1-2\sqrt{8k\sqrt{\delta}}
=1-\sqrt{32k\sqrt{\delta}}\geq \exp\left(-\sqrt{72k}\delta^{1/4}\right)$.  

Finally, plugging the bound $\beta \geq \exp\left(-\sqrt{72k}\delta^{1/4}\right)$ 
into \eqref{ineq:likelihood-bound} to get that
\[
\quad \frac{L(K')}{L_{\textsc{yes}}(K)}\geq \beta^{3m} \exp\left(-C'_{k} \delta^{1/4} m\right)
\geq \exp\left(-\sqrt{648k}\delta^{1/4}m\right)\cdot \exp\left(-C'_{k} \delta^{1/4} m\right)
=\exp\left(-C_{k} \delta^{1/4} m\right),
\]
where $C_{k}=C'_{k}+\sqrt{648k}$.
Or equivalently, $\ell(K') \leq \ell_{\textsc{yes}}(G_{\phi})+C_{k}\delta^{1/4}$ for some constant $C_k$ which depends only on $k$.
This completes the proof of Theorem~\ref{thm:dim3}.

\section{Putting it all together -- Proof of Theorem~\ref{thm:soundness}}\label{sec:soundness}

Now it is time to assemble the parts we have built so far and prove the following soundness theorem,
which basically says that there must be a gap between the log likelihood functions of the training set 
derived from the YES instance BOT graph and that from the NO instance BOT graph. 
\begin{reptheorem}{thm:soundness}[Soundness theorem, restatement]
Let $\ell_{\textsc{yes}}(G_{\phi})$ be as defined in Definition~\ref{def:ell-yes}, namely, the optimal log-likelihood assuming that $\phi$ were a YES instance (cf.~Theorem~\ref{thm:completeness}). Let $k$ and $\eps'$ be as those defined in Lemma~\ref{lem:BOT-graph} and let $n=|V(G_{\phi})|$ be the number of vertices in the BOT graph. Then there exists a constant $C=C(k,\eps')$ such that the following holds: If there exists a DPP marginal kernel $K$ for $G_{\phi}$ of rank $3$ satisfying 
\[
\ell(K)\leq \ell_{\textsc{yes}}(G_{\phi})+\frac{C}{\log^2{n}},
\]
then $G_{\phi}$ is $\eps'$-close to $3$-colorable. That is, there exists a set $E' \subset E(G_{\phi})$ with $|E'| \leq \eps'|E(G_{\phi})|$ such that $G_{\phi}\setminus E'$ is $3$-colorable. 
\end{reptheorem}

In conventional graph coloring, a mapping $\chi: V \to \{c_1, \ldots, c_k\}$ \emph{satisfies} an edge $(u, v)$ if $\chi(u) \neq \chi(v)$. Vector coloring generalizes this constraint by requiring the vectors assigned to $u$ and $v$ to be (nearly) orthogonal. While Lemma~\ref{lem:BOT-graph} establishes that an ``almost perfect'' $3$-coloring of $G_{\phi}$ can decode a truth assignment satisfying most clauses in $\phi$, our reduction for maximum likelihood DPP learning relies on vertex embeddings. Consequently, we must determine if a truth assignment can be similarly recovered from an \emph{almost perfect $3$-vector-coloring}. We observe that the equality and clause gadgets from~\cite{BOT02} are ``robust'' in the following sense: a truth assignment satisfying most clauses can indeed be decoded if the vector coloring is sufficiently close to perfect. We formalize this using several simple facts of spherical geometry detailed in the following.

From now on, for any vector $v\in \R^3$, we use $(v_1, v_2, v_3)$ to denote its Cartesian coordinates.
\begin{claim} \label{claim:equality-gadget-dot-product}
Let $a,b,c,d \in S^2$ be unit vectors in $\R^3$ and let $0\leq t \leq 1/5$. If 
$|\ip{a}{b}|\leq t$, $|\ip{b}{c}|\leq t$, $|\ip{c}{a}|\leq t$, 
$|\ip{d}{b}|\leq t$, and $|\ip{d}{c}|\leq t$ then $|\ip{a}{d}|\geq 1-5t^2$.
\end{claim}
Intuitively, if vectors $a, b$ and $c$ in $S^2$ are close to forming an orthonormal basis,
and a fourth vector $d\in S^2$ is almost orthogonal to both $b$ and $c$, then $d$ is necessarily close to $a$.

\begin{proof}
Without loss of generality, assume that $b=(0,0,1)$. Then $0\leq |a_3|, |c_3|, |d_3| \leq t$.
Let $a', c'$ and $d'$ be the projection to the $X$-$Y$ plane of $a, c$ and $d$, respectively.
Note that $\lVert a' \rVert, \lVert c' \rVert, \lVert d' \rVert \geq \sqrt{1-t^2}$,
and since $|\ip{a}{c}|\geq \ip{a'}{c'}-|a_3||c_3|$, we have $|\ip{a'}{c'}|\leq t+t^2$.
Similarly, $|\ip{d'}{c'}|\leq t+t^2$.

Without loss of generality, assume that $c'=(0, \lVert c' \rVert)$.
Then, because $|\ip{a'}{c'}|\leq t+t^2$ and $\lVert c' \rVert \geq \sqrt{1-t^2}$, we have 
$|a_2|\leq \frac{t+t^2}{\sqrt{1-t^2}}$, and consequently 
$|a_1^2|\geq 1-t^2-\left(\frac{t+t^2}{\sqrt{1-t^2}}\right)^2$.
Similar bounds hold for $d$. Now
\begin{align*}
|\ip{a}{d}| &\geq |\ip{a}{d}|-|a_3||d_3| \\
            &\geq |a_1||d_1|-|a_2||d_2|-t^2 \\
			&\geq 1-t^2-\left(\frac{t+t^2}{\sqrt{1-t^2}}\right)^2 - \frac{t^2(1+t)^2}{1-t^2}-t^2\\
			&=1-2t^2-2t^2\frac{1+t}{1-t}\\
			&\geq1-5t^2,
\end{align*}
where the last step we use $\frac{1+t}{1-t}\leq 3/2$ for $0\leq t\leq 1/5$.
\end{proof}

\begin{corollary}\label{cor:equality-gadget-angles}
Let $a,b,c,d$ be unit vectors in $S^2$ satisfy that 
$\frac{\pi}{2} -\theta \leq  
\theta(a, b),\theta(a, c),\theta(b, c),\theta(b, d),\theta(c, d) \leq \pi/2$ 
with $0\leq \theta \leq \theta_0=\sin^{-1}(1/5)$. 
Then $\theta(a, d) \leq 3\theta$.
\end{corollary}
\begin{proof}
Let $(x,y) \in \{(a,b),(a,c),(b,c),(b,d),(c,d)\}$ and $t=\sin\theta$. 
Then $|\ip{x}{y}| = \cos \theta(x, y) \leq \cos(\pi/2-\theta) =\sin\theta=t$, 
so by Claim~\ref{claim:equality-gadget-dot-product}, 
$|\ip{a}{d}|\geq 1 - 5 \sin^2 \theta$.  
It follows that $\theta(a, d) \leq \cos^{-1} (1 - 5\sin^2 \theta) \leq 3\theta$,
because $\sin^2(3\alpha)\geq 5\sin^2\alpha$ for all $0\leq \alpha \leq \pi/7$
and $\theta_0 < \pi/7$.
\end{proof}

The following claim facilitates transforming between 
the angle between a pair of vectors in $S^2$ and their inner product,
especially for the cases when the two vectors are almost orthogonal or very close to each other.
\begin{claim}\label{claim:sin_to_angle}
We have the following inequalities between inner product of unit vectors and the angle between them.
\begin{enumerate}
\item Suppose $0\leq \theta \leq \pi/2$ and $\sin^{2}\theta = 1-\eps$ for some $0\leq \eps <1$.
Let $\alpha:=\pi/2-\theta$ be the angular distance between $\theta$ and $\pi/2$. 
Then $\sqrt{\eps} \leq \alpha \leq \frac{\pi}{2}\sqrt{\eps}$.
\item If $u,v\in S^2$ and $|\ip{u}{v}|=\eps$ (i.e. $u$ and $v$ are close to orthogonal), and 
let $\alpha:=\pi/2-\theta(u,v)$ be the angular distance between $\theta(u,v)$ and $\pi/2$, 
then $\eps \leq \alpha \leq \frac{\pi}{2}\eps$.
\item If $u,v\in S^2$ and $|\ip{u}{v}|=1-\eps$ (i.e. $u$ and $v$ are $\eps$-close), 
then $\sqrt{2\eps} \leq \theta(u,v) \leq \frac{\pi}{\sqrt{2}}\sqrt{\eps}$.
\end{enumerate}
\end{claim}
\begin{proof}
Item $1$ follows from the equality $\sin^2{\alpha}=\eps$ and the fact that, 
for $\alpha \in [0,\pi/2]$, $\frac{2}{\pi}\alpha \leq \sin{\alpha} \leq \alpha$.
Item $2$ follows directly from the above inequality. For Item $3$,
since $\sin{\theta(u,v)}=1-\eps$, $\sin^2\frac{\theta(u,v)}{2}=\eps^2/2$. 
Now apply Item $1$ to $\frac{\theta(u,v)}{2}$.
\end{proof}

\begin{remark}\label{rem:angles}
Note that our $3$-vector-coloring only specifies the inner product between any pair of vertex vectors
$\chi^{\top}(u)\chi(v)$, and for the purpose of proving hardness of DPP learning, 
all we need is that $\sin^{2}\theta(\chi(u), \chi(v))=1-|\chi^{\top}(u)\chi(v)|^2$.
This information does not uniquely determine all the vertex vectors on $S^2$ though, 
even if we fix certain vertex vector to be $(1,0,0)$, say. This is because, for a fixed vector $\chi(v)$, 
$\chi(u)$ and $-\chi(u)$ give rise to the same value of $\sin^{2}\theta(\chi(u), \chi(v))$. 
Therefore, from now on, we will not distinguish between $\chi(u)$ and $-\chi(u)$, 
and always assume $\theta(\chi(u), \chi(v))$ to be in $[0, \pi/2]$.
\end{remark}

Now we can show that the equality gadget in Fig.~\ref{fig:equality_gadget} is robust against small noise.

\begin{lemma}[The equality gadget is robust]\label{lem:equality_robust}
Let $\chi: V(G_{\phi}) \to S^2$ be a $3$-vector-coloring which satisfies that $\theta(\chi(u), \chi(v))\geq \pi/2-\delta$
for every $(u,v)\in E(G_{\phi})$. Then, for the the equality gadget in Fig.~\ref{fig:equality_gadget}, 
the angle between node $x$ and node $y$ is at most $3\delta$. 
\end{lemma}
\begin{proof}
This follows directly by plugging $a=\chi(x)$, $b=\chi(y)$, $c=\chi(u)$ and $d=\chi(v)$ into 
Corollary~\ref{cor:equality-gadget-angles}.
\end{proof}

The case of clause gadget requires a bit more work.
Note that we make no attempt to optimize the robustness parameters in the following statements, but
focus on arguing qualitatively that the clause gadgets are robust against small noise.

By abuse of notation, in the following, we shall write the name of a node $u$ in the gadget to also represent the
coloring vector $\chi(u)$ of the node $u$, wherever there is no risk of confusion.

\begin{lemma}[The clause gadget is robust]\label{lem:clause_robust}
Let $\chi: V(G_{\phi}) \to S^2$ be a $3$-vector-coloring which satisfies that $\theta(u,v)\geq \pi/2-\delta_e$
for every $(u,v)\in E(G_{\phi})$, where $\delta_e=o(1)$ for our BOT hypergraph setting.
Further, assume that for any pair of two \textsc{True} nodes in the clause gadgets as shown in Fig.~\ref{fig:clause_gadget}
or for any pair of two \textsc{Dummy} nodes in the literal blocks as shown in Fig.~\ref{fig:literal_block},
the angle between the two color vectors $u$ and $v$ satisfy that $\theta(u,v)\leq \delta_p$,
where $\delta_p=\pi/300$ is some small absolute constant. 
Consider the clause gadget as shown in Fig.~\ref {fig:clause_gadget}. 
Then at least one of the three literals must be ``assigned'' value \textsc{True} in the sense 
that its color vector is close to that of its corresponding \textsc{True} node:
there exists at least one $i \in \{a, b, c\}$ such that $|\ip{\ell_i}{T_i}| > t_0$,
where $t_0=0.98$ (this corresponds to $\theta(\ell_i, T_i)<\pi/15$). 
\end{lemma}

We first need to argue about the coloring vectors of a simplified clause gadget as shown in Fig.~\ref{fig:proof_gadget}.
\begin{figure}
    \centering
        \begin{tikzpicture}[auto, thick]
		\node[g_node] (T)  at (0,0) {$T$} ;
		\node[g_node] (l)  at (2,2) {$\ell$} ;
		\node[w_node] (u)  at (2,0) {$u$} ;
		\node[w_node] (D)  at (0,2) {$D$} ;

		\foreach \source/\dest in {T/u, T/D, l/u, l/D}
			\path (\source) edge (\dest);
		\end{tikzpicture}
\caption{The simplified clause gadget.}
\label{fig:proof_gadget}
\end{figure}
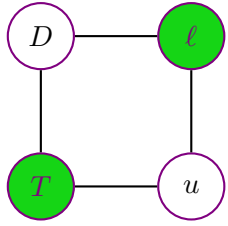

\begin{claim}\label{claim:simple_clause}
Let $D,T,u,\ell \in S^2$ be unit vectors in $\R^3$ and $t_0=0.98$ as in Lemma~\ref{lem:clause_robust}. 
Suppose that
$|\ip{D}{T}|\leq \eps$, $|\ip{D}{\ell}|\leq \eps$, $|\ip{u}{T}|\leq \eps$, $|\ip{u}{\ell}|\leq \eps$.
If $|\ip{T}{\ell}|\leq t_0$, then $|\ip{D}{u}|\geq 1-202\eps^2$.
\end{claim}
Basically the Claim asserts that if the $3$-vector-coloring satisfies all four edges in the gadget with closeness parameter
$1-\eps$,
then if vectors $T$ and $\ell$ are not close to each other, this would 
force vectors $D$ and $u$ to be aligned along almost the same direction --- or equivalently, be assigned to
the same ``color'' (note that if $T$ and $\ell$ are very close, then $u$ would have the freedom to choose
between the color of $D$ or the color that is ``orthogonal'' to both $T$ and $D$; namely, the color of 
\textsc{False} nodes).

\begin{proof}
For convenience, we take the plane spanned by $T$ and $\ell$ to be the $X$-$Y$ plane and in particular 
set $T=(1,0,0)$ and $t:=|\ip{T}{\ell}|\leq t_0$. Then without loss of generality, let 
$\ell=(t,\sqrt{1-t^2},0)$.
Let $u=(u_1, u_2, u_3)$. From $|\ip{u}{T}|\leq \eps$ and $|\ip{u}{\ell}|\leq \eps$,
we deduce that $|u_1|\leq \eps$ and $|u_2|\leq \frac{1+t}{\sqrt{1-t^2}}\eps =\sqrt{\frac{1+t}{1-t}}\eps \leq 10\eps$.
It follows that $u_3^2\geq 1-101\eps^2$. 
By symmetry, the same bounds hold for the three coordinates of vector $D$. Therefore
\begin{align*}
|\ip{D}{u}|& =|D_1 u_1+D_2 u_2 +D_3 u_3|\geq |D_3| |u_3|-|D_1||u_1| -|D_2| |u_2| \\
           &\geq 1-101\eps^2- \eps^2-(10\eps)^2\\
		   &=1-202\eps^2. \qedhere
\end{align*}
\end{proof}

\begin{proof}[Proof of Lemma~\ref{lem:clause_robust}]
Pick any fixed pair of \textsc{True} node and \textsc{Dummy} node in a literal block (see Fig.~\ref{fig:literal_block})
and denote them by $T$ and $D$ respectively.

Suppose that none of the literal color vectors $\ell_a, \ell_b$ and $\ell_c$ in Fig.~\ref{fig:clause_gadget}
is close to their corresponding \textsc{True} node, i.e. $|\ip{T_i}{\ell_i}|\leq t_0$,
for every $i \in \{a,b,c\}$.

Let us consider nodes $\ell_a$ and $u_a$. Since $\ell_a$ is connected to some $D_a$ node in its literal block (so 
the angle between $\ell_a$ and $D_a$ is $\pi/2-o(1)$), and the angle between $D_a$ and the fixed 
\textsc{Dummy} node $D$ is at most $\delta_p$, it follows that $\theta(\ell_a, D)\geq \pi/2-\delta_p-o(1)$.
On the other hand, since $D$ and $T$ is connected by an edge, they are very close to orthogonal; 
moreover, the angle between $T_a$ and $T$ is at most $\delta_p$. Therefore, 
$\theta(T_a, D)\geq \pi/2-\delta_p-o(1)$.
Since $u_a$ is connected to both $T_a$ and $\ell_a$, we have 
$\theta(T_a, u_a)\geq \pi/2-o(1)$ and $\theta(\ell_a, u_a)\geq \pi/2-o(1)$.
Therefore, the conditions of Claim~\ref{claim:simple_clause} are met for vectors $D$, $T_a$, $u_a$ and $\ell_a$
for $\eps=\sin(\delta_p+o(1))\leq \delta_p+o(1)\leq 1.01\delta_p$ for all sufficiently large BOT graph size $n$.
Hence, by Claim~\ref{claim:simple_clause}, $|\ip{D}{u_a}|\geq 1-202\eps^2 \geq 1-207\delta_p^2$. 
By the same reasoning, we have $|\ip{D}{u_b}|\geq 1-207\delta_p^2$ and $|\ip{D}{u_c}|\geq 1-207\delta_p^2$.
By Claim~\ref{claim:sin_to_angle}, item 3, 
the angle between $D$ and each of the $u_a, u_b$ and $u_c$ is at most $32\delta_p$.

To finish the proof, note that, since each $u_i$ is connected with $v_i$ for every $i \in \{a,b,c\}$,
it follows that the angle between $D$ and each $v_i$ is at least $\pi/2-33\delta_p$ (whenever $n$ is large enough).
Applying Corollary~\ref{cor:equality-gadget-angles} and noting that $v_a$, $v_b$ and $v_c$
form an almost orthonormal basis, and $D$ is close to orthogonal to both $v_a$ and $v_b$ with closeness 
parameter $\theta=33\delta_p$,
we conclude that $\theta(D, v_c)\leq 99\delta_p$. But we also know that
the angle between $D$ and $v_c$ is at least $\pi/2-33\delta_p$, therefore we reach 
a contradiction as long as $\delta_p<\pi/264$.
\end{proof}

\begin{proof}[Proof of Theorem~\ref{thm:soundness}]
By the discussion in Section~\ref{sec:dim3-overview}, without loss of generality, $K$ has diagonal as specified in Theorem~\ref{thm:diagonal}; and if we decompose $K$ as $K=Q^\top Q$, then $q_{(u,v)}$ (which only appears in the examples $\{u,v,(u,v)\}$ by construction) is orthogonal to both $q_u$ and $q_v$. Therefore we may begin by writing the maximum log likelihood of a DPP given by kernel $K$ as
\[
\ell(K)=3\log{m}-\frac{1}{m}\sum_{(u,v)\in E(G_{\phi})}
\left(\log(\deg_{G_{\phi}}(u))+\log(\deg_{G_{\phi}}(v))+\log(\sin^2\theta(u,v))\right).
\]
Let $K$ be a DPP kernel of rank $3$ which satisfies that $\ell(K)\leq \ell_{\textsc{yes}}(G_{\phi})+\frac{C}{\log^2{n}}$, 
where $C$ is some very small absolute constant to be determined later.
Therefore, 
$\EX_{(u,v)\in E(G_{\phi})}[\log(\sin^2\theta(u,v))]\geq -\frac{C}{\log^2{n}}$. 
By Markov's inequality, except for an $\eps''$ fraction of the edges in the BOT graph $G_{\phi}$, 
all edges $(u,v)$ satisfy that 
\[
\log(\sin^2\theta(u,v))\geq -\frac{C/\eps''}{\log^2{n}},
\]
where $\eps''$ is an absolute constant to be specified later. Call an edge in the BOT graph that satisfies this
condition ``good'', and ``bad'' otherwise. 
Consequently, 
\[
\sin^2\theta(u,v) \geq e^{-\frac{C/\eps''}{\log^2{n}}}\geq 1-\frac{C}{\eps''\log^2{n}}. 
\]
By Claim~\ref{claim:sin_to_angle}, the angle between the two vertices $u$ and $v$ of a ``good'' edges $(u,v)$ satisfies that 
$\theta(u,v) \geq \frac{\pi}{2}-\sqrt{\frac{C\pi}{2\eps''}}\frac{1}{\log{n}}$. We remove all ``bad'' edges from $G_{\phi}$ and call the resulting graph $G'_{\phi}$. Note that we remove $\eps''|E(G_{\phi})|$ edges from $G_{\phi}$ and all edges in $G'_{\phi}$ are ``good''.

Next we record some facts about the ``robustness'' of a very strong expander against deletion of a small fraction of its edges. Following Alon and Capalbo, we adopt the following notion of expanders which is slightly weaker than \emph{very strong} expanders.
\begin{definition}[Strong expanders~\cite{AC07}]
A graph $G=(V,E)$ is called a \emph{$d$-strong expander} on $n$ vertices if its minimum vertex degree is at least $d$, 
the average degree in any subgraph of $G$ on at most $n/10$ vertices is at most $2d/9$, and 
the average degree in any subgraph of $G$ on at most $n/2$ vertices is at most $8d/9$.
\end{definition}

\begin{lemma}[\cite{AC07}] \label{lem:Dense}
Let $G=(V,E)$ be a very strong $d$-regular expander on $n$ vertices. Let $E' \subset E$ be an arbitrary subset of edges satisfying $|E'|\leq \frac{nd}{150}$. Consider the process that we first delete $E'$ from $G$, and then repeatedly delete from $G$ the set of vertices with degree smaller than $3d/4+2$ (and all the edges incident to such vertices) as long as such vertices exist. Let $\mathrm{Dense}(G, E')$ be the resulting subgraph of $G$ of this vertex-trimming process, and let $S$ denote the set of vertices deleted from $G$ during this process. Then $|S|\leq \frac{15|E'|}{d}$.
\end{lemma}

Let $\mathrm{Dense}(G, E')$ be as defined in Lemma~\ref{lem:Dense}. It follows that $\mathrm{Dense}(G, E')$ is a $\frac{3d}{4}$-strong expander. The following lemma shows that the diameter of $\mathrm{Dense}(G, E')$ is $O(\log{n})$.
\begin{lemma}[\cite{AC07}] \label{lem:strong_expander_diameter}
For any $d$ and $n$, a $d$-strong expander on $n$ vertices has diameter at most $\frac{2}{3}\log{n}+14$.
\end{lemma}

Let $n$ be the number of variables in the $3$-CNF instance $\phi$. Let $G_{\mathrm{exp}}$ be the underlying very strong expander in our construction of the BOT graph (c.f. Section~\ref{sec:BOT-graph-construction}). Then $G_{\mathrm{exp}}$ has $|V(G_{\mathrm{exp}})|=2nk$ vertices and $|E(G_{\mathrm{exp}})|=nkd$ edges, where $d$ is the degree of the very strong expander. Also recall that the BOT graph has $|V(G_{\phi})|=\Theta(nk)\cdot \max(k,d)$ vertices, $|E(G_{\phi})|=\Theta(nk)\cdot \max(k,d)$ edges and maximum degree $d_{G_{\phi}}\leq 2\max(k,d)+3$.

Recall that we remove $\eps''|E(G_{\phi})|$ ``bad'' edges from $G_{\phi}$ to get $G'_{\phi}$. This process incurs removing edges of the underlying very strong expander $G_{\mathrm{exp}}$: namely, if any of the equality gadget edges between two literal blocks is deleted, we remove the edge between the corresponding two vertices in $G_{\mathrm{exp}}$. Let $E'$ denote the set of edges in $E(G_{\mathrm{exp}})$ that are removed due to the removal of ``bad'' edges from $G_{\phi}$. Then $|E'| \leq \eps''|E(G_{\phi})|$. If $\eps''$ is chosen to satisfy the inequality
\begin{align}\label{eqn:eps-double-prime}
\eps''|E(G_{\phi})| \leq \frac{|V(G_{\mathrm{exp}})|d}{150}=\frac{2knd}{150} \quad\text{or}\quad \eps'' \leq \frac{knd}{75|E(G_{\phi})|},
\end{align}
then the fraction of removed edges from $G_{\mathrm{exp}}$ is small enough so that we can apply the vertex trimming process described in Lemma~\ref{lem:Dense} to $G_{\mathrm{exp}}$, and obtain a $\frac{3d}{4}$-strong expander, which will henceforth denoted by $G'_{\mathrm{exp}}$. That is, $G'_{\mathrm{exp}}:=\mathrm{Dense}(G_{\mathrm{exp}}, E')$. The number of vertices deleted from $G_{\mathrm{exp}}$ is, according to Lemma~\ref{lem:Dense}, at most $\frac{15}{d}\eps''|E(G_{\phi})|$. This vertex trimming process incurs edge removal from $G_{\phi}$ as well: if a vertex, say the $x_{i}^{(j)}$-block, in $G_{\mathrm{exp}}$ is deleted, then we remove all edges incident to any of the three vertices in $G_{\phi}$ (i.e. $T_{i}^{(j)}$, $F_{i}^{(j)}$ and $D_{i}^{(j)}$) that are associated with this vertex in $G_{\mathrm{exp}}$. Let $G''_{\phi}$ be the resulting graph. As the maximum degree of $G_{\phi}$ is bounded by $2\max(k,d)+3$, we remove at most 
\[
\frac{15}{d}\eps''|E(G_{\phi})| \cdot 3 \cdot (2\max(k,d)+3) \leq C_{1} \eps''|E(G_{\phi})|,
\] 
additional edges from $G_{\phi}$, where $C_{1}$ is some constant depends only on $k$ and $d$.

To summarize, starting from the graph $G_{\phi}$, we remove in total at most $\eps''|E(G_{\phi})| + C_{1} \eps''|E(G_{\phi})| \leq (1+ C_{1})\eps''|E(G_{\phi})|$ edges to arrive at the graph $G''_{\phi}$. If we set 
\[
\eps''=\min\left\{\frac{\eps'}{1+ C_{1}}, \frac{knd}{75|E(G_{\phi})|}\right\} 
\]
where $\eps'$ is the constant in Lemma~\ref{lem:BOT-graph}, then $\eps''=\Theta(\eps')$ as both $d$ and $k$ are constants and $|E(G_{\phi})|=\Theta(nk)\cdot \max(k,d)$. 

By choosing $\eps''$ this way, on the one hand, $G''_{\phi}$ is obtained by removing at most $\eps'$-fraction of the edges in $G_{\phi}$. On the other hand, our choice of $\eps''$ ensures that the inequality in~\eqref{eqn:eps-double-prime} is satisfied, hence the vertex trimming process in Lemma~\ref{lem:Dense} can be applied to the very strong expander $G_{\mathrm{exp}}$. Consequently, $G'_{\mathrm{exp}}$ is a $\frac{3d}{4}$-strong expander, and by Lemma~\ref{lem:strong_expander_diameter}, the distance~\footnote{By constructing a path that first reaches the closest $\textsc{True}$, $\textsc{False}$ or $\textsc{Dummy}$ node in a variable block, then uses the shortest path in the strong expander $G'_{\mathrm{exp}}$ between the two corresponding variable blocks; c.f.\ the BOT graph construction in Section~\ref{sec:BOT-graph-construction}.} between any two vertices in $G''_{\phi}$ is at most $\frac{2}{3}\log(2nk)+14+6 \leq C_{2}\log{n}$, where $C_2$ is some absolute constant (e.g.\ we may take $C_{2}=1$ when $n$ is large enough). This latter property of the graph $G''_{\phi}$ makes it possible to consistently $3$-color all the $\textsc{True}$ nodes, $\textsc{False}$ nodes and $\textsc{Dummy}$ nodes in $G''_{\phi}$ as we demonstrate next.

\begin{lemma}\label{lem:decoding_color_nodes}
If the constant $C$ in Theorem~\ref{thm:soundness} is small enough, then $G''_{\phi}$ is $3$-colorable; that is, there is an (efficient) ``decoding'' algorithm which decodes the embedding vectors of all $\textsc{True}$ nodes (resp. $\textsc{False}$ nodes and $\textsc{Dummy}$ nodes) in $G''_{\phi}$ into the same color.
\end{lemma} 
\begin{proof}
The decoding algorithm works as follows. Pick any survived literal block in $G''_{\phi}$ (as $\eps'$ is small enough, such a survived literal block always exists). Set the vector of the $\textsc{True}$ node in the literal block to be $(0,0,1)$. Set the direction in $S^2$ that is closest to the $\textsc{False}$ vector and orthogonal to direction $(0,0,1)$ to be $(1,0,0)$. Since the $\textsc{True}$, $\textsc{False}$ and $\textsc{Dummy}$ nodes in the literal block are interconnected by good edges, the angle between $\textsc{False}$ vector and $(1,0,0)$ together with the angle between $\textsc{Dummy}$ vector and $(0,1,0)$ are both at most $\sqrt{\frac{C\pi}{2\eps''}}\cdot \frac{1}{\log{n}}$.

By the facts that the distance between any pair of $\textsc{True}$ nodes in $G''_{\phi}$ is $O(\log{n})$,
equality gadgets are robust (Lemma~\ref{lem:equality_robust}), and angles on the points in $S^2$ 
are a metric space (in particular, the angle metric satisfies the triangle inequality), 
we may choose a constant $C$ in Theorem~\ref{thm:soundness} small enough so that the angle between any
pair of $\textsc{True}$ vectors is at most $\delta_p:=\pi/300$. Specifically, the constant $C$ needs to satisfy that 
\[
\sqrt{\frac{C\pi}{2\eps''}}\cdot \frac{1}{\log{n}} \cdot C_{2}\log{n} \leq \frac{\pi}{300} \Longleftrightarrow C \leq \frac{2\pi \eps''}{90,000 C_{2}^2}.
\]
Denote\footnote{Recall that, by our convention,
$S_{\mathrm{True}}$ is actually the union of two disconnected regions: $\{v\in S^2: \theta(v, (0,0,1)) \leq \delta_p\} \cup \{v\in S^2: \theta(v, (0,0,-1)) \leq \delta_p\}$;
that is, $S_{\mathrm{True}}$ is a small cap centered around the ``north pole'' plus 
a small cap centered around the ``south pole''.
Similar conventions hold for $S_{\mathrm{False}}$ and $S_{\mathrm{Dummy}}$ as well.} 
\[
S_{\mathrm{True}}=\{v\in S^2: \theta(v, (0,0,1)) \leq \delta_p\}.
\]
Then all $\textsc{True}$ vectors in $G''_{\phi}$ fall into the region $S_{\mathrm{True}}$.
Similarly define $S_{\mathrm{False}}:=\{v\in S^2: \theta(v, (1,0,0)) \leq \delta_p\}$ and 
$S_{\mathrm{Dummy}}:= \{v\in S^2: \theta(v, (0,1,0)) \leq \delta_p\}$. 
Then a similar argument shows that
all $\textsc{False}$ (resp. $\textsc{Dummy}$) vectors in $G''_{\phi}$ 
fall into the region $S_{\mathrm{False}}$ (resp. $S_{\mathrm{Dummy}}$).
Finally, since $S_{\mathrm{True}}$, $S_{\mathrm{False}}$ and $S_{\mathrm{Dummy}}$ are clearly disjoint,
we therefore unambiguously decode all the $\textsc{True}$ nodes, $\textsc{False}$ nodes and $\textsc{Dummy}$ nodes
in $G''_{\phi}$.
\end{proof}
To finish the proof, note that Lemma~\ref{lem:decoding_color_nodes} demonstrates that we can decode a consistent $3$-coloring for all the $\textsc{True}$, $\textsc{False}$ and $\textsc{Dummy}$ nodes. These globally consistent assignments allow us to further ``locally'' color the remaining nodes in $G''_{\phi}$. 

First, for each $1\leq i\leq n$, we decode a consistent truth assignment --- which, for BOT graphs and even BOT subgraphs, implies a valid $3$-coloring --- for all $k$ copies of $x_{i}^{(j)}$ and $\bar{x}_{i}^{(j)}$ as follows. Since the angle between every pair of $\textsc{True}$ node vectors is at most $\delta_p:=\pi/300$, the conditions of Lemma~\ref{lem:clause_robust} are satisfied. Consequently, by Lemma~\ref{lem:clause_robust}, at least one of the three literals in each clause lies within an angular distance of $\pi/15+\pi/300=21\pi/300$ from the direction $(0,0,1)$, ensuring every clause is satisfied. 

Furthermore, because all surviving copies of the same literal are interconnected by an equality gadget, each of these (at most $k$) copies is within $21\pi/300+o(1)<\pi/12$ of $(0,0,1)$. Conversely, since each literal is connected to its negation by a ``good'' edge, its negation is at least $\pi/2-\pi/12-o(1)>\pi/3$ from $(0,0,1)$. We can therefore unambiguously decode the truth assignment for each variable:
\begin{itemize}
\item If all surviving copies of $x_i$ are within $\pi/12$ of $(0,0,1)$, assign $x_i$ to $\textsc{True}$;
\item If all surviving copies of $\bar{x}_i$ are within $\pi/12$ of $(1,0,0)$, assign $x_i$ to $\textsc{False}$;
\item Otherwise, $x_i$ is a free variable and may be assigned a truth value arbitrarily.
\end{itemize}

Next, since every clause is satisfied, it can be easily verified that all auxiliary nodes in the clause gadgets (as shown in Fig.~\ref{fig:clause_gadget}) can be $3$-colored. Finally, as no equality gadget in $G''_{\phi}$ is violated, the two auxiliary nodes in the equality gadgets (as shown in Fig.~\ref{fig:equality_gadget}) can always be properly colored as $\{\textsc{True}, \textsc{Dummy}\}$ or $\{\textsc{False}, \textsc{Dummy}\}$.

In summary, we can properly $3$-color all nodes in $G''_{\phi}$. Because $G''_{\phi}$ is both $\epsilon'$-close to $G_{\phi}$ and $3$-colorable, this completes the proof of Theorem~\ref{thm:soundness}.
\end{proof}

\section{The approximation algorithm}\label{sec:algorithm}
\newcommand{\algname}{\texttt{DiagonalKernel} }
\newcommand{\dom}{N}
Let \algname be the algorithm that, 
given a collection of subsets $X_1,X_2,\ldots,X_m \subseteq [N]$, 
returns the $N\times N$ diagonal matrix $K$ such that 
$K_{ii} = \frac{1}{m} \cdot |\{j : X_j \ni i\}|$ for all $i \in [N]$. 
Clearly, \algname runs in at most $O(mN)$ time. 

\begin{theorem}\label{thm:approxalg}
Let $\ell$ be the log likelihood of the DPP defined by the marginal kernel output by \algname 
on examples $X_1, X_2, \ldots, X_m$, and let $\ell^*$ be the optimal log likelihood achieved by  
a DPP kernel on the same training set.  
Then $\ell \leq \left(1+\frac{1}{\log(\frac{m}{a_{\max}})}\right) \ell^*$, 
where $a_{\max}:=\max_{i \in [N]}|\{j : X_j \ni i\}|$ is the maximum element frequency in examples. 
It follows immediately that $\ell \leq (1 + \frac{1+o(1)}{\log{N}}) \ell^*$, 
when $a_{\max}=O(m)/N$.
As an unconditional weaker bound, we always have $\ell \leq (1+(1+\frac{1}{m-1})\log m)\ell^*$.
\end{theorem}

\begin{proof}
Let $\ell$ and $\ell^{*}$ be the log likelihoods of \algname kernel and optimal kernel, respectively. 
In the following we upper bound the ratio $\frac{\ell}{\ell^*}$.

For each $i \in [\dom]$, define $a_i := |\{j : X_j \ni i\}|$ to be the frequency of element $i$ in the examples.  
Then $K$ is a diagonal matrix with $K_{ii} = \frac{a_i}{m}$ for each $i \in [\dom]$. 
If $\bX$ is a random variable distributed according to this diagonal DPP, 
then $\bX$ is distributed as a product distribution of $N$ independent random variables, 
with $\Pr[i \in \bX] = K_{ii}$ for every $i \in [\dom]$.  
It follows that, for every subset example $X_j \subseteq [N]$,

\begin{align*}
\Pr_{\bX \sim \calP_{K}}[\bX = X_j]& = \prod_{i \in X_j} K_{ii} \prod_{i \notin X_j} (1-K_{ii}) 
 = \prod_{i \in X_j} \frac{a_i}{m} \prod_{i \notin X_j} \left(1-\frac{a_i}{m}\right).
\end{align*}

Since there are exactly $a_i$ sets that contain element $i$ and $m-a_i$ that do not contain element $i$, 
the log likelihood of the marginal DPP kernel $K$ is
\begin{align*}
\ell = -\log\left[\prod_{j=1}^m 
 \left(\prod_{i \in X_j} \frac{a_i}{m} \prod_{i \notin X_j} (1-\frac{a_i}{m})\right) \right]
= -\log\left[\prod_{i=1}^\dom \left(\frac{a_i}{m}\right)^{a_i} \left(1-\frac{a_i}{m}\right)^{m-a_i}\right]. 
\end{align*}

By Lemma~\ref{lem:BMRU_diagonals}, we can assume that the diagonal entries of a kernel $K^*$ 
achieving log likelihood $\ell^*$ match those of the diagonal entries of $K$.  
Let $\bX^*$ be a random variable distributed according to such a marginal kernel, 
then by Lemma~\ref{lem:det-leq-prod-main-diagonal} (Hadamard's Inequality) we have

\begin{align*}
\Pr_{\bX^* \sim \calP_{K^*}}[\bX^* = X_j] & \leq \Pr_{\bX^* \sim \calP_{K^*}}[X_j \subseteq \bX^*] 
= \det(K^*_{X_j}) \\ 
& \leq \prod_{i \in X_j} K^*_{ii} 
= \prod_{i \in X_j} K_{ii} = \prod_{i \in X_j} \frac{a_i}{m}.
\end{align*}
It follows that
\[
\ell^* \geq -\log\left(\prod_{i=1}^\dom \left(\frac{a_i}{m}\right)^{a_i}\right).
\]
We now observe that
\allowdisplaybreaks
\begin{align*}
\frac{\ell}{\ell^*} 
& \leq \frac{\log\left(\prod_{i=1}^\dom \left(\frac{a_i}{m}\right)^{a_i} \left(1-\frac{a_i}{m}\right)^{m-a_i}\right)}
{\log\left(\prod_{i=1}^\dom \left(\frac{a_i}{m}\right)^{a_i}\right)}
= 1 + \dfrac{\log (\prod_{i=1}^\dom \left(1-\frac{a_i}{m})^{m-a_i}\right)}
{\log\left(\prod_{i=1}^N \left(\frac{a_i}{m}\right)^{a_i}\right)} \\
&= 1+\dfrac{\log \left(\prod_{i=1}^\dom (1-\frac{a_i}{m})^{1-\frac{a_i}{m}}\right)}
  {\log\left(\prod_{i=1}^N \left(\frac{a_i}{m}\right)^{\frac{a_i}{m}}\right)} \\
&= 1+\dfrac{\sum_{i=1}^\dom \log \left((1-\frac{a_i}{m})^{1-\frac{a_i}{m}}\right)}
    {\sum_{i=1}^\dom \log\left( \left(\frac{a_i}{m}\right)^{\frac{a_i}{m}}\right)}. 
\end{align*}
For $x\in (0,1)$, let $g(x):=-\log \left((1-x)^{1-x}\right)$ and $h(x):=-\log\left( x^x \right)$, and define
\[
f(x) := \frac{g(x)}{h(x)}=\dfrac{\log \left((1-x)^{1-x}\right)}{\log\left( x^x \right)} 
= \dfrac{(1-x) \log (1-x)}{x \log x}.
\]
Then it is easy to check that, for every $0<x<1$, both $g(x)$ and $h(x)$ and hence $f(x)$ are positive;
moreover, $f(x)$ is an increasing function in $x$ . Therefore,
\begin{align}\label{ineq:algo-ratio}
\frac{\ell}{\ell^*} 
&\leq 1+\dfrac{\sum_{i=1}^\dom g(\frac{a_i}{m})}{\sum_{i=1}^\dom h(\frac{a_i}{m})} 
   \leq 1+\max_{i \in [\dom]} \dfrac{g(\frac{a_i}{m})}{h(\frac{a_i}{m})} \nonumber\\
&= 1+\max_{i \in [\dom]} \dfrac{\log \left((1-\frac{a_i}{m})^{1-\frac{a_i}{m}}\right)}
  {\log\left( \left(\frac{a_i}{m}\right)^{\frac{a_i}{m}}\right)} \nonumber\\
&= 1+ f(\frac{a_{\max}}{m}).
\end{align} 
Using the inequality that $(1+x)\log{(1+x)}\geq x$ for all $x>-1$, we 
get that $f(x) \leq - \frac{1}{\log x}$ for all $0 < x < 1$.  
Thus, when $\frac{a_{\max}}{m} \leq \frac{C}{\dom}$, 
the kernel output by \texttt{DiagonalKernel} satisfies 
$\ell \leq (1 - \frac{1}{\log (C/\dom)}) \ell^* = (1 + \frac{1}{\log (\dom/C)}) \ell^*$. 

When the condition $\frac{a_{\max}}{m} \leq \frac{C}{\dom}$ is not satisfied, 
in order to obtain an unconditional upper bound on the log likelihood of \algname, 
observe that without loss of generality, we may consider only the cases when $a_{\max}\leq m-1$. 
This is because elements occur in all the training subsets will 
be assigned probability~$1$ in a maximum likelihood DPP, hence have $1$ at the diagonal entry 
and $0$'s at all other off-diagonal entries. This means that such elements can be ``factored out'' from the 
marginal kernel. Thus, plugging $\frac{a_{\max}}{m}=\frac{m-1}{m}$ into \eqref{ineq:algo-ratio}, we obtain
\[
\frac{\ell}{\ell^*}\leq 1+\frac{\frac{1}{m}\log\frac{1}{m}}{(1-\frac{1}{m})\log(1-\frac{1}{m})}
=1+\frac{\log{m}}{(m-1)\log(1+\frac{1}{m-1})} \leq 1+(1+\frac{1}{m-1})\log{m},
\]
where in the last step we use again the inequality
$(1+x)\log(1+x)\geq x$ for all $x>-1$ (and setting $x=\frac{1}{m-1}$).
\end{proof}
Thus this simple algorithm does surprisingly well, unless there are some high frequent elements that appear 
in a non-trivial fraction of the training data, which is often not the case in practice.
Further, it's worth noting that the training set constructed in our hardness of learning proof
also satisfies that $\frac{a_{\max}}{m} \leq \frac{C}{\dom}$ for some constant $C$.

\section{Discussion and open problems}\label{sec:discussion}

In this work, we establish that it is NP-hard to obtain a $1-O(\frac{1}{\log^9 N})$-approximation 
to the maximum log likelihood of DPPs. 
We also demonstrate a simple polynomial-time algorithm that achieves $\frac{1}{(1+o(1))\log{m}}$-approximation. 
One immediate open problem is to close this large gap. 
A natural and plausible approach is to prove the cardinality-rank conjecture or at least
to improve the bound in Theorem~\ref{thm:dim3}.
Note that our hardness result does not rule out efficient learning with some constant factor of approximation: 
it is still possible that there is a polynomial-time algorithm that obtains a DPP kernel with a 99\%-approximation 
to the maximum log likelihood. As observed earlier, we cannot preclude constant-factor approximations 
by a better analysis of the constructed hard instance: the approximation algorithm shows that
our hardness result is tight up to a polynomial factor for the type of subset collections
employed in the proof. 
Therefore any stronger hardness proof would require constructing a collection of subsets 
in which some element appears in a non-trivial fraction of the subsets.

Our investigation just takes a first stab at understanding the computational landscape of learning DPPs.
In particular, our knowledge for the complexity of learning DPPs when the data set is indeed generated
by an unknown DPP is still very limited:
Can one design efficient algorithms for such a task? 
Can the DPP kernel be learned with arbitrary accuracy? And if not, what is the best approximation factor
can such an algorithm achieve? Note that the data here is no longer a worst-case {\em data set}, 
but only sampled from a worst-case {\em DPP}. The underlying model is thus a semi-random one, 
and it seems challenging to extend NP-completeness hardness to such settings; 
some kind of {\em average-case hardness} is likely to be the best one can hope for. 
This is essentially the approach that \cite{BMRU17} had envisioned when examining 
the optimization landscape of the likelihood function for DPP kernels. 
The convergence of the empirical log-likelihood function to the true log-likelihood function 
only holds with high probability, and so in particular doesn't carry over to the kinds of worst-case data sets 
produced by our reduction. Thus, their conjectured property may still hold, 
and may be a route to an efficient algorithm in this setting. 
On the other hand, ``realizability'' is probably too strong of an assumption for practical purposes; 
DPPs are generally used to model processes featuring negative association, 
and it often seems implausible that the data actually follows a DPP distribution. 
Therefore, finding more appropriate assumptions is yet to be explored, and an algorithm for a realizable setting 
would be a natural first step along this direction.

On the other side of the coin, it is entirely conceivable that such efficient algorithms may not exist at all.
One may view our main result as proving the hardness of ``agnostic-learning'' DPPs, while here the
task would be proving hardness of ``PAC-learning'' DPPs. 
Presumably this is more difficult, as PAC-learning is in general easier than
agnostic learning, it is thus harder to obtain lower bounds in the former setting.
In particular, the usual approach of proving PAC-learning lower bounds involves 
uniform distributions over some prescribed collections of subsets.
Such distributions are within the scope of PAC-learning model as it allows arbitrary distributions. 
By contrast, DPPs are normally unable to generate the uniform distribution over an arbitrary collection of subsets. 
Indeed, we believe that the data set we construct would be atypical for all DPPs. 
This is why contrary to the usual representation-specific hardness theorems in PAC-learning, 
we believe that an average-case hardness assumption will be necessary here.

\section*{Acknowledgements}
We are indebted to the anonymous reviewers for carefully reading the manuscript. Their keen eyes caught numerous errors and omissions, and their insightful comments and suggestions greatly improved the clarity of presentation and overall quality of this work. B.J. was partially supported by NSF awards IIS-1908287, IIS-1939677, and CCF-1718380. E.G. was partially supported by NSF CCF-1910659 and NSF CCF-1910411.
N.X. was partially supported by U.S. Army Research Office (ARO) under award number W911NF1910362.

\bibliographystyle{plain}

\end{document}